\documentclass[10pt,english]{article} 
\usepackage{amsmath, amsfonts, amsthm, amssymb, amscd}
\usepackage[mathcal]{euscript}
\usepackage{mathrsfs}
\usepackage{a4wide}
\usepackage{graphicx}
\usepackage{ulem}  
\usepackage{subfig}
\usepackage{mathrsfs}
\usepackage{rotating}
\usepackage{verbatim}
\usepackage{slashed}
\usepackage{tensor}    
\usepackage{dsfont} 
\usepackage{extarrows} 
\usepackage{lineno}
\usepackage {bbm}
 
\usepackage{xr} 
\DeclareFontFamily{U}{BOONDOX-calo}{\skewchar\font=45 }
\DeclareFontShape{U}{BOONDOX-calo}{m}{n}{
<-> s*[1.05] BOONDOX-r-calo}{}
\DeclareFontShape{U}{BOONDOX-calo}{b}{n}{
<-> s*[1.05] BOONDOX-b-calo}{}
\DeclareMathAlphabet{\mathcalboondox}{U}{BOONDOX-calo}{m}{n}
\SetMathAlphabet{\mathcalboondox}{bold}{U}{BOONDOX-calo}{b}{n}
\DeclareMathAlphabet{\mathbcalboondox}{U}{BOONDOX-calo}{b}{n}

\usepackage{perpage}
\MakePerPage{footnote} 
\usepackage[most]{tcolorbox} 
\newtcolorbox{done}{%
enhanced, breakable, size=minimal, parbox=false, after={\par}, 
before upper={\indent}, colback=white, 
overlay = {\draw[line width=2pt, blue] (frame.north east) -|
([xshift=3mm]frame.east)|-(frame.south east);},
overlay first={\draw[line width=2pt, blue] (frame.north east) -|
([xshift=3mm]frame.south east);},
overlay middle={\draw[line width=2pt, blue] ([xshift=3mm]frame.north east) -- 
([xshift=3mm]frame.south east);},
overlay last={\draw[line width=2pt, blue] ([xshift=3mm]frame.north east)|-
(frame.south east);},
} %
\newtcolorbox{plfok}{%
enhanced, breakable, size=minimal, parbox=false, after={\par}, 
before upper={\indent}, colback=white, 
overlay = {\draw[line width=2pt, blue, dashed] (frame.north east) -|
([xshift=3mm]frame.east)|-(frame.south east);},
overlay first={\draw[line width=2pt, blue, dashed] (frame.north east) -|
([xshift=3mm]frame.south east);},
overlay middle={\draw[line width=2pt, blue, dashed] ([xshift=3mm]frame.north east) -- 
([xshift=3mm]frame.south east);},
overlay last={\draw[line width=2pt, blue,dashed] ([xshift=3mm]frame.north east)|-
(frame.south east);},
} 
\usepackage{graphicx}
\makeatletter
\def\widebreve#1{\mathop{\vbox{\m@th\ialign{##\crcr\noalign{\kern\p@}%
\brevefill\crcr\noalign{\kern0.1\p@\nointerlineskip}%
$\hfil\displaystyle{#1} \hfil$\crcr}}} \limits}
\def\brevefill{$\m@th \setbox \z@\hbox{}%
\hfill\scalebox{0.6}{\rotatebox[origin=c]{90}{(}} \kern1pt $}
\makeatletter 
\newtheorem{theorem}{Theorem}[section]
\newtheorem{proposition}[theorem]{Proposition}
\newtheorem{lemma}[theorem]{Lemma}

\newtheorem{corollary}[theorem]{Corollary}

\numberwithin{equation}{section}
\newcommand \epss {\eps_\star}

\newcommand \crochet {\mathbf X}

\newcommand \Fenergy {\mathcalboondox F}           
\newcommand \Eenergy {\mathcalboondox E}    
\newcommand \Time {\mathbf T}

\newcommand \E {\mathcal{E}} 

\newcommand \ME {\mathcal{EM}} 
\renewcommand \H {\mathcal H} 

\newcommand \N {\mathcal N} 

\newcommand \Lcal {\mathcal L} 

\newcommand \Mscr       {\mathscr M} 
\newcommand \Mext       {\Mscr^{\mathcal E}} 

\newcommand \MH 		{\Mscr^{\mathcal{H}}} 

\newcommand \MME 		{\Mscr^{\mathcal{EM}}}  

\newcommand \MM 		{\Mscr^{\mathcal{M}}} 
\newcommand \Mtran      {\Mscr^{\mathcal{M}}}  

\newcommand \MMEfar 	{\Mscr^{\textbf{far}}} 
\newcommand \MMEnear  	{\Mscr^{\textbf{near}}} 
 
\newcommand \Mnear {\Mscr^\textbf{near}}
\newcommand \Mfar{\Mscr^\textbf{far}}

\newcommand \near {\textbf{near}}
\newcommand \far {\textbf{far}}

\newcommand \delH {\del^{\mathcal H}} 
\newcommand \delsH {\slashed \del^{\mathcal H}}

\newcommand \delN {\del^{\mathcal{N}}}
\newcommand \delsN {\slashed \del^{\mathcal{N}}}
\newcommand \PsiN {{\Psi^{\mathcal N}{}}}  
\newcommand \PhiN {{\Phi^{\mathcal N}}{}}

\newcommand \HN {H^{\mathcal{N}}{}}
\newcommand \hN {h^{\mathcal{N}}}

\newcommand \delsE {\slashed \del^\E}

\newcommand \delsEH {\slashed {\del}^{\mathcal{EH}}}
\newcommand \delsME {\slashed {\del}^{\mathcal{EM}}} 

\newcommand \rhoH {{\mathbf r}^{\mathcal{H}}}    
\newcommand \rhoE {{\mathbf r}^{\mathcal{E}}}                

\newcommand {\hs}{\slashed{h}}

\newcommand {\us}{\slashed{u}}

\newcommand {\Abb}{\mathbb{A}}

\newcommand {\Cbb}{\mathbb C}

\renewcommand {\Bbb}{\mathbb{B}}

\newcommand \Fbb {\mathbb F}

\newcommand \Pbb {\mathbb P}

\newcommand \Qbb {\mathbb Q}

\newcommand \SbbME{\mathbb S^{\mathcal{EM}}}

\newcommand \Ibb	{\mathbb I}  	

\newcommand \wR {\tensor[^{(w)}]{}{}  R}

\newcommand \Ebf {\mathbf E}

\newcommand \la \langle
\newcommand \ra \rangle

\newcommand \Div {\text{Div}}
\newcommand \init {\textbf{init}}

\newcommand \Tr {\text{Tr}}

\newcommand \dive {\text{div} \hskip.05cm}

\newcommand \del \partial 

 \newcommand \delu \delH 
 
\newcommand \delEH {\del^{\mathcal{E \hskip-.06cm H}}} 

\newcommand \delus  \delsH

\newcommand \delts  \delsN

\def\blfootnote{\gdef\@thefnmark{} \@footnotetext}
\newcommand \Rwave {\tensor[^{(w)} ]{R}{^\star}}

\newcommand \vep \epsilon

\newcommand \Sch {S}

\newcommand \source{\textbf{sour}}

\newcommand \Boxt {\widetilde \Box}

\newcommand \Lscr{\mathscr{L}}

\newcommand \BoxChapeau {\widehat \Box} 
\newcommand \Hcal {\mathcal H}

\newcommand \RR{\mathbb{R}}

\newcommand \Abf	{\mathbf{A}}
\newcommand \Bbf	{\mathbf{B}}

\newcommand {\eps} \epsilon

\let\oldmarginpar\marginpar
\renewcommand\marginpar[1]{\- \oldmarginpar[\raggedleft\footnotesize #1]%
{\raggedright\footnotesize #1}}

\newcommand \TJ T

\newcommand \gd g  
 
\newcommand \Rd  R 
\newcommand \Rbd R 
 
\newcommand \omegad  \omega 
\newcommand \nablabd \nabla  
\newcommand \Gammad \Gamma 
\newcommand \Gammab \Gamma 
\newcommand \Deltab \Delta 

\newcommand \coef \kappa

\newcommand \Ocal {\mathcal O}

\renewcommand \BoxChapeau {\widetilde \Box}

\newcommand \lbf {\mathbf l}

\newcommand \ord {\textbf{ord}}

\renewcommand \deg {\textbf{deg}}

\newcommand \Mink {\textbf{Mink}}
\newcommand \rank {\textbf{rank}}
\newcommand \gMink {g_\textbf{Mink}}
\renewcommand \Sch {\textbf{Sch}}

\newcommand \pertur {\textbf{pertur}}

\newcommand \Tbb {\mathbb T}

\newcommand \Mgood{\Mscr^{\textbf{good}}}
\newcommand \Mbad{\Mscr^{\textbf{bad}}}

\newcommand \glue {{}}                 

\newcommand \err {\textbf{err}}

\newcommand \easy {\textbf{easy}}
 
\newcommand \super {\textbf{super}} 

\newcommand \LOmega {Y_\text{rot}} 

\newcommand \exposant \rho
\newcommand \vecnnu \nu
\newcommand \notrelapse L
 
\newcommand{\smallbullet}{{\scriptscriptstyle\mspace{.5mu}\bullet\mspace{.5mu}}}

\let\oldaligned\aligned
\def\aligned{\oldaligned\relax}
\newcommand \be {\begin{equation}}
\newcommand \ee {\end{equation}}
\newcommand \bei {\begin{itemize}}
\newcommand \eei {\end{itemize}} 

\begin{document}

\title{\bf Einstein-Klein-Gordon spacetimes
\\
in the harmonic near-Minkowski regime}  

\author{\large Philippe G. LeFloch\footnote{Laboratoire Jacques-Louis Lions and Centre National de la Recherche Scientifique,
Sorbonne Universit\'e, 
4 Place Jussieu, 75252 Paris, France. Email: {\sl contact@philippelefloch.org}.
\newline $^\dag$ School of Mathematics and Statistics, Xi'an Jiaotong University, Xi'an, 710049 Shaanxi, People's Republic of China.
Email: {\sl yuemath@mail.xjtu.edu.cn}.
\hfill Version: August 2022. 
} 
\, and Yue Ma$^\dag$} 

\date{}

\maketitle  


\begin{abstract}   
We study the initial value problem for the Einstein-Klein-Gordon system and establish the global nonlinear stability of massive matter in the near-Minkowski regime when the initial geometry is a perturbation of an asymptotically flat, spacelike hypersurface in Minkowski spacetime and the metric enjoys the harmonic decay $1/r$ (in term of a suitable  distance function $r$ at spatial infinity).  
Our analysis encompasses matter fields that have small energy norm and solely enjoys a slow decay at spacelike infinity. Our proof is based on the Euclidean-hyperboloidal foliation method recently introduced by the authors, and distinguishes between the decay along asymptotically hyperbolic slices and the decay along asymptotically Euclidean slices. We carefully analyze the decay of metric component at the harmonic level $1/r$, especially the metric component in the direction of the light cone. In presence of such a slow-decaying matter field, we establish a global existence theory for the Einstein equations expressed as a coupled system of nonlinear wave and Klein-Gordon equations.  
\end{abstract}
 
  
{
\small 

\setcounter{secnumdepth}{2}
\setcounter{tocdepth}{1} 
\tableofcontents
} 


\section{Introduction} 

\subsection{Global evolution problem for the Einstein-matter equations}

\paragraph{Main purpose.}

We consider solutions to the Einstein equations in the near-Minkowski regime and study the global nonlinear evolution problem when a suitable initial data set is prescribed, consisting of a Riemannian  metric, a symmetric two-tensor, and matter data. These data are assumed to be close to a  sufficiently flat and asymptotically Schwarzschild-like perturbation of a hypersurface in Minkowski spacetime. For the {\sl vacuum} Einstein equations, this problem was solved first by Christodoulou and Klainerman \cite{CK}, while an alternative proof in wave coordinates was given later on by Lindblad and Rodnianski \cite{LR1,LR2}. Solutions with lower decay at spacelike infinity were constructed by Bieri~\cite{Bieri} (and \cite{BieriZipser}), while most recent contributions are due to Hintz and Vasy \cite{HintzVasy1,HintzVasy2}. All of these results easily extend to {\sl massless} matter fields. 

Our aim in the present paper is to solve the global evolution problem for the Einstein-matter system when the matter field under consideration is {\sl massive.} We are going to describe the problem and first state a simplified version of our main result (Theorem~\ref{theo-main-result}) when the data have Schwarzschild decay. Later in this paper (cf.~Theorem~\ref{theo-main-result-qualitative}) we state a more general result, when the metric at spatial infinity enjoys harmonic-type decay. 

A recent literature on the global dynamics of self-gravitating massive fields is now available. The present paper is a companion to our work~\cite{PLF-YM-main} and, on the one hand, provides the arguments that are necessary in order to establish the Schwarzschild-type decay of solutions to the Einstein equations and, on the other hand, can also be considered as an overview of the more general proof in~\cite{PLF-YM-main} (see also \cite{PLF-YM-long}). These results grew from earlier work by the authors in \cite{PLF-YM-book,PLF-YM-two}.

While our project came under completion we learned that Ionescu and Pausader simultaneously solved the same problem by introducing a completely different methodology, which is based on the notion of spacetime resonances; see \cite{IP3}. A somewhat different class of initial data sets is covered therein, as far as the functional norms and the spatial decay of solutions are concerned (specially since the regularity in \cite{IP3} is stated in weighted Fourier spaces). We will not attempt to review here the vast literature on the subject, and we refer the reader to our detailed overview in \cite{PLF-YM-main}. Recent work includes, among others, the major contributions by Bigorgne~\cite{Bigorgne,Bigorgne2}, Fajman et al. \cite{FJS,FJS3}, Lindblad et al. \cite{LTay}, Smulevici \cite{Smulevici} and Wang \cite{Wang}. 

 
\paragraph{Einstein equations and massive fields.}

We consider four-dimensional manifolds $(\Mscr, g)$ in which $\Mscr \simeq [0, + \infty) \times \RR^3$ and $g$ is a Lorentzian metric with signature $(-, +, +, +)$ and Levi-Civita connection $\nabla_g= \nabla$. Greek indices $\alpha,\beta, \ldots$ describe $0, 1,2,3$, and we denote by $R_{\alpha\beta}$ the Ricci curvature of $g$ and by $R = g^{\alpha\beta} R_{\alpha\beta}$ its scalar curvature. Throughout, we use implicit summation over repeated indices, as well as raising and lowering indices with respect to the metric $g_{\alpha\beta}$ and its inverse denoted by $g^{\alpha\beta}$. 

We impose that Einstein-matter equations hold, namely 
\begin{equation}\label{eq 1 einstein-massif}
G_{\alpha\beta} =  8\pi \, T_{\alpha\beta}, 
\end{equation}
in which the left-hand side is Einstein's curvature tensor defined (in abstract indices) as 
\begin{equation}
G_{\alpha\beta} = R_{\alpha\beta} - {R \over 2} \, g_{\alpha\beta}. 
\end{equation}
The right-hand side of \eqref{eq 1 einstein-massif} is the energy-momentum tensor, which depends upon the matter content and, for instance, is taken to vanish identically for vacuum Einstein spacetimes, so that \eqref{eq 1 einstein-massif} then reduces to the Ricci-flat condition $R_{\alpha\beta}=0$. 


While our method of analysis should apply to many massive matter models, we are interested here in Klein-Gordon scalar fields $\phi: \Mscr \to \RR$ with energy-momentum tensor 
\begin{equation}\label{eq:Talphabeta}
T_{\alpha\beta} = \nabla_\alpha \phi \nabla_\beta \phi - \Big( {1 \over 2} \nabla_\gamma \phi \nabla^\gamma \phi + U(\phi) \Big) g_{\alpha\beta}. 
\end{equation}
Here, the potential function $U=U(\phi)$ depends on the nature of the matter under consideration and, throughout, we assume that
\begin{equation}\label{eq:Uofphi}
U(\phi) = {1\over 2} c^2  \phi^2 + \Ocal(\phi^3)
\end{equation}
for some (mass-like) constant $c>0$. 
Recalling the Bianchi identity $\nabla^\alpha R_{\alpha\beta} = {1 \over 2} \nabla_\beta R$, we have $\nabla^\alpha G_{\alpha\beta} = 0$ and, consequently, 
\begin{equation}\label{Eq1-15bis}
\nabla^\alpha T_{\alpha\beta} =0.
\end{equation}
In turn, we deduce that $\phi$ satisfies the nonlinear Klein-Gordon equation
\begin{equation}\label{eq-KGG}
\Box_g \phi - U'(\phi) = 0. 
\end{equation}
For instance, when $U(\phi)={c^2 \over 2} \phi^2$,  this is $\Box_g \phi - c^2 \phi = 0$, which is linear in $\phi$ but also involves the unknown metric. 


\paragraph{Initial value problem.}

For suitable initial data, the equation \eqref{eq-KGG} is expected to uniquely determine the evolution of the matter.
Our challenge is precisely to study the nonlinear coupling problem when the metric $g$ itself is one of the unknowns and solves Einstein equations with suitably prescribed initial data. 
When a foliation by spacelike hypersurfaces is chosen and a suitable gauge choice is made,  we use Latin indices varying between $1$ and $3$ and the Einstein equations decompose into 
and evolution equations 
\begin{equation}\label{equa-Einstein} 
G_{ab} = 8\pi \, T_{ab}, \qquad \quad a,b=1,2,3. 
\end{equation}
and constraint equations 
\begin{equation}\label{equa-the-constraints} 
G_{00} = 8\pi \, T_{00}, \qquad \quad
G_{0a} = 8\pi \, T_{0a}, \qquad \quad a,b=1,2,3. 
 \end{equation} 
The initial value problem is then formulated by specifying a Riemannian metric $g_0$ and a symmetric two-tensor field $k_0$, defined on a manifold with topology $\RR^{3}$, and satisfying the constraints \eqref{equa-the-constraints}, namely 
\begin{equation}
\label{eq:ee11}
\aligned
R_{g_0} + (\Tr_{g_0} k_0)^2 -  | k_0 |_{g_0}^2 & = 16 \pi T_{00},   
\qquad
\Div_{g_0} \big ( k_0 - (\Tr_{g_0} k_0) g_0 \big) = -8 \pi T_{0 \smallbullet}.   
\endaligned
\end{equation}
where $R_{g_0}$ denotes the scalar curvature, $\Tr_{g_0}$ the trace, and $\Div_{g_0}$ the divergence of the metric $g_0$. 
We then seek a globally hyperbolic Cauchy development~\cite{CBG,YCB} consisting of a spacetime metric $g$ together with a matter field $\phi$ satisfying Einstein's evolution equations and defining a manifold whose induced geometry (first- and second-fundamental forms) on an initial hypersurface is given by the pair $(g_0, k_0)$. 
  

\subsection{Self-gravitating fields in the asymptotically Schwarzschild regime}

\paragraph{Wave-Klein-Gordon system of interest.}

We work on the manifold $\Mscr \simeq \RR^{1+3}_+ := \{(t,x)\in\RR^{1+3},t\geq 1\}$ covered by a single coordinate chart $(t,x)= (t,x^a)$ with $t \geq 1$ and $x \in \RR^3$. We also write $r^2 := |x|^2 = \sum_{a=1}^3|x^a|^2$. We introduce the Minkowski metric $g_\Mink := -dt^2 + \sum_a(dx^a)^2$, and observe that $\Mscr$ is the future of the initial hypersurface $\{t=1\}$. We also introduce the  \textsl{outgoing light cone} 
\begin{equation}\label{equa-lightcone}  
\qquad 
\Lscr:= \big\{ r = t-1 \big\} \subset \RR_+^{3+1}
\end{equation} 
and its constant-$t$ slices are denoted by $\Lscr_t$.

More precisely, we rely on global coordinate functions $x^\alpha: \Mscr \to \RR$ (with $\alpha=0,1,2,3$) satisfying 
the wave gauge conditions 
\begin{equation}\label{eq:wcooE}
\Box_g x^\alpha = 0, 
\end{equation}
and we express the Einstein equations as a nonlinear system of second-order partial differential equations, supplemented with second-order constraints. The unknowns are the metric coefficients $g_{\alpha\beta}$ and the matter field $\phi$. Specifically, we find 
\begin{equation}\label{MainPDE-limit}
\aligned
\BoxChapeau_g g_{\alpha\beta} = \Fbb_{\alpha\beta}(g, g;\del g,\del g) 
-16\pi \, \big( \del_{\alpha}\phi\del_{\beta}\phi + U(\phi)g_{\alpha\beta} \big),
\qquad
\BoxChapeau_g \phi  - U'(\phi) = \, 0, 
\endaligned
\end{equation}
where  
$\BoxChapeau_{\gd} := \gd^{\alpha'\beta'} \del_{\alpha'} \del_{\beta'}$ denotes a modified wave operator, and the wave gauge constraints take the form  
\begin{equation}\label{eq:gamnul3}
\aligned
\Gamma^\alpha & =   g^{\alpha\beta} \Gamma_{\alpha\beta}^\lambda = 0, 
\qquad
\Gamma_{\alpha \beta}^{\lambda}
 = {1 \over 2} \,  g^{\lambda \lambda'} \big(\del_\alpha g_{\beta \lambda'}
+ \del_\beta g_{\alpha \lambda'} - \del_{\lambda'} g_{\alpha \beta} \big). 
\endaligned
\end{equation} 
The nonlinearities $\Fbb = \Pbb+\Qbb$ have an important structure especially in connection with the Euclidian-hyperboloidal foliation, which is described and play a central role in~\cite{PLF-YM-main}. 
 

\paragraph{Merging the Minkowski and the Schwarzschild solutions.}

In wave coordinates the Schwarz\-schild metric $g_{\Sch}$ takes the form (with $\omega_a := x_a/r$) 
\begin{equation}\label{eq Sch-wave-new}
\aligned
g_{\Sch,00} & =   - \frac{r-m}{r+m},
\qquad g_{\Sch,0a} = 0,
\qquad 
g_{\Sch, ab} = \frac{r+m}{r-m} \omega_a \omega_b + \frac{(r+m)^2}{r^2}(\delta_{ab} - \omega_a \omega_b). 
\endaligned
\end{equation}
Consider a (regular) cut-off function $\chi^\star(r)$ that vanishes for $r\leq 1/2$ and is identically $1$ for all $r\geq 3/4$. Given a (small) mass coefficient $m>0$ we introduce the reference metric 
\begin{equation}\label{equa-defineMS-new} 
g_\glue^\star = \gMink + \chi^\star (r) \, \chi^\star(r/(t-1)) (g_\Sch - \gMink), 
\qquad 
t \geq 2. 
\end{equation} 
This metric coincides with $\gMink$ in the cone $\big\{ r/(t-1)< 1/2 \big\}$ and coincides with $g_\Sch$ in the exterior $\big\{ r/(t-1)\geq 3/4 \big\}$ which contains the light cone $\Lscr$.  This metric satisfies the so-called light-bending property, saying by definition that the light cone coefficient 
\begin{equation}\label{equa-Sch-bending-new} 
r \, g_{\glue}^{\star}(\lbf,\lbf) = 4m + \Ocal(1/r) \qquad \text{ for the metric } g_\glue^\star, 
\end{equation}
is positive, where the {\sl light cone direction} is defined by 
\begin{equation}
\lbf := \del_t - (x^a/r)\del_a.
\end{equation}  
This is essentially the same construction as in~\cite{LR1}. 


\paragraph{Class of initial data sets.}

We are interested in solutions to the Einstein-matter system in the harmonic regime (corresponding to $\lambda=1$ in~\cite{PLF-YM-main}) and we establish a sharp decay property in the case where the reference metric is constructed from the Schwarzschild metric. The initial metric $g_0$ is assumed to be sufficiently close to the Euclidean metric while the initial second fundamental form $k_0$ is sufficiently small. Let us introduce the following decomposition (and $a, b=1,2,3$) 
\begin{equation}\label{equa-decomp-data} 
g_{0ab} = g^\star_{0ab} + u_{0ab} = \delta_{ab} + h^\star_{0ab} + u_{0ab}, 
\qquad \quad
k_{0 ab} = k^{\star}_{0ab} + l_{0ab}. 
\end{equation}
We aim at covering a variety of asymptotic behaviors and, at this juncture, it is convenient to introduce the following terminology. 

\begin{itemize}

\item The part $h^\star_{0}$ is referred to as the {\bf initial reference} and will be assumed to be small in a (weighted, high-order) pointwise norm. 

\item The part $u_{0}$ is referred to as the {\bf initial perturbation} and will be assumed to be small in  (weighted, high-order) energy norm. 

\end{itemize} 

An example of a such decomposition is provided by the construction in Lindblad and Rodnianski~\cite{LR1}, where the initial data is decomposed as the sum of a finite-energy perturbation plus an (asymptotically) Schwarzschild metric outside of a compact set (with sufficiently small and positive mass). In our theory, the two parts are treated differently. Indeed, $h_{0}^{\star}$ is the initial trace of $h^\star$ while $u_0$ is propagated.  

For the metric perturbation, we introduce the energy norms 
\begin{equation}\label{equa-norms} 
\aligned
\Ebf^{\text{metric}}_{\kappa,N} (g_0, k_0)
= &
\sum_{|I|\leq N}\big\|\la r\ra^{\kappa + |I|} \big( |\del_x^I\del_x u_0| + |\del_x^I u_1| \big) \big\|_{L^2(\RR^3)},
\\
\Ebf^{\text{matter}}_{\mu,N} (\phi_0, \phi_1) 
:= &
\sum_{|I|\leq N}\big\|\la r\ra^{\mu + |I|} (|\del_x^I\del_x\phi_0| + |\del_x^I\phi_0| + |\del_x^I\phi_1|)\big\|_{L^2(\RR^3)}.
\endaligned
\end{equation}
While our proof below will provide a more general result as stated below, we find it convenient to state first our result for the perturbation of the Schwarzschild solution. 

 
Given an initial data set we decompose it according to \eqref{equa-decomp-data} and we introduce the {\sl linear development} denoted by $u_\init$ of the initial data set $(u_{0 \alpha\beta},u_{1 \alpha\beta})$, that is, we introduce the solution to the (free, linear) wave equation with this initial data. It was established in~\cite[Proposition~\ref{prop1-12-01-2022}]{PLF-YM-main} that, under the assumption that the norm above is finite, 
\begin{equation}\label{eq3-09-05-2021-new}
|u_{\init}|  \lesssim C_0\eps (t+r+1)^{-1}. 
\end{equation} 
Our main assumption beyond the smallness on the norms~\eqref{equa-norms} is the following sign condition which we referred to as the {\sl light-bending condition}
\begin{equation}\label{eq3'-27-05-2020-initial-new-Sch}
\inf_{\Mscr^\near_\ell} \big(4 m + r \, u_\init(\lbf,\lbf) \big) 
\geq \epss. 
\end{equation}
Here, a parameter $\ell \in (0,1/2]$ being fixed once for all, we have defined the {\sl near-light cone} domain to be 
\begin{equation}\label{eq1-06-07-2021-new}
\Mscr^\near_{\ell}  := \Big\{  t \geq 2, \quad t-1 \leq r \leq \frac{t}{1-\ell} \Big\}, 
\end{equation} 
where in agreement with~\cite{PLF-YM-main} it is convenient to restrict attention  to $t \geq 2$. Interestingly, this condition can be easily satisfied when 

\bei 

\item either $\eps$ is small with respect to $\epss$, so that the contribution $m$ from the Schwarzschild metric dominates, 

\item both initial data $u_0 \geq 0$ and $u_1 \geq 0$ are non-negative (as is clear from the fact that fundamental solution to the wave equation is a non-negative measure),  

\item or yet a combination of the above two extreme examples, namely, the negative contribution of the perturbation is small with respect to the Schwarzschild mass. 

\eei


\paragraph{Main statement.}

We are in a position to state our main result. In fact, a more general statement concerning perturbations of reference metrics with harmonic decay is actually established in this paper; see Theorem~\ref{theo-main-result-qualitative}. 

\begin{theorem}[Nonlinear stability of self-gravitating Klein-Gordon fields. Perturbations of the Schwarzschild solution]
\label{theo-main-result} 
A constant $C_\star >0$ being fixed, the following result holds for all sufficiently small $\eps, \epss$ satisfying $\eps\leq C_\star \epss$. Consider the reference metric $g_\glue^\star$ defined in \eqref{equa-defineMS-new} by merging together the Minkowski metric $g_\Mink$ and the Schwarzschild metric with mass $\epss$ along a light cone. Consider a set of contraint-satisfying initial data $(g_0,k_0,\phi_0,\phi_1)$, a large integer $N$, and some exponents $(\kappa,\mu,\eps)$  
satisfying  
\begin{equation}\label{eq2-10-04-2022-M}
\kappa\in(1/2,1), 
\qquad 
\mu\in(3/4,1),
\qquad 
\kappa\leq \mu.  
\end{equation}
Then provided the initial data satisfies the light bending condition \eqref{eq3'-27-05-2020-initial-new-Sch} together with the smallness condition
\begin{equation}\label{eq1-10-04-2022-M}
\aligned 
\Ebf^{\text{metric}}_{\kappa,N}  (g_0, k_0)
+ \Ebf^{\text{matter}}_{\mu,N} (\phi_0, \phi_1) 
 \leq \eps, 
\qquad 
\endaligned
\end{equation}
the maximal globally hyperbolic Cauchy development of $(g_0,k_0,\phi_0,\phi_1)$ associated with the Einstein-massive field system \eqref{eq 1 einstein-massif} and \eqref{eq:Talphabeta} is future causally geodesically complete, and asymptotically approaches Minkowski spacetime in all (timelike, null, spacetime) directions. Moreover, the component $g(\lbf, \lbf)$ has a harmonic decay and enjoys the light bending condition, namely 
\begin{equation}\label{eq3-09-05-2021-new-g}
|g(\lbf, \lbf)| \lesssim {\eps_\star+\eps \over t+r+1}, 
\qquad \quad
\inf_{\Mscr^\near_{\ell}} r \, g(\lbf,\lbf)  
\geq \eps_\star/2.
\end{equation}
\end{theorem} 


\paragraph{Harmonic decay.}
 
Our proof of Theorems~\ref{theo-main-result} and \ref{theo-main-result-qualitative} follows from the general method in~\cite{PLF-YM-main} in which we now take the parameter $\lambda$ therein to be the critical value $\lambda = 1$. 
We would like to point out here several significant differences and several new ingredients that are required for the proof in the present paper.  
First of all, the use of a hierarchy property for commutators was not necessary for the range $\lambda<1$, but turns out to be essential now. On the other hand, we revisit \cite[Section~\ref{section-8-added}]{PLF-YM-main} concerning the pointwise decay of metric components (and their derivatives) and we use a decomposition of the wave operators
in which the component $\HN^{00}$ is directly controlled. 
In turn, in the analysis in \cite[Section~\ref{section---12}]{PLF-YM-main} we  now decouple the contributions of the Hessian and of the commutators. Interestingly, this leads to significant simplifications in comparison with the arguments in \cite{PLF-YM-main}, while simultaneously new estimates are required. 
For instance, in the analysis of the (now harmonic) decay of the null metric component in \cite[Section~\ref{section---13}]{PLF-YM-main},  we no longer introduce on the ``loss'' exponent $\theta$. In order to recover the desired $1/r$ behavior for the metric, the application of the Kirchhoff formula must now be done at the $1/r$ level of decay, and this requires a sharp control of the source-terms in the wave equation. Such terms are due to the Ricci curvature of the reference metric (assumed to have sufficient decay, and even to vanish for the Schwarzschild metric), the commutators (discussed earlier) and the quasilinear terms (also discussed earlier). Observe that it is also important to distinguish between zero-order estimates (i.e. without differentiation) and high-order estimates. 
Finally, we point out that, interestingly, our results also apply in the {\sl vacuum spacetimes} by taking $\phi$ to vanish identically: our proof then is somewhat simpler than the one in \cite{LR2}, thanks to the fact that it takes advantage of the light-bending condition ---a consequence of the assumed positivity of the mass of the Schwarzschild metric. 


\section{Euclidean-hyperboloidal foliation and functional inequalities} 

\subsection{Spacetime foliation and vector fields}

\paragraph{Preliminary.}

We provide here an overview of the technical tools introduced by the authors in \cite{PLF-YM-main} when developing the Euclidean-Hyperboloidal Foliation Method. This method generally applies to establish the global existence of solutions with small amplitude for coupled systems of wave and Klein-Gordon equations. The method relies on the following key ingredients. 

\begin{itemize}

\item Construction of a spacetime foliation consisting of asymptotically hyperboloidal slices and asymptotically flat slices.

\item Functional analysis tools on such a foliation: admissible vector fields, Sobolev inequalities, and energy estimates. 

\item Sup-norm estimates for wave equations and Klein-Gordon equations. 

\end{itemize}

\noindent For the full set of notions and results, we refer to \cite{PLF-YM-main}. 

 
\paragraph{Foliation of interest.}
 
In order to foliate $\Mscr$ by spacelike hypersurfaces, we introduce a cut-off function $\chi: \RR \to [0,1]$ satisfying  
\begin{subequations}
\begin{equation}\label{eq3-04-05-2020-one-new} 
\chi(x) = 
\begin{cases}
0,   & x \leq 0,
\\
1,   & x> 1. 
\end{cases}
\end{equation} 
We also introduce the {\sl hyperboloidal and Euclidean radii} at a time $s$ 
\begin{equation}\label{equa-rhoHrhoE-new}
\rhoH(s) := {1 \over 2} (s^2 -1), \qquad \rhoE(s) := {1 \over 2} (s^2 +1), 
\end{equation} 
Then, a function $\xi$ referred to as the {\sl foliation coefficient} is defined by 
\begin{equation}\label{equa-def-xi-new} 
\xi(s,r) := 1-\chi(r- \rhoH(s)) 
= \begin{cases}
1, \quad & r <\rhoH(s),
\\
0, \quad & r > \rhoE(s). 
\end{cases}
\end{equation}
We next define the {\sl Euclidean--hyperboloidal time function} $T = \Time(s,r)$ by solving the ordinary differential equation 
\begin{equation}\label{eq5-05-05-2020-one-new}
\del_r \Time(s,r) = \frac{r \, \xi(s,r)}{(s^2+r^2)^{1/2}},
\qquad 
\Time(s,0) = s. 
\end{equation}
\end{subequations}
In turn, our foliation $\Mscr_s := \{(t,x^a)\in\Mscr\,/\, t = \Time(s,r) \}$ is a one-parameter family of asymptotically Euclidean, spacelike hypersurfaces. The future of the initial surface $\{t=1\}$ is decomposed as 
\begin{subequations}
\begin{equation}
\{t\geq 1\} = \Mscr^{\init}\cup \bigcup_{s\geq 2}\Mscr_s, 
\qquad
\Mscr_s = \MH_s \cup \Mtran_s \cup \Mext_s, 
\end{equation}
with $\Mscr^{\init} := \{(t,x)\,/\, 1\leq t\leq T(2,r)\}$ and 
\label{equa-nosconditions-new}
\begin{equation}\label{eq:617def-one-new}
\aligned
\MH_s :& =   \big \{ t = \Time(s,|x|), \quad \quad
|x| \leq \rhoH(s) \big\}
\qquad 
&& \text{asymptotically hyperboloidal,}
\\
\MM_s :& =   \big\{t = \Time(s,|x| ), \quad \, \rhoH(s)\leq  |x| \leq \rhoE(s) 
\big\}
\qquad 
&&  \text{merging (or transition),}
\\
\Mext_s :& =   \big\{ t=\Time(s), \quad \qquad\quad\rhoE(s) \leq |x| 
\big\}
&& \text{asymptotically Euclidean}.
\endaligned
\end{equation} 
We also write $\MME_s := \Mext \cup \MM$. 
\end{subequations}

 
\paragraph{Frames of interest.}

The following terminology will be used. 

\begin{subequations}

\begin{itemize}

\item {\sl The semi-hyperboloidal frame} (SHF)
\begin{equation}\label{eq:semihf-one-new}
\delH_0: = \del_t, 
\qquad\qquad
\delH_a = \delsH_a: = \frac{x^a}{t} \del_t + \del_a
\end{equation}
is defined in $\Mscr_s$ and is the appropriate frame in the hyperboloidal domain in order to establish the relevant {\sl decay properties in timelike and null directions.} Some of our arguments will involve radial integration based on  
$\delsH_r := (x^a /r)\delsH_a$.

\item {\sl The semi-null frame}  (SNF)
\begin{equation}\label{eq=nulllframedef-one-new}
\delN_0 : = \del_t, 
\qquad \qquad
\delN_a = \delsN_a:= {x^a \over r} \del_t + \del_a 
\end{equation}
is defined in $\Mscr_s$ {\sl except} on the center line $r=0$, and is the appropriate frame within the Euclidean-merging domain in order to exhibit the structure of the nonlinearities of the field equations and, in turn, to establish the  relevant {\sl decay properties in spatial and null directions}. 

\item {\sl The Euclidean--hyperboloidal frame} (EHF) is defined as 
\begin{equation}\label{equation-87-one-new} 
\aligned
\delEH_0 := \del_t, 
\qquad\qquad
\delEH_a & = \delsEH_a 
:=   \del_a + (x^a/r)\del_r \Time \, \del_t 
= \del_a + x^a\xi(s,r) (s^2 + r^2)^{-1/2} \del_t
\endaligned
\end{equation} 
consists of tangent vectors  $\delsEH_a$ to the slices $\Mscr_s$. Observe that $\delEH_a = \delH_a$ in $\MH_s$, while $\delEH_a = \del_a$ in $\Mext_s$. The expressions of the vectors $\delEH_a$ are more involved in the merging $\MM_s$, where the vectors $\delEH_a$ interpolate between $\delH_a$ and $\del_a$.  Some of our arguments will involve radial integration based on  
$\delsEH_r := (x^a/r)\delsEH_a$. 

\end{itemize} 

\end{subequations}


\paragraph{Admissible vector fields.}

Minkowski spacetime admits three sets of {\sl Killing fields}. 
\begin{itemize}

\item The {spacetime translations} generated by the vector fields $\del_\alpha$ ($\alpha=0,1,2,3$). 

\item  The {Lorentz boosts} generated by the vector fields 
$L_a := x_a \del_t  + t \, \del_a$ ($a=1,2,3$).  

\item The {spatial rotations} generated by the vector fields 
$\Omega_{ab} := x_a \del_b - x_b \del_a$ ($a, b=1,2,3$).
\end{itemize} 
\noindent 
We refer to $\del_\alpha, L_a, \Omega_{ab}$ as the {\sl admissible fields} which commute with the wave and Klein-Gordon operators in Minkowski spacetime. In defining high-order norms, we combine admissible vector fields together. An operator $Z=\del^I L^J \Omega^K$ is called an {\sl ordered admissible operator.} To such an operator $Z = \del^I L^J \Omega^K$, we associate its {\bf order, degree,} and {\bf rank} by  
\begin{equation}
\ord(Z) = |I|+|J| + |K|, 
\qquad 
\deg(Z) = |I|, 
\qquad 
\rank(Z) = |J| + |K|. 
\end{equation} 


\subsection{Energy and Sobolev inequalities}

\paragraph{Weighted energy inequality.}

For the foliation under consideration, the fundamental energy functional associated with the wave equation and, more generally, the Klein-Gordon equations involve the {weight} $\zeta = \zeta(t,x)$ defined by 
\begin{equation}\label{equa-defzeta-one-new}
\zeta(s,r)^2 : = 1 - \frac{r^2 \xi^2(s,r)}{s^2+r^2}
= \begin{cases}
s^2/t^2, \quad & r <\rhoH(s),
\\
1, \quad & r > \rhoE(s). 
\end{cases}
\end{equation}
In addition, we introduce a weight which reduces to a constant in the interior of the light cone and is essentially the distance to the light cone in the exterior domain. Throughout, we are given a smooth and non-decreasing function $\aleph$ satisfying $\aleph(y) = 0$ for $y \leq { -1}$ and $\aleph(y) = y+1$ for $y \geq { 0}$, and we set 
\begin{equation}\label{eq:weight-one-new}  
\crochet := 1 + \aleph({r-t}). 
\end{equation}  

Consider the wave or Klein-Gordon equation (with $c \geq 0$) with unknown $u$, namely 
$
g^{\alpha\beta} \del_\alpha \del_\beta u - c^2 u = f$ 
associated with a metric 
$g^{\alpha\beta} =:  g_\Mink^{\alpha\beta} + H^{\alpha\beta}$ and a right-hand side $f$. The energy-flux vector
\begin{subequations}
\begin{equation}\label{eq:39-one-new} 
V_{g, \eta,c}[u]
:=
-\crochet^{2 \eta} \Big(
g^{00} |\del_t u|^2 - g^{ab} \del_au \, \del_bu - c^2u^2, \ 2 g^{a\beta} \del_t u  \del_\beta u \Big)
\end{equation}
depends upon the metric $g$ as well as the weight $\crochet^\eta$.
By setting   
\begin{equation}\label{eq7-08-05-2020-one-new}
\aligned
\Omega_{g, \eta,c}[u] 
& := 
-2\eta\crochet^{-1} \aleph'({ r-t}) (-1,x^a/r) \cdot V_{g,\eta,c}[u]
\\
& = 2 \, \eta\crochet^{2\eta-1} \aleph'({ r-t}) \, \big(g^{\N ab}\delsN_au\delsN_bu - H^{\N00}|\del_tu|^2 + c^2u^2\big),
\\
- G_{g, \eta}[u] 
& :=    \del_tH^{00} | \crochet^\eta \del_t u|^2 - \del_tH^{ab} \crochet^{2 \eta} \del_au\del_b u + 2 \crochet^{2 \eta}  \del_aH^{a\beta} \del_t u\del_{\beta} u, 
\endaligned
\end{equation}
we find 
\begin{equation}\label{eq 2 energy-curved-mod-0-one-new}
\aligned
-2 \crochet^{2 \eta} \, \del_tu f
= \dive  V_{g, \eta,c}[u] + \Omega_{g, \eta,c}[u] - G_{g, \eta}[u]. 
\endaligned
\end{equation}   
\end{subequations}

\begin{subequations}
Next, we consider the integral  
\begin{equation}
\Eenergy_{g, \eta,c}(s,u) := \int_{\Mscr_s}V_{g, \eta,c}[u] \cdot n_s \, d\sigma_s.
\end{equation}
By integrating the energy identity \eqref{eq 2 energy-curved-mod-0-one-new} over the domain limited by a slice of the foliation and the initial slice and using Stokes' formula,  we arrive at the \textsl{energy identity} 
\begin{equation}\label{eq2-08-05-2020-one-new}
\Eenergy_{g, \eta,c}(s_1,u) - \Eenergy_{g, \eta,c}(s_0,u)
+ \int_{\Mscr_{[s_0,s_1]}} \!\!\!\!\!\!\!\!\big(\Omega_{g, \eta,c}[u] - G_{g, \eta}[u]\big) \, dxdt 
= -2 \int_{\Mscr_{[s_0,s_1]}} \!\!\!\!\!\!\!\!  \del_t uf \, \crochet^{2 \eta} dxdt.
\end{equation}
\end{subequations}
The Jacobian $J$ of our parameterization $(t,x) \mapsto (s,x)$, as established in \cite[Lemma~\ref{lem1-22-05-2020}]{PLF-YM-main}, satisfies 
\begin{equation}\label{borne-sup-J} 
J \leq
\begin{cases}
s/t
 \quad      & \text{ in } \MH_s,
\\
\xi s \, (s^2 +r^2)^{-1/2} + (1- \xi) \, 2s
\quad    & \text{ in } \MM_s, 
\\
2s \quad  & \text{ in } \Mext_s,  
\end{cases}
\end{equation} 
and 
\begin{equation}\label{lem1-22-05-2020-new}
\zeta^2 s \lesssim J \lesssim  
\, \zeta^2 s \qquad \text{ in } \MM. 
\end{equation}
After a change of variable and differentiation we can rewrite \eqref{eq2-08-05-2020-one-new} in the equivalent form 
\begin{subequations}
\begin{equation}
\aligned
&\frac{d}{ds} \Eenergy_{g,\eta,c}(s,u)
+ 2\eta \int_{\Mscr_s}  \big( g^{\N ab} \delsN_au\delsN_bu + c^2u^2
\big)  \aleph'({ r-t}) \crochet^{2\eta-1}\ Jdx
\\
&= \int_{\Mscr_s} \Big(G_{g, \eta}[u] + \eta\crochet^{2\eta-1} \aleph'({ r-t}) \HN^{00}  |\del_t u|^2   \Big)\ Jdx 
+ \int_{\Mscr_s}  |\del_t u f| \, \crochet^{2 \eta} \, Jdx, 
\endaligned
\end{equation}   
in which the latter integral is controlled by 
\begin{equation}
\aligned
\int_{\Mscr_s}  |\del_t u f| \, \crochet^{2 \eta} \, Jdx 
\lesssim 
\int_{s_0}^{s_1} \Eenergy_{\eta,c}(s,u)^{1/2} 
\|J\zeta^{-1}\crochet^{\eta} f\|_{L^2(\Mscr_s)} \, ds. 
\endaligned
\end{equation}  
\end{subequations}

We will rely on the following weighted energy estimate in the Euclidean-merging domain concerning any solution $u: \MME_{[s_0,s_1]} \to \RR$ to the wave or Klein-Gordon equation $g^{\alpha\beta} \del_\alpha \del_\beta u - c^2 u = f$ with right-hand side $f: \MME_{[s_0,s_1]} \to \RR$: 
\begin{equation}\label{prop energy-ici-exterior-new-equation} 
\aligned
& \frac{d}{ds} \Eenergy_{g,\eta,c}^{\ME}(s,u) + \frac{d}{ds} \Eenergy^{\Lcal}_{g, c}(s,u;s_0)  
+  2 \eta  \int_{\MME_s}\big( g^{\N ab} \delsN_au\delsN_bu +  c^2u^2\big)  \crochet^{2\eta-1} \aleph'({ r-t}) \, Jdx
\\
& = \int_{\MME_s} \Big(G_{g, \eta}[u] + \eta\crochet^{2\eta-1} \aleph'({ r-t}) \HN^{00}  |\del_t u|^2   \Big) \, Jdxds 
+ \int_{\MME_s} \crochet^{2 \eta} \del_t u f \, Jdxds, 
\endaligned
\end{equation}
in which the latter integral is bounded by 
$\int_{s_0}^{s_1} \, (\Eenergy_{\eta,c}^{\ME}(s,u))^{1/2} \, \big\|J \, \zeta^{-1} \crochet^{\eta} f\big\|_{L^2(\MME_s)} \, ds$, while the second term in the left-hand side is defined as 
\begin{equation}\label{eq8-27-03-2021-new}
\frac{d}{ds} \Eenergy_{g,c}^{\Lcal}(s,u;s_0) 
= s \int \big(- \HN^{00} |\del_t u|^2 + g^{\N ab} \delsN_au\delsN_bu + c^2u^2 \big) \, d\sigma,
\end{equation}
where in the domain of integration is defined by $t=r+1$ and $\rhoH(s_0)\leq r \leq \rhoH(s_1)$.


\paragraph{Commutator estimates.}
 
We also introduce the Japanese bracket
$\la y \ra := \sqrt{ 1+ |y|^2}$ for all real $y$. 
Then, for any function $u=u(t,x)$ we define
$
|u|_N :=  \max_{\ord{Z}\leq N} |Z u|
$
and $
|u|_{p,k} :=  \max_{\ord{Z}\leq N\atop \rank {Z}\leq k} |Z u|$, 
where the first maximum is over all ordered admissible operators.
Various calculus rules based on this notation were established in \cite{PLF-YM-main}, which we will not repeat here. Let us only extract an important consequence of~\cite[Proposition~\ref{prop1-12-02-2020}]{PLF-YM-main}. 
We also decompose each slice $\MME_s$ into near and far regions so that 
\begin{equation}\label{equa-nearfardefinition-new}
\MMEnear_s  := \MME_s\cap \big\{  t-1 \leq r \leq 2t \big\}, \qquad\qquad
 \MMEfar_s := \MME_s\cap \big\{r\geq 2t \big\}. 
\end{equation}

\begin{proposition}[Hierarchy property for quasi-linear commutators. Euclidean-merging domain] 
\label{prop1-12-02-2020-new}

\hskip.1cm

{\bf 1. Estimate in the near-light cone domain.} For any function $u$ defined in the near-light cone domain and for any operator $Z$ with $\ord(Z) = p$ and $\rank(Z) = k$, one has 
\begin{subequations}
\label{eq4-12-02-2020-second}
\begin{equation}\label{eq4a-12-02-2020-second}
\aligned
& \, \big|[Z,H^{\alpha\beta} \del_\alpha \del_{\beta}]u\big|
\\
& \lesssim \big( |\HN^{00}| + t^{-1}|r-t| |H|\big) \, |\del\del u|_{p-1,k-1} 
+ T_{p,k}^\textbf{hier}[H,u] + T_{p,k}^\easy[H,u] + T_{p,k}^{\textbf{super}}[H,u]
\quad \text{ in } \MMEnear_{[s_0, s_1]}, 
\endaligned
\end{equation}
with 
\begin{equation}\label{eq4b-12-02-2020-second}
\aligned
T_{p,k}^\textbf{hier}[H,u] & := \sum_{p_1+p_2=p\atop p_1+k_2=k} \big(|L\HN^{00} |_{p_1-1} + t^{-1}|r-t| | L H |_{p_1-1} \big) \, |\del\del u|_{p_2,k_2},
\\
T_{p,k}^\easy[H,u]
& :=\sum_{p_1+p_2=p\atop k_1+k_2=k} \big(|\del \HN^{00} |_{p_1-1,k_1} + t^{-1}|r-t| |\del H|_{p_1-1,k_1} \big) \, |\del\del u|_{p_2,k_2}, 
\\ 
T_{p,k}^{\textbf{super}}[H,u]
 & := t^{-1}|H||\del u|_p + t^{-1} \hskip-.3cm \sum_{0\leq p_1\leq p-1}|H|_{p_1+1}|\del u|_{p-p_1}. 
\endaligned
\end{equation}
\end{subequations} 
	
{\bf 2. Estimate in the Euclidean-merging domain.} For any function $u$ defined in the Euclidean-merging domain, one has (with $\LOmega= L_a$ or $\Omega_{ab}$) 
\begin{equation}\label{eq5-12-02-2020-new}
\aligned
|[Z,H^{\alpha\beta} \del_\alpha \del_{\beta}]u|
& \lesssim   
|H| \, | \del\del u|_{p-1,k-1} 
\\
& \quad 
+ \hskip-.6cm
\sum_{p_1+p_2=p\atop p_1+k_2=k  \text{ with } k_1=p_1}  
\hskip-.6cm 
| \LOmega H|_{p_1-1} |\del \del u|_{p_2,k_2}
+ \hskip-.3cm 
\sum_{p_1+p_2=p\atop k_1+k_2=k} \hskip-.3cm 
|\del H|_{p_1-1,k_1} |\del\del u|_{p_2,k_2}
\quad \text{ in } \MME_{[s_0,s_1]}. 
\endaligned
\end{equation}
\end{proposition}

The above estimates will be applied to study the evolution of our high-order energy functionals. In \eqref{eq4-12-02-2020-second}, the term $T_{p,k}^\textbf{hier}[H,u]$ is the most challenging contribution to the change of the energy functional but involves terms at a {\sl lower rank}, that is, contains {\sl strictly fewer boosts or rotations}; this structure will allow us to formulate an induction argument on the rank. The term $T^\easy$ is easier since it contains the factor $\del H$ which enjoys (comparatively) good $L^2$ and pointwise decay, while $T^{\textbf{super}}$ contains a favorable $1/t$ factor. 


\paragraph{Functional inequalities.}

It is necessary to revisit the standard Sobolev inequalities and formulate them along the foliation of interest. We only select here two  results. 
In the inequalities below, recall that $\delsME{}^I$ denotes any $|I|$-order operator determined from the fields $\{\delsME_a\}_{a=1,2,3}$, while $\delsE{}^I$ denotes any a $|I|$-order operator determined from the fields $\{\del_a\}_{a=1,2,3}$.  

\begin{proposition}[Sobolev inequality in the Euclidean-merging domain]
\label{pro204-11-2-one-new}
Fix an exponent $\eta \geq 0$ and set $C(\eta) := 1 + \eta + \eta^2$.
For all sufficiently regular functions defined in $\Mscr_{[s_0,s_1]}$ with $2 \leq s_0 \leq s \leq s_1$, one has  
\begin{subequations}
\begin{equation}\label{ineq 2 sobolev-one-new}
r \, \crochet^\eta |u(t,x)| 
\lesssim 
C(\eta) \sum_{|I| + |J| \leq 2} \|\crochet^\eta  \delsME{}^I\Omega^J u\|_{L^2(\MME_s)}, 
\qquad (t,x)\in \MME_s, 
\end{equation}
\begin{equation}\label{ineq 1 sobolev-one-new}
r \, \crochet^\eta |u(t,x)| 
\lesssim
C(\eta) \sum_{|I| + |J| \leq 2} \|\crochet^\eta  \delsE{}^I\Omega^J u\|_{L^2( \Mext_s)},
\qquad 
(t,x)\in \Mext_s. 
\end{equation}
\end{subequations}
\end{proposition}
 
Next, we rely on the boosts $L_a = x^a\del_t+ t \, \del_a$ which are tangent to the hyperboloidal slices. 

\begin{proposition}[Sobolev inequality in the hyperboloidal domain]
\label{prop:glol-Soin-one-new}
For any function defined on a hypersurface $\MH_s$, the following estimate holds (in which $t^2 = s^2+ |x|^2$): 
\begin{equation}
\sup_{\MH_s} t^{3/2} \, |u(t,x)|
\lesssim 
\sum_{|J| \leq 2} \| L^J u\|_{L^2(\MH_s)}
\simeq 
\sum_{m=0,1,2} \| t^m (\slashed \del^\H)^m  u\|_{L^2(\MH_s)}.   
\end{equation}
\end{proposition}


\paragraph{Estimates based on the energy functional.}

We also have the following decay properties for the wave or Klein-Gordon equation; cf.~\cite[Propositions~\ref{eq3-15-05-2020}, \ref{lem 2 d-e-I}, and~\ref{lem 1 d-KG-e}]{PLF-YM-main}, respectively, together with~\cite[Proposition~\ref{17-08-2022-1}]{PLF-YM-main} (namely \eqref{eq1-18-08-2021}  therein). 

\begin{proposition}[Hardy-Poincar\'e inequality for high-order derivatives] 
\label{eq3-15-05-2020-new}
For any $\eta= 1/2 +\delta$ with $\delta>0$ and any sufficiently decaying function $u$ defined in $\Mscr_{[s_0,s_1]}$ and for all $s \in [s_0, s_1]$ one has 
\begin{equation}\label{17-08-2022-trois} 
\| \crochet^{-1 + \eta} |u|_{p,k}\|_{L^2(\MME_s)} 
\lesssim \big(1+\delta^{-1} \big) \, \Fenergy_\eta^{\ME,p,k}(s,u) + \Fenergy_{\eta}^{0}(s,u),
\end{equation}
\begin{equation}\label{eq1-18-08-2021-lambda1}
\| \crochet^{-1+\eta} \zeta|\LOmega u|_{k-1}\|_{L^2(\MME_s)}\lesssim  (1+\delta^{-1}) \, \Fenergy_{\eta}^{\ME,k}(s,u).
\end{equation}
\end{proposition} 

\begin{proposition}[Sobolev decay for wave fields] 
\label{lem 2 d-e-I-new}
For all $\eta \in [0,1)$ and all functions $u$, one has (with $k \leq p$)
\begin{subequations}
\begin{equation}\label{eq 1 lem 2 d-e-I-new}
\big\| r  \, \crochet^\eta \, |\del u|_{p,k} \big\|_{L^\infty(\MME_s)}  
+ \big\|  r^{1+ \eta} \, | \delsN  u |_{p,k} \big\|_{L^\infty(\MME_s)} 
\lesssim (1-\eta)^{-1} \, \Fenergy_\eta^{\ME,p+3, k+3}(s,u)
\end{equation} 
and, for $1/2 < \eta = 1/2+\delta < 1$,
\begin{equation}\label{eq 1 lem 2 d-e-I-facile-new}
\|r \, \crochet^{-1+\eta}|u|_{N-2}\|_{L^\infty(\MME_s)} 
\lesssim \delta^{-1} \, \Fenergy_\eta^{\ME,N}(s,u) + \Fenergy_{\eta}^{0}(s,u). 
\end{equation}
\end{subequations}
\end{proposition}

\begin{proposition}[Pointwise decay of Klein-Gordon fields]
\label{lem 1 d-KG-e-new}
Given any exponents $\eta \in (0,1)$, for any  solution $v$ to $- \Box v + c^2 \, v = f$ defined in $\MME_{[s_0,s_1]}$ one has 
$$ 
c^2 \, |v|_{p,k} \lesssim 
\begin{cases}
r^{-2} \crochet^{1-\eta} \, \Fenergy_{\eta,c}^{\ME,p+4,k+4}(s,v) + |f|_{p,k}
&  \text{in }\Mnear_{[s_0,s_1]},
\\
r^{-1-\eta} \, \Fenergy_{\eta,c}^{\ME,p+2,k+2}(s,v)\quad 
& \text{in }\Mfar_{[s_0,s_1]}.
\end{cases}
$$
\end{proposition}
 

\section{Nonlinear stability of Einstein-Klein-Gordon spacetimes}

\subsection{The class of reference metrics}  
 
We introduce a class of reference spacetime metrics $(\RR^{3+1}_+, g^\star)$ which represent ``approximate solutions'' to Einstein's vacuum equations.
\begin{subequations} 
\label{equa-assume-reference} 
\begin{itemize}

\item The metric $g^\star = \gMink + h^\star$ is {\sl asymptotically Minkowski} in the sense that 
\begin{equation}\label{equa-31-12-20-new}    
|h^{\star}|_{N+2} + \la r+t\ra|\del h^{\star}|_{N+1} + \la r+t\ra^2|\del\del h^{\star}|_{N}
\leq 
 \epss \la r+t\ra^{-1},
\end{equation} 

\item This metric is {\sl almost Ricci flat} in the sense that  
\begin{equation}\label{eq4-09-05-2021-new}
|\wR^{\star}|_{N} + \la r-t \ra|\del\wR^{\star}|_{N-1}
\leq
\begin{cases}
\epss^2 \la r+t\ra^{-4} \quad &\text{ in } \MME_{[s_0, + \infty)},
\\
\epss \la r+t\ra^{-3} \quad &\text{ in } \MH_{[s_0, + \infty)}, 
\end{cases}
\end{equation}  
where $\wR^{\star}$ denotes the ``reduced'' Ricci curvature classically defined by removing suitable contractions of the Christoffel symbols. 

\item This metric satisfies the {\sl light-bending property,} that is, 
\begin{equation}\label{equa-bending-repeat-new} 
\inf_{\Mscr^\near_{\ell}} \big( r \, h^\star(\lbf,\lbf)\big)
\geq \epss, 
\end{equation}
in which $g^\star(\lbf,\lbf)$ is (linearly) equivalent to $- h_\star^{\N 00}$ and we recall that $\lbf = \del_t - (x^a/r)\del_a$ (as stated earlier). 

\end{itemize}

\end{subequations}

\noindent As a consequence of \eqref{eq4-09-05-2021-new} and \eqref{borne-sup-J}  and after introducing a small parameter $\delta>0$ and some decay exponent 
\begin{equation}\label{equa-paramet-repeat-new}
\kappa \in (1/2, 1), 
\end{equation} 
the following rough bound follows by integration: 
\begin{equation}\label{eq1-21-05-2021-new}
\|\crochet^\kappa J{\zeta}^{-1} \, |\wR^\star |_N \|_{L^2(\MME_s)} 
\leq 
R_{\err}^{\star}(s) 
= C_{R^{\star}}  \epss^2\delta^{-1/2} \,  s^{-1-\delta}, 
\end{equation}
in which and $C_{R^{\star}} >0$ is a constant determined by $N$ and 
\begin{equation}
\int_{s_0}^sR_{\err}^{\star}(s')
\, ds'
\leq C_{R^{\star}} \, \epss^2 \delta^{-3/2} s_0^{-\delta}. 
\end{equation}
For further motivations and results we also refer the reader to the more general theory in~\cite[Section~\ref{sec1-23-05-2021}]{PLF-YM-main}.  Furthermore, it is straightforward to check that the Minkowski-Schwarzschild reference metric $g_\glue^\star$ in the introduction satisfies all of the conditions above.


\subsection{The class of initial data sets} 

In the following $\delta\in(0,1)$ denotes a fixed exponent which should be small in comparison to a small multiple of $\kappa-1/2$, which arises as a critical value for the decay of the metric perturbation.  It will be interesting to keep track, in our estimates, of the most relevant parameters and constants. 
For the initial data we assume the following smallness conditions. 
\begin{itemize} 

\item In the hyperboloidal domain,  
we assume 
\begin{subequations}\label{eqs-int-new-init}
\begin{equation}\label{eqsa-int-new-init}
\Fenergy^{\Hcal, N-5}(s_0,u) + s_0^{-1/2} \, \Fenergy_c^{\Hcal,N-5}(s_0, \phi) \leq C_0 \eps s_0^{\delta}, 
\end{equation}
\begin{equation}\label{eqsb-int-new-init}
\Fenergy^{\Hcal,N-7}(s_0,u) + \Fenergy_c^{\Hcal,N-7}(s_0, \phi) \leq C_0 \eps s_0^{\delta}.
\end{equation}
\end{subequations}

\item In the Euclidean-merging domain, 
we assume 
\begin{subequations}\label{eqs1-14-01-2021-new-init} 
\begin{equation}\label{eq1-14-01-2021-new-init}
\Fenergy_{\kappa}^{\ME,N}(s_0,u) + s_0^{-1} \, \Fenergy_{\mu,c}^{\ME,N}(s_0, \phi) \leq C_0 \eps s_0^{\delta}, 
\end{equation}
\begin{equation}\label{eq2-14-01-2021-new-init}
\Fenergy_{\kappa}^{\ME,N-5}(s_0,u) +\Fenergy_{\mu,c}^{\ME,N-5}(s_0, \phi) \leq C_0 \eps s_0^{\delta}. 
\end{equation}
\end{subequations} 

\item 
Finally, we impose the linear light-bending condition on the {\sl total} initial data (defined as the sum of the reference data and the  perturbation)  (with $\ell \in (0, 1/2]$ being fixed): 
\begin{equation}\label{eq3'-27-05-2020-initial-new}
\inf_{\Mscr^\near_{\ell, [s_0, + \infty)}} \big( r \, g^\star(\lbf,\lbf) + r \, u_\init(\lbf,\lbf) \big) 
\geq \epss. 
\end{equation} 

\end{itemize} 

Concerning \eqref{eq3'-27-05-2020-initial-new}, we point out that the contributions due to the nonlinearities of the Einstein-matter equations on the key metric component $h^{\N 00}$ will turn out to be negligible in comparison with the background and initial linear contributions which are of amplitude $\epss/r$. As a consequence of our energy bounds, recalling \cite[\eqref{eq5-07-01-2022}]{PLF-YM-main} and applying the Klainerman-Sobolev inequality on the initial slice, we obtain
$$
|\del_t u(1,x)| + |\del_x u(1,x)|\lesssim C_0 \eps \la r\ra^{-3/2-\kappa}.  
$$
Integrating with the (spatial) radial derivative from spatial infinity, we obtain
$$
\la r\ra|\del_tu(1,x)| + |u(1,x)|\lesssim C_0 \eps \la r\ra^{-1/2-\kappa}.
$$
Recalling that $\kappa +1/2> 1$, we apply~\cite[Proposition \ref{prop1-21-12-2021}]{PLF-YM-main} (based on Kirchkoff formula) and obtain 
\begin{equation}\label{eq3-09-05-2021-new}
|u_{\init}|_{N-4} \lesssim C_0\eps (t+r+1)^{-1}
\quad \text{ in }  \Mscr_{[s_0, + \infty)}.
\end{equation} 
Here, the linear development $u_\init$ of the initial data $(u_{0 \alpha\beta},u_{1 \alpha\beta})$ is the solution of the (free, linear) wave equation with this initial data. (See also \cite[Proposition~\ref{prop1-12-01-2022}]{PLF-YM-main}.) The passage from geometric statements to statements in coordinates was already discussed in an appendix of~\cite{PLF-YM-main} to which the reader is referred. 


\subsection{Nonlinear stability statement} 

We distinguish between estimates at low- or high-order of differentiation, estimates within the hyperboloidal and Euclidean-merging domains, and a positivity condition near the light cone. The following energy bounds will be established at each time $s \geq s_0$. 

\begin{itemize} 

\item In the hyperboloidal domain,  
we have 
\begin{subequations}\label{eqs-int-new}
\begin{equation}\label{eqsa-int-new} 
\Fenergy^{\Hcal, N-5}(s,u) + s^{-1/2} \, \Fenergy_c^{\Hcal,N-5}(s, \phi) \leq (\epss + C_1\eps) \, s^{\delta},
\end{equation}
\begin{equation}\label{eqsb-int-new}
\Fenergy^{\Hcal,N-7}(s,u) + \Fenergy_c^{\Hcal,N-7}(s, \phi) \leq (\epss + C_1\eps) \, s^{\delta}. 
\end{equation}
\end{subequations}

\item In the Euclidean-merging domain, 
we have 
\begin{subequations}\label{eqs1-14-01-2021-new}  
\begin{equation}\label{eq1-14-01-2021-new}
\Fenergy_{\kappa}^{\ME,N}(s,u) + s^{-1} \, \Fenergy_{\mu,c}^{\ME,N}(s, \phi) \leq (\epss + C_1\eps) \, s^{\delta},
\end{equation}
\begin{equation}\label{eq2-14-01-2021-new}
\Fenergy_{\kappa}^{\ME,N-5}(s,u) +\Fenergy_{\mu,c}^{\ME,N-5}(s, \phi) \leq (\epss + C_1\eps) \, s^{\delta}. 
\end{equation}
\end{subequations} 
 \item  
Near the light cone, we have the light-bending condition
\begin{equation}\label{eq3-27-05-2020-new} 
\inf_{\Mscr^\near_{\ell, s}} 
(- r \, \hN{}^{00} )\geq 0.
\end{equation}

\end{itemize} 

   
\begin{theorem}[Nonlinear stability of self-gravitating Klein-Gordon fields. Perturbations of metrics with harmonic decay]
\label{theo-main-result-qualitative} 
A constant $C_\star>0$ being fixed, the following result holds for all sufficiently small $\eps, \epss$ satisfying $\eps\leq C_\star \epss$. Consider a reference metric $g^\star$, defined as a perturbation of the Minkowski metric satisfying  \eqref{equa-assume-reference}  for some decay exponent 
\begin{equation}
\kappa \in (1/2,1).
\end{equation}
Consider an initial data set $(u_0, u_1,\phi_0,\phi_1)$ satisfying the decay and regularity conditions \eqref{eqs-int-new-init}--\eqref{eq3'-27-05-2020-initial-new} at the initial time $s_0$  
for some parameters $(N,\mu,\eps)$ for a sufficiently large integer $N$ with  
\begin{equation}\label{eq2-10-04-2022-M-repeat}
\mu\in(3/4,1),
\qquad 
\kappa\leq \mu.  
\end{equation}
Then, the initial value problem for the Einstein-Klein-Gordon system in wave gauge \eqref{MainPDE-limit} admits a global in time solution $(g, \phi)$ defined for all $t \geq 1$ and $x \in \RR^3$. This solution satisfies the decay and regularity conditions \eqref{eqs-int-new}--\eqref{eq3-27-05-2020-new} for all times $s \geq s_0$ and remains close to the reference spacetime $(\RR_+^{3+1}, g^\star)$ (in the norms under consideration) and enjoys the following decay estimate
\begin{equation}\label{eq1-24-02-2022-M}
\inf_{\Mscr^\near_{\ell, s}} 
(- r \, \hN{}^{00} )\geq \epss/2,
\qquad 
\quad 
|h(\lbf,\lbf)| \lesssim {\epss \over t+ r+1}. 
\end{equation}  
\end{theorem}
 

\section{Consequences of the energy estimates} 
\label{sec1-23-05-2021-new}

\subsection{Bootstrap assumptions and basic estimates}
\label{subsec1-30-05-2020}

\paragraph{Bootstrap assumptions.}

This result will be established in the following sections by following the strategy already outlined in the introduction. From now on we assume that the local-in-time solution in fact extends over an interval $[s_0,s_1]$, and we assume that $[s_0,s_1]$ is the maximal interval of time within which \eqref{eqs-int-new}, \eqref{eqs1-14-01-2021-new} and \eqref{eq3-27-05-2020-new} hold for all $s \in [s_0, s_1]$ so that, by continuity, one of these conditions is an equality at the end time $s_1$. Our objective is to establish the following {\bf  improved estimates} for all $s \in [s_0, s_1]$:
\begin{subequations}\label{eqs'-int-new}
\begin{equation}\label{eqs'-int-new-a}
\Fenergy^{\Hcal, N-5}(s,u) + s^{-1/2} \, \Fenergy_c^{\Hcal,N-5}(s, \phi) \leq \frac{1}{2} (\epss + C_1\eps) \, s^{\delta},
\end{equation}
\begin{equation}
\Fenergy^{\Hcal,N-7}(s,u) + \Fenergy_c^{\Hcal,N-7}(s, \phi) \leq \frac{1}{2} (\epss + C_1\eps)\, s^{\delta},
\end{equation}
\end{subequations}
\begin{subequations}\label{eqs'-ext}
\begin{equation}\label{eq5'-03-05-2020-new}
\Fenergy_{\kappa}^{\ME,N}(s,u) + s^{-1} \, \Fenergy_{\mu,c}^{\ME,N}(s, \phi) \leq\frac{1}{2} (\epss + C_1\eps) \, s^{\delta},
\end{equation}
\begin{equation}\label{eq6'-03-05-2020-new}
\Fenergy_{\kappa}^{\ME,N-5}(s,u) +\Fenergy_{\mu,c}^{\ME,N-5}(s, \phi) \leq \frac{1}{2} (\epss + C_1\eps) \, s^{\delta},
\end{equation}
\end{subequations}
\begin{equation}\label{eq3'-27-05-2020-new}
\inf_{\Mscr^\near_{\ell, [s_0,s_1]}} (- r \, \hN{}^{00} ) 
\geq  {1\over 2} \epss.
\end{equation}
Observe that the factor $s^{-1/2}$ arises in \eqref{eqs'-int-new-a}, while $s^{-1}$ arises in \eqref{eq5'-03-05-2020-new}. In the following, our statements will be established under the assumptions on the reference metric and initial data and the bootstrap assumptions, as stated above. 


\paragraph{Energy-based estimates.}

The following $L^2$ estimates for the metric perturbation and the matter field are immediate from the bootstrap assumption \eqref{eq1-14-01-2021-new} and \eqref{eq2-14-01-2021-new} in the Euclidean-merging domain:
\begin{subequations}\label{eq7-03-05-2020-new}
\begin{equation}\label{eq7a-03-05-2020-new}
\| \crochet^\kappa \zeta \, | \del u|_N \|_{L^2(\MME_s)}
+
\| \crochet^\kappa |\delts u|_N\|_{L^2(\MME_s)} 
\lesssim (\epss + C_1\eps) \, s^{\delta}, 
\end{equation}
\begin{equation}\label{eq7b-03-05-2020-new}
\| \crochet^\mu \zeta \, |\del \phi|_p\|_{L^2(\MME_s)} + 
\| \crochet^\mu |\delts \phi|_p\|_{L^2(\MME_s)} +
\| \crochet^\mu |\phi|_p\|_{L^2(\MME_s)}
\lesssim
(\epss + C_1\eps) \, 
\begin{cases} 
s^{1+\delta}, \, & p=N,
\\
s^{\delta}, & p=N-5. 
\end{cases}
\end{equation}
\end{subequations}  
As a consequence of the weighted Poincar\'e inequality in Proposition~\ref{eq3-15-05-2020-new} we find 
\begin{equation}\label{eq1-12-05-2020-new}
\| \crochet^{-1+\kappa} |u|_{p,k} \|_{L^2(\MME_s)} 
\lesssim \delta^{-1} \, \Fenergy_{\kappa}^{\ME,p,k}(s,u) + \Fenergy_{\kappa}^{0}(s,u)
\lesssim \delta^{-1} (\epss + C_1\eps) \, s^{\delta},
\end{equation}
which provides us with a control of the metric components possibly without partial derivatives. 


\paragraph{Pointwise decay of the metric.} 

Basic sup-norm estimates are then derived by applying the generalized Sobolev estimate\footnote{We neglect a factor $1/(1-\kappa)$ since we have fixed $\kappa<1$.} in Proposition~\ref{lem 2 d-e-I-new}, and we thus control the metric perturbation $\del u$ at order $N-3$, as follows:  
\begin{equation}\label{eq10-02-05-2020-new}
\| r \, \crochet^\kappa \, |\del u|_{N-3} \|_{L^\infty(\MME_s)}
+\| r^{1+\kappa} |\delts u|_{N-3} \|_{L^\infty(\MME_s)}
\lesssim (\epss + C_1\eps)\, s^{\delta}. 
\end{equation}
In view of $h_{\alpha\beta} = h^{\star}_{\alpha\beta} + u_{\alpha\beta}$, by combining this result with the asymptotically Minkowski condition \eqref{equa-31-12-20-new} on the reference metric, for the metric unknown $\del h$ we obtain 
\begin{equation}\label{eq1-26-05-2021-new}
\| r \, \crochet^{\kappa}|\del h|_{N-3}  \|_{L^\infty(\MME_s)} + \| r^{1+\kappa} |\delsN h|_{N-3} \|_{L^\infty(\MME_s)}
\lesssim (\epss + C_1\eps) s^{\delta}.
\end{equation}
We apply the Sobolev decay for wave fields \eqref{eq 1 lem 2 d-e-I-facile-new}, together with \eqref{eqsa-int-new} and \eqref{eq1-14-01-2021-new}, and obtain 
\begin{equation}\label{eq7-15-05-2020-new}
\| r \, \crochet^{\kappa-1} \,  |u|_{N-2}  \|_{L^\infty(\MME_s)} 
\lesssim \delta^{-1} \, (\epss + C_1\eps) \, s^{\delta}
\end{equation}
and, using again \eqref{equa-31-12-20-new} enjoyed by the reference, 
\begin{equation}\label{eq1-09-05-2021-new}
\| r^{\kappa} |h|_{N-2}  \|_{L^\infty(\MME_s)} 
\lesssim 
\delta^{-1} (\epss + C_1\eps) s^{\delta}. 
\end{equation}


\paragraph{Pointwise decay of the matter field.}

The Sobolev decay inequality in~\cite[Proposition~\ref{lem 2 d-e-I}]{PLF-YM-main} 
and the bootstrap assumptions \eqref{eqs1-14-01-2021-new} provide us with sup-norm estimates 
\begin{equation}\label{eq11a-02-05-2020-new}      
\| r \, \crochet^\mu \, |\del \phi|_{p-3} \|_{L^\infty(\MME_s)}
+\|r^{1+\mu} |\delts \phi|_{p-3} \|_{L^\infty(\MME_s)}
\lesssim (\epss + C_1\eps)\, 
\begin{cases} 
s^{1+\delta},\, &p=N,
\\
s^{\delta}, \, & p=N-5,
\end{cases}
\end{equation}
and, thanks to the consequence  in~\cite[\eqref{eq decay-v-repeat000-two}]{PLF-YM-main} of our generalized Sobolev inequality,
\begin{equation}\label{eq1-18-05-2020-new}
\| r \, \crochet^\mu \, |\phi|_{p-2} \|_{L^\infty(\MME_s)}
\lesssim  
(\epss + C_1\eps) \, 
\begin{cases}
s^{1+\delta}, \quad
& p=N,
\\
s^{\delta}, & p=N-5.
\end{cases}
\end{equation}
However, within $\Mnear_s$ this is not sufficient for our purpose below and we establish a stronger decay, as follows.

\begin{lemma} 
\label{lemma-111-new} 
The matter field satisfies the pointwise bound 
$$
\| r^2 \, \crochet^{\mu-1} \, |\phi|_{p-4} \|_{L^\infty(\MME_s)}
\lesssim (\epss + C_1\eps)\,
\begin{cases}
 s^{1+2\delta},  \quad  
& p=N,
\\
s^{2\delta}, \quad 
&  p=N-5. 
\end{cases}
$$
\end{lemma} 

\begin{proof} In view of the inequality \eqref{eq1-18-05-2020-new}, we want to ``trade'' a factor $\crochet$ for a factor $r$, and we only need to deal with the domain $\MMEnear_s$. We need the decay property near the light cone stated in Proposition~\ref{lem 1 d-KG-e-new}. We consider the Klein-Gordon equation $g^{\alpha\beta} \del_{\alpha} \del_{\beta} \phi - c^2 \phi = 0$ and with the notation therein, we set 
$$
f = h^{\mu\nu} \del_{\mu} \del_{\nu} \phi 
= h^{\star\mu\nu} \del_{\mu} \del_{\nu} \phi + u^{\mu\nu} \del_{\mu}\del_{\nu}\phi.
$$ 
For the first term $h^{\star\mu\nu} \del_{\mu} \del_{\nu} \phi$ above, the decay condition on $h^{\star}$ in \eqref{equa-31-12-20-new} yields us 
$$
| h^{\star\mu\nu} \del_{\mu} \del_{\nu} \phi|_{p-4} \lesssim 
\epss (\epss + C_1\eps) \, r^{-2} \crochet^{-\mu}
\begin{cases} 
s^{1+\delta}, \quad & p = N,
\\
s^{\delta}, \quad & p=N-5.
\end{cases}
$$
For the second term $u^{\mu\nu} \del_{\mu}\del_{\nu}\phi$, by~\cite[Lemma~\ref{lem-small}]{PLF-YM-main} we have $|u^{\mu\nu} |_p\lesssim |u|_p$ and, by recalling the Sobolev decay \eqref{eq7-15-05-2020-new} and \eqref{eq11a-02-05-2020-new}, we find  
$$
|u^{\mu\nu} \del_{\mu} \del_{\nu} \phi|_{p-4} \lesssim |u|_{p-4} |\del\phi|_{p-3} \lesssim 
\delta^{-1}(\epss + C_1\eps)^2 r^{-2} \crochet^{1-\kappa - \mu}
\begin{cases} 
s^{1+2\delta},  & \quad p=N,
\\
s^{2\delta},   & \quad p=N-5. 
\end{cases}
$$ 
On the other hand, recalling \eqref{eqs1-14-01-2021-new} we have 
$$ 
r^{-2} \crochet^{1-\mu} \, \Fenergy_{\mu,c}^{\ME,p,k}(s,\phi)\lesssim 
(\epss + C_1\eps) \, r^{-2}\crochet^{1-\mu}
\begin{cases} 
s^{1+\delta},\quad &p=N,
\\
s^{\delta},\quad &p=N-5.
\end{cases}
$$ 
We are thus in a position to use the general pointwise decay enjoyed by Klein-Gordon fields, as  stated \cite[Proposition~\ref{lem 1 d-KG-e}]{PLF-YM-main} and we arrive at the desired conclusion.
\end{proof}  


\subsection{Basic estimates for nonlinearities: energy norm} 
\label{sectionn33-new}

\paragraph{Improving the energy estimates.}

We differentiate the metric and matter evolution equations \eqref{MainPDE-limit} with respect to $Z=\del^IL^J\Omega^{K}$ 
(with $\ord(Z)=|I|+|J|+|K|\leq N$ or $\leq N-5$) and obtain 
\begin{subequations}
\label{eq11-15-05-2020-ab-new}
\begin{equation}\label{eq11-15-05-2020-new}
\aligned
\Boxt_g Z u 
& =   -[Z,h^{\mu\nu} \del_{\mu} \del_{\nu}]u_{\alpha\beta} + Z \big( \Pbb_{\alpha\beta}^{\star}[u]\big) + Z \big( \Qbb_{\alpha\beta}^{\star}[u] \big) 
\\
& \quad + Z \Big( \Ibb^{\star}_{\alpha\beta}[u] + 2 \, \Rwave_{\alpha\beta}
- u^{\mu\nu} \del_{\mu} \del_{\nu} g^\star_{\alpha\beta}  \Big)
- 8\pi \, Z \big(2 \, T_{\alpha\beta} - Tg_{\alpha\beta} \big), 
\endaligned
\end{equation}
and
\begin{equation}\label{eq12-15-05-2020-new}
\Boxt_g Z \phi - c^2Z\phi = -[Z, h^{\mu\nu} \del_{\mu} \del_{\nu}]\phi,
\end{equation}
\end{subequations}
in which $\Ibb^{\star}[u]$ represent the reference-perturbation interaction terms presented in \cite[\eqref{eq1-06-02-2022}]{PLF-YM-main}. Our main task is to control 
$
\|J \, \zeta^{-1}  \crochet^{\kappa}|T|_N\|_{L^2(\MME_s)},
$ 
where $T$ represents any of the terms in the right-hand sides of \eqref{eq11-15-05-2020-ab-new}. 
Thanks to \eqref{lem1-22-05-2020-new} we have 
\begin{equation}\label{eq1-17-08-2021-new}
\|J \, \zeta^{-1}  \crochet^{\kappa}  |T|_N\|_{L^2(\MME_s)} \lesssim \| s \, \crochet^{\kappa} \zeta |T|_N\|_{L^2(\MME_s)}.
\end{equation}  


\paragraph{Linear-critical and super-critical nonlinearities.} 

We treat first the comparatively easier terms, that is, 
the reference-perturbation interaction terms $\Ibb^{\star}[u]$, 
the term $u^{\mu\nu} \del_{\mu} \del_{\nu} h^\star_{\alpha\beta}$,  
and the source terms associated with the scalar field. 
On the other hand, the null terms, the quasi-null term and the commutators require different arguments and will be the subject of later sections. We set 
\begin{equation}
\aligned
W^{\textbf{linear}} 
& :=  \Fbb_{\alpha\beta}(g^\star, g^\star;\del g^\star, \del u) + \Fbb_{\alpha\beta}(g^\star, g^\star;\del u, \del g^\star),
\\
W^\super & :=  \Fbb_{\alpha\beta}(u,g^\star;\del g^\star, \del g^\star) + \Fbb_{\alpha\beta}(g^\star,u;\del g^\star, \del g^\star)
\\
& \quad + \Bbb^\star_{\alpha\beta} [u] + \Cbb^{\star}_{\alpha\beta}[u]
- 8\pi \, (2 \, T_{\alpha\beta} - Tg_{\alpha\beta}) + 2 \, \Rwave_{\alpha\beta}.
\endaligned 
\end{equation}
It is a simple matter to check that  
\begin{equation}\label{eq2-02-06-2022-second}
\|J\,\zeta^{-1}\crochet^{\kappa}|W^{\textbf{linear}}|_{p,k}\|_{L^2(\MME_s)}
\lesssim \delta^{-1}(\epss+C_1\eps)^2s^{-1-\delta}, 
\end{equation}
while 
\begin{equation}\label{eq1-22-03-2021-second}
\|J \, \zeta^{-1}  \crochet^{\kappa}| W^\super |_{p,k}\|_{L^2(\MME_s)}
\lesssim  
\delta^{-1}(\epss + C_1\eps)^2s^{-1-\delta}.
\end{equation}
On the other hand, for super-critical terms involving the reference metric we find 
\begin{equation}
\aligned
& \|s\crochet^{\kappa}\zeta|\Fbb_{\alpha\beta}(u,g^\star;\del g^\star, \del g^\star)|_N\|_{L^2(\MME_s)} + \|s\crochet^{\kappa}\zeta|\Bbb^{\star}_{\alpha\beta}[u]|_N \|_{L^2(\MME_s)}
+\|s\crochet^{\kappa}\zeta|\Cbb^{\star}_{\alpha\beta}[u]|_N \|_{L^2(\MME_s)} 
\\
& \lesssim  
(\epss + C_1\eps)^2s^{-1-\delta},
\endaligned
\end{equation}
as well as for super-critical terms involving the matter field
$$
\| s \, \crochet^{\kappa} \zeta |2 \, T_{\alpha\beta} - Tg_{\alpha\beta} |_N\|_{L^2(\MME_s)} 
\lesssim  (\epss + C_1\eps)^2 s^{-1-\delta}. 
$$
We can now apply the Hardy-Poincar\'e inequality in Proposition~\ref{eq3-15-05-2020-new} and obtain 
\begin{equation}\label{eq5-17-06-2020-second}
\|s\crochet^{\kappa} \zeta \, | u \, \del \del h^{\star} |_N\|_{L^2(\MME_s)} 
\lesssim \delta^{-1}\epss (\epss + C_1\eps) s^{-3 + \delta} \lesssim \delta^{-1}(\epss + C_1\eps)^2 s^{-2}. 
\end{equation}


\subsection{Basic estimates for nonlinearities: pointwise norm}

\paragraph{Estimates for the metric.}

From \cite[\eqref{eq1-04-12-2020} and \eqref{eq11-04-06-2020}]{PLF-YM-main} we also recall  
\begin{equation}
\label{eq1-04-12-2020-second}
|\Ibb^{\star}_{\alpha\beta}[u]|_{N-4} + |u^{\mu\nu}\del_{\mu}\del_{\nu}g^{\star}|_{N-4} 
\lesssim { \delta^{-1}}(\epss + C_1\eps)^2 r^{-3}\crochet^{-\kappa}s^{2\delta}
\end{equation}
and  
\begin{equation}\label{eq11-04-06-2020-second}
|u^{\mu\nu} \del_{\mu} \del_{\nu}g^{\star}_{\alpha\beta}|_{N-3} \lesssim {\delta^{-1}} 
\epss (\epss + C_1\eps) \, r^{-4} \crochet^{1-\kappa}s^{\delta}. 
\end{equation}  


\paragraph{Estimates for the matter field.}

For the matter interaction terms, thanks to \eqref{eq11a-02-05-2020-new}, \eqref{eq1-18-05-2020-new}, and Lemma~\ref{lemma-111-new} we have  
\begin{equation}\label{eq1-28-11-2020-second}
\sum_{\alpha, \beta} |2 \, T_{\alpha\beta} - ( T_{\gamma\gamma} g^{\gamma\gamma} ) \, g_{\alpha\beta}|_{N-3} 
=: | \Tbb(\phi)|_{N-3}
\lesssim (\epss + C_1\eps)^2
r^{-3}\crochet^{1-2\mu} s^{1+3\delta}. 
\end{equation}
Here and from now on, we use the short-hand notation $\Tbb(\phi)$ for the matter term contributions $2 \, T_{\alpha\beta} - (T_{\gamma\gamma} g^{\gamma\gamma} ) \, g_{\alpha\beta}$. For the proof, we recall the expression of $T_{\alpha\beta}$ and, using the pointwise decay of Klein-Gordon fields in Proposition~\ref{lem 1 d-KG-e-new} obtain (see \cite[\eqref{eq11a-02-05-2020} and Lemma~\ref{lemma-111}]{PLF-YM-main}) 
\begin{equation}
|\del\phi\del\phi|_{N-3}
\lesssim |\del\phi|_{N-3}|\del\phi|_{[(N-3)/2]}
\lesssim 
(\epss + C_1\eps)^2 \, r^{-3}\crochet^{1-2\mu}s^{1+3\delta}.
\end{equation} 
The bound for $|\phi^2|_{N-3}$ is similar and we omit the details. Using the fact that $|h^{\alpha\beta}|_{N-3} \lesssim 1$,  we obtain \eqref{eq1-28-11-2020-second}. 


\paragraph{Bounds on null metric component.} 

The null component $g^{\N00}$ of the metric plays a special role in our analysis, and the decay of its gradient considered now. Thanks to the wave gauge condition as derived in \cite[ Lemma~\ref{lemma-12-04-2020}]{PLF-YM-main}, we have the decay property 
\begin{equation}\label{eq1-17-07-2020-second}
|\del g^{\N00}|_{N-3} \lesssim 
{ \delta^{-1}}\big(\epss + C_1\eps\big) r^{-1-\kappa}s^{\delta}
\qquad \text{ in } \MME_{[s_0,s_1]}.  
\end{equation}   
Indeed, in the right-hand side the inequality in \cite[Lemma~\ref{lemma-12-04-2020}]{PLF-YM-main} with $p=N-3$, namely 
$$ 
\aligned
|\del g^{\N00} |_{N-3} 
\lesssim  
& |\delsN h|_{N-3}  + r^{-1} |h|_{N-3} + \sum_{p_1+p_2 = N-3} |h|_{p_1} |\del h|_{p_2}.
\endaligned
$$ 
We need to substitute the Sobolev bounds \eqref{eq1-26-05-2021-new} on $|\del h|_{N-3}$ and $|\delts h|_{N-3}$, together with the bound for $|h|_{N-2}$ from \eqref{eq1-09-05-2021-new}. 


\section{Estimates based on the structure of the Einstein equations}
\subsection{Commutator and Hessian estimates for the metric}

\paragraph{Basic pointwise estimates.} 

Throughout, the bootstrap assumptions \eqref{eqs1-14-01-2021-new} and \eqref{eq3-27-05-2020-new} are assumed, and $(\epss + C_1\eps)$ is taken to be sufficiently small. So far all the estimates are similar to the ones we had derived in the general regime (but by taking $\lambda = 1$) in \cite{PLF-YM-main}. Now, we focus on the {\sl harmonic decay} estimates. From \cite[Sections~\ref{sectionN-7} and \ref{section-8-added}]{{PLF-YM-main}}, we recall the following result (extracted from Propositions~\ref{prop1-12-02-2020} and \ref{prop1-22-05-2020} therein). Observe that the inequalities now are simpler in comparison to the general $\lambda$ regime, due to the fact that now we have a ``good'' control of the component $H^{\N00}$. We first write  \cite[Proposition \ref{prop1-12-02-2020}]{PLF-YM-main} in a different form that will be more appropriate in the present paper, and next we state our improved estimate for the Hessian.

\begin{proposition}[Hessian for the wave equation near the light cone] 
\label{prop1-22-05-2020-second}
Consider the wave operator $\Boxt_g u=  g^{\alpha\beta} \del_\alpha \del_\beta u$ in which $g^{\alpha\beta} = g_\Mink^{\alpha\beta} + H^{\alpha\beta}$, and assume that, for some $\eps_1 \ll 1$, 
\begin{equation}\label{eq2-24-02-2022-M}
|H|\leq \eps_1 \quad \text{ in } \Mscr^{\near}_{\ell, [s_0,s_1]}, 
\end{equation}
\begin{equation}\label{eq1-15-02-2020-second}
\big|\HN^{00} \, \big|\leq \eps_1\frac{1+|r-t|}{r} \quad \text{ in } \Mscr^{\near}_{\ell, [s_0,s_1]}. 
\end{equation}
Then, for any function $u$  
one has 
\begin{equation}\label{eq3-15-02-2020-second}
\frac{1+|r-t|}{r} |\del\del u|
\lesssim |\Boxt_g u| + t^{-1} |\del u|_{1,1} 
\quad \text{ in } \Mscr^{\near}_{\ell, [s_0,s_1]}
\end{equation}
and, more generally at arbitrary order $(p,k)$, 
\begin{equation}\label{eq3-28-12-2020-second}
\frac{1+|r-t|}{r} |\del\del u|_{p,k} 
\lesssim |\Boxt_g u|_{p,k} + t^{-1} |\del u|_{p+1,k+1} 
+ T_{p,k}^\textbf{hier}[H,u] + T_{p,k}^\textbf{easy}[H,u] + T_{p,k}^\textbf{super}[H,u]
\quad \text{ in } \Mscr^{\near}_{\ell, [s_0,s_1]}, 
\end{equation} 
where the notation in \eqref{eq4-12-02-2020-second} is used. 
\end{proposition}

\begin{proof}[Sketch of the proof of Proposition~\ref{prop1-22-05-2020-second}] From \cite[Proposition~\ref{prop1-22-05-2020}]{PLF-YM-main} (and by 
recalling \eqref{eq3-28-12-2020} therein) we find
\begin{equation}\label{eq3-28-12-2020-new}
\frac{1+|r-t|}{r} |\del\del u|_{p,k} \lesssim |\Boxt_g u|_{p,k} + t^{-1} |\del u|_{p+1,k+1} 
+ \sum_{\ord(Z) \leq p \atop \rank(Z) \leq k}  \big| [Z,H^{\alpha\beta} \del_{\alpha} \del_{\beta}]u \big|,
\end{equation}
in which we control the last term thanks to the hierarchy property for commutators \eqref{eq4-12-02-2020-second}, namely
$$
\aligned
\frac{1+|r-t|}{r} |\del\del u|_{p,k} 
& \lesssim  |\Boxt_g u|_{p,k} + t^{-1} |\del u|_{p+1,k+1} + \big(|\HN^{00}| + t^{-1}|r-t| |H|\big) \, |\del\del u|_{p-1,k-1} 
\\
& \quad 
 + T_{p,k}^\textbf{hier}[H,u] + T_{p,k}^\textbf{easy}[H,u] + T_{p,k}^\textbf{super}[H,u].
\endaligned
$$
Thanks to our smallness assumptions \eqref{eq2-24-02-2022-M} and \eqref{eq1-15-02-2020-second} (for a sufficiently small $\eps_1$), the third term in the right-hand side above is absorbed in the left-hand side, and \eqref{eq3-28-12-2020-second} is established.
\end{proof}


We will also use the second part of \cite[Proposition~\ref{prop1-22-05-2020}]{PLF-YM-main}, restated in the following form. 

\begin{proposition}[Hessian for the wave equation away from the light cone]
\label{propo2-22-05-2020-new}
Suppose that, for some $\eps_1\ll \ell$, 
\begin{equation}\label{eq2-10-07-2022}
|H|\leq \eps_1 \quad \text{ in } \Mscr^{\far}_{\ell, [s_0,s_1]}. 
\end{equation}
Then, for any function $u$ the following pointwise inequality holds (where the commutator is bounded by \eqref{eq5-12-02-2020-new}): 
$$
\aligned
|\del\del u|_{p,k} & \lesssim   (1 + t \, \crochet^{-1}) \big(|\Boxt_g u|_{p,k} + t^{-1} |\del u|_{p+1,k+1}\big) 
+ \sum_{\ord(Z) \leq p\atop \rank(Z) \leq k} |[Z,H^{\alpha\beta} \del_\alpha \del_{\beta}]u|
\qquad \text{ in } \Mfar_{\ell,[s_0,s_1]}.
\endaligned
$$
\end{proposition}


\paragraph{Estimates on the Hessian of the metric perturbation.}

The arguments in \cite[Section~\ref{section---12}]{PLF-YM-main}  somewhat simplify since the coupling between the commutator and the Hessian contributions is no longer required at the leading order. For clarity in the presentation, we treat first the remaining terms in \eqref{eq3-28-12-2020-second}. 

\begin{lemma} Using the notation \eqref{eq4b-12-02-2020-second}, for all $p\leq N-4$ the metric perturbation satisfies 
\begin{equation}\label{eq1-29-05-2022}
\aligned
& |T^\textbf{hier}[H,u]|_{p,k} + |T^\easy[H,u]|_{p,k} +|T^{\super}[H,u]|_{p,k}
\\
& \lesssim  \big((\epss + C_1\eps) + r\crochet^{-1}|Lh^{\N00}|_{k-1}\big)\frac{\crochet}{r}|\del\del u|_{p-1,k-1}
  + \delta^{-1}(\epss + C_1\eps)^2\la r\ra^{-2}\crochet^{-\kappa}s^{2\delta} 
\quad \text{ in } \Mnear_{\ell,[s_0,s_1]}, 
\endaligned
\end{equation}  
while 
\begin{equation}\label{eq3-10-07-2022}
|[Z,H^{\alpha\beta} \del_\alpha \del_{\beta}]u|\lesssim (\epss + C_1\eps)|\del\del u|_{p,k},
\qquad 
 \ord(Z)\leq p,\ \rank(Z)\leq k
\quad \text{ in } \Mfar_{\ell,[s_0,s_1]}. 
\end{equation}
\end{lemma} 

\begin{proof} The result is analogous to the one in~\cite[Proposition \ref{Proposition12.1}]{PLF-YM-main}. In view of \eqref{equa-31-12-20-new}, \eqref{eq1-26-05-2021-new}, \eqref{eq7-15-05-2020-new} and \eqref{eq1-17-07-2020-second}, we have 
\begin{equation}\label{eq1-30-05-2020-new}
\aligned
|\HN^{00}|_k 
& \lesssim |h^{\N00}|_k,
&&
|\del \HN^{00}|_{N-3} 
\lesssim \delta^{-1}
\big(\epss + C_1\eps\big) r^{-1-\kappa}s^{\delta},
\\ 
|H|_{N-3}
&  
\lesssim { \delta^{-1}}(\epss + C_1 \eps) r^{-\kappa}s^{\delta},
\qquad \qquad 
&& 
|\del H|_{N-3}
\lesssim (\epss + C_1\eps)r^{-1}\crochet^{-\kappa}s^{\delta}.
\endaligned
\end{equation}
Observing that 
$$
|\del\del u|_{N-4} \lesssim |\del u|_{N-3} \lesssim (\epss + C_1\eps) r^{-1}\crochet^{-\kappa}s^{\delta}
$$
and substituting these bounds in \eqref{eq4b-12-02-2020-second}, we obtain the desired result \eqref{eq1-29-05-2022}. The estimate \eqref{eq3-10-07-2022} is established in the same manner and we omit the details.
\end{proof}

\begin{proposition}\label{prop1-04-06-2022} 
There exists a constant $\eps_s>0$, determined by $N$, such that provided 
\begin{equation}\label{eq1-30-05-2022}
|h^{\N00}|_{N-4}\leq \eps_sr^{-1}
\quad \text{ in } \MME_{[s_0,s_1]}, 
\end{equation}
then the metric perturbation satisfies 
\begin{equation}\label{eq2-04-06-2022}
\frac{\crochet}{r}|\del\del u|_{N-4} + |\del\delsN u|_{N-4} \lesssim (\epss + C_1\eps) r^{-2}\crochet^{-\kappa}s^{2\delta}
\quad \text{ in }\Mnear_{\ell,[s_0,s_1]},
\end{equation}
together with  
\begin{equation}\label{eq1-15-08-2021-new}
|\del\del u|_{N-4} \lesssim \ell^{-1}(\epss + C_1\eps) t^{-1}r^{-1} \crochet^{-\kappa}s^{2\delta}
\quad \text{ in }\Mfar_{\ell,[s_0,s_1]}. 
\end{equation}
\end{proposition} 

\begin{proof} We recall \cite[Lemma \ref{lemma-123}]{PLF-YM-main}, which we apply with the choice $\lambda$ replaced by $\kappa$ which is possible since $\kappa < 1$. We obtain
\begin{equation}
|\Boxt_g u|_{p,k} \lesssim (\epss + C_1\eps) \, r^{-2}\crochet^{-\kappa}s^{2\delta}
+ \epss r^{-3} 
\qquad \text{ in } \MME_{[s_0,s_1]},
\quad p\leq N-4. 
\end{equation}  
Then we substitute \eqref{eq1-29-05-2022} and \eqref{eq10-02-05-2020-new} (applied in $t^{-1}|\del u|_{p+1,k+1}$) into \eqref{eq3-28-12-2020-second}, and we obtain the bound for $|\del\del u|_{p,k}$, provided $(\epss + C_1\eps) + r\crochet^{-1}|Lh^{\N00}|_{k-1}$ is sufficiently small. The bound for $|\del\delsN u|_{p,k}$ is obtained thanks to~\cite[\eqref{eq1-06-02-2020}]{PLF-YM-main}. For \eqref{eq1-15-08-2021-new},  we apply \eqref{eq3-10-07-2022} and Proposition~\ref{propo2-22-05-2020-new} with $(\epss + C_1\eps)$ sufficiently small.
\end{proof}


\paragraph{Hessian of the null component.}

We also need the following pointwise estimate which is a direct application of the wave gauge condition \eqref{eq:wcooE}. 

\begin{lemma} 
\label{eq7-08-12-2020-new}
As a consequence of \eqref{eq2-04-06-2022} hold and the wave gauge condition one has   
$$
|\del\del h^{\N00}|_{N-4} \lesssim (\epss + C_1\eps)r^{-1-\kappa}\crochet^{-\kappa}s^{2\delta}
\qquad \text{ in }\Mnear_{[s_0,s_1]}.
$$
\end{lemma}

\begin{proof} We give only an outline the argument. Recalling~\cite[\eqref{lem2-08-12-2020}]{PLF-YM-main}, we have 
$$
\aligned
|\del_t\del_t h^{\N 00}|_{p,k}
& \lesssim  |\del\delsN h|_{p,k} + r^{-1}| \del h|_{p,k} 
+ \hskip-.3cm  \sum_{p_1+p_2=p\atop k_1+k_2=k}   \hskip-.1cm   
\big( |h|_{p_1,k_1}    |\del\del h|_{p_2,k_2} 
+ |\del h|_{p_1,k_1}|\del h|_{p_2,k_2} \big)
+  r^{-1} \hskip-.3cm \sum_{p_1+p_2=p\atop k_1+k_2=k}|\del h|_{p_1,k_1}|h|_{p_2,k_2}.
\endaligned
$$
By substituting \eqref{eq10-02-05-2020-new}, \eqref{eq1-09-05-2021-new}, and \eqref{eq2-04-06-2022}, the desired result holds under the condition $\delta^{-1}(\epss+C_1\eps)\lesssim 1$.
\end{proof}


\subsection{Asymptotically harmonic decay of the null metric component} 

\paragraph{Objective.}

We now establish the following result which is similar to~\cite[Proposition \ref{prop1-14-08-2021}]{PLF-YM-main}. Importantly, our bound here is {\sl sharper} than the one derived in  \cite{PLF-YM-main}, which was only concerned with a {\sl sub-harmonic decay.} 

\begin{proposition}\label{prop1-10-06-2022}
Under the bootstrap assumptions within $[s_0, s_1]$, 
 one has 
\begin{equation}\label{eq5-10-06-2022}
|h^{\N00}|_{N-4} \lesssim (\epss +C_1\eps)\la r\ra^{-1}
\quad \text{ in } \MME_{[s_0,s_1]}, 
\end{equation}
\begin{equation}\label{eq5-11-06-2022}
\inf_{\Mscr^\near_{\ell, [s_0,s_1]}} (- r \, \hN{}^{00} ) 
\geq  {1\over 2} \epss. 
\end{equation}
\end{proposition}

\paragraph{Pointwise estimates.}

An important ingredient of the method is the derivation of pointwise estimates, as now stated. Given a triple of data $(f, u_0, u_1)$ we 
consider the solution $u=u(t,x)$ to  
\begin{equation}\label{eq7-28-12-2020-one}
\Box u = f, \qquad  
u(1,x) = u_0(x),\qquad  
\del_t u(1,x) = u_1(x), 
\qquad x \in \RR^3, 
\end{equation}
and we use the notation 
\begin{equation}
\aligned 
& u = \Box^{-1}[u_0,u_1,f],
\qquad 
\Box^{-1}_\init[u_0,u_1] := \Box^{-1}[u_0,u_1,0],
\qquad 
\Box^{-1}_\source[f] := \Box^{-1}[0,0,f]. 
\endaligned
\end{equation}
We consider the effect of a  decaying source, represented by the operator $\Box^{-1}_{\source}$.  
We denote by $\Lambda_{t,x} := \big\{ (\tau,y) \big/ \, t-\tau = |x-y|, \, 1\leq  \tau\leq t \big\}$ the past light cone associated with a point $(t,x)$. In the present paper, we will rely especially on the critical case below (and refer to \cite{PLF-YM-main} for additional details). 

\begin{proposition}[Sup-norm estimates for the wave equation]
\label{Linfini wave}
Consider a function $f$ satisfying, for some $\alpha_1, \alpha_2, \alpha_3$, 
\begin{equation}\label{eq1-27-12-2020-one}
|f(\tau,y)| \lesssim C_1 \, 
\tau^{\alpha_1}(\tau + |y| )^{\alpha_2} \big( 1 + | \tau - |y| | \big)^{\alpha_3}, 
\qquad 
(\tau,y) \in \Lambda_{t,x}. 
\end{equation} 

\begin{subequations}

\noindent {\bf Case 1 (typical).} When $\alpha_1 = -1+\upsilon$ and $\alpha_2 = -1-\nu$ and $\alpha_3=-1+\mu$ 
for some  
$\upsilon + \mu < \nu$ and $0<\mu,\nu,\upsilon\leq 1/2$,
one has 
\begin{equation}\label{eq5-24-12-2020-one}
|\Box^{-1}_\source[f](t,x)|
\lesssim
C_1 \, \big(\upsilon^{-1} + \mu^{-1} + |\mu-\nu|^{-1}\big) \, |\upsilon + \mu-\nu|^{-1} (t+r)^{-1}.
\end{equation} 

\vskip.15cm

\noindent{\bf Case 2 (sub-critical).} When $\alpha_1=0$ and $\alpha_2 = -2-\nu$ and $\alpha_3 = -1+\mu$ 
for some $0 < \nu, \mu \leq 1/2$, one has 
\begin{equation}\label{Linfini wave ineq-one}
|\Box^{-1}_\source[f](t,x)|
\lesssim C_1 \, 
\begin{cases} 
\mu^{-1}|\mu-\nu|^{-1} (t+r)^{-1}t^{\mu-\nu},\qquad  
&\mu>\nu,
\\
\mu^{-1} (t+r)^{-1} \ln (t+1),\quad  & \mu=\nu, 
\\
\mu^{-1}|\mu-\nu|^{-1} (t+r)^{-1}, & \mu<\nu.
\end{cases} 
\end{equation} 

\vskip.15cm

\noindent{\bf Case 3 (critical).} When $\alpha_1 = 0$ and $\alpha_2 = -2$ and $\alpha_3= -1-\mu$ for some 
$\mu\in (0,1/2)$, one has  
\begin{equation}\label{eq1-10-01-2021-one}
|\Box^{-1}_\source[f](t,x)| \lesssim C_1\,\mu^{-1} (t+r)^{-1}\Big(1 + \crochet^{-\mu}\ln\Big(\frac{t}{\crochet}\Big)\Big).
\end{equation} 

\vskip.15cm

\noindent{\bf  Case 4 (super-critical).} When $\alpha_1 = 0$ and $\alpha_2 = -2+\nu$ and $\alpha_3 = -1-\mu$ for some 
$0<\nu< \mu < 1/2$, one has  
\begin{equation}\label{eq1-10-01-2021-case4-one}
|\Box^{-1}_\source[f](t,x)| \lesssim
C_1 \, \big(|\mu-\nu|^{-1} + \mu^{-1}\nu^{-1}\crochet^{-\mu}t^{\nu}\big) (t+r)^{-1}.
\end{equation}  
\end{subequations}
\end{proposition}


\paragraph{Controlling $|\Box u|_{p,k}$.}

We recall the following decomposition introduced in \cite[Section~\ref{section---13}]{PLF-YM-main}:
\begin{equation}\label{equa-def-good-bad-new} 
\aligned
\Mgood_{[s_0,s_1]} 
& := \big\{ r\geq t-1 + (\epss + C_1\eps)t^{1/2} \big\}\cap \MME_{[s_0,s_1]},
\quad
\\
\Mbad_{[s_0,s_1]} 
& := \big\{ t-1\leq r\leq t-1 + (\epss + C_1\eps)t^{1/2} \big\}\cap \MME_{[s_0,s_1]}.
\endaligned
\end{equation}
The formulation \cite[\eqref{eq1-19-04-2021}--\eqref{eq3-19-04-2021}]{PLF-YM-main}, and specifically 
\cite[\eqref{eq4-19-04-2021}]{PLF-YM-main},
can be restated as
\begin{equation}\label{eq4-11-06-2022}
|\Abb^{\alpha\beta}[h]|_{N-2}
\lesssim 
\delta^{-2}(\epss + C_1\eps)^2r^{-2\kappa}s^{2\delta}. 
\end{equation}
Then the statement in \cite[\eqref{eq3-19-04-2021}]{PLF-YM-main} still holds. Furthermore, thanks to  \eqref{equa-31-12-20-new} and \eqref{eq3-09-05-2021-new}, we have 
$$
r|\Xi^{\N00}|_{N-4} + r|u_{\init}^{\N00}|_{N-4} \lesssim (\epss+C_1\eps). 
$$
Then we find 
\begin{equation}\label{eq4-10-06-2022}
(1+r+t) \, |h^{\N00}|_{p,k} \lesssim \la r\ra \, |u_{\source}|_{p,k} + (\epss + C_1\eps) ,\qquad p\leq N-4.
\end{equation}
The only issue is to estimate $r|u_{\source}|_{p,k}$ and, for this purpose, we need to bound $|\Box u|_{p,k}$. The following result is analogous to~\cite[Lemma \ref{lem1-25-04-2021}]{PLF-YM-main}.


\begin{lemma} 
\label{eq3-10-06-2022-lem}
Under the condition \eqref{eq1-30-05-2022} so that \eqref{eq2-04-06-2022} and \eqref{eq1-15-08-2021-new} hold, one has 
$$
\aligned
|\Box u|_{N-4}
& \lesssim 
\ell^{-1}\delta^{-1}(\epss + C_1\eps)^2t^{-1+(3/2)\delta}r^{-1-\kappa}\crochet^{-1+(1-\kappa)} 
 + \delta^{-1}(\epss + C_1\eps)^2r^{-2-3\delta}\crochet^{-1+\delta}
\quad 
 \text{ in } \Mfar_{\ell,[s_0,s_1]}
\endaligned
$$
and 
$$
\aligned
|\Box u|_{N-4}
& \lesssim (\epss + C_1\eps) r^{-2+\delta}\crochet^{-1-\kappa}(r \, |h^{\N00}|_{N-4}) 
+  \delta^{-1}(\eps + C_1\eps)^2r^{-2-3\delta}\crochet^{-1+\delta}
\quad \text{ in } \Mnear_{\ell,[s_0,s_1]}. 
\endaligned
$$
\end{lemma}

\begin{proof}
As in the proof of \cite[Lemma \ref{lem1-25-04-2021}]{PLF-YM-main}, we need to establish pointwise bounds on the following quantities:
\begin{equation}\label{eq3-23-04-2021-new}
\aligned
&
|h^{\mu\nu}\del_{\mu}\del_{\nu}u|_k,
\qquad |u^{\mu\nu}\del_{\mu}\del_{\nu}h^{\star}|_k,
\qquad 
|\Pbb^{\star}[u]|_k,
\qquad |\Qbb^{\star}[u]|_k,
\\
&
\qquad |\Ibb^{\star}[u]|_k,
\qquad |\Tbb[\phi]|_k,
\qquad
|^{(w)}R^{\star}_{\alpha\beta}|_k.
\endaligned
\end{equation}
The estimates established in the general theory are sufficient for all but the first term and, for convenience, we recall the previous bounds as follows: 
\begin{equation}\label{eq3-04-06-2022}
\aligned
&|\Pbb^{\star}[u]|_{N-4} + |\Qbb^{\star}[u]|_{N-4}
\lesssim (\epss + C_1\eps)^{2-4\delta} r^{-2-3\delta}\crochet^{-1+\delta},
\\
&|\Ibb^{\star}[u]|_{N-4} + |u^{\mu\nu}\del_{\mu}\del_{\nu} h^{\star}|_{N-4}
\lesssim  \delta^{-1}(\epss+C_1\eps)^2r^{-2-3\delta}\crochet^{-1+\delta},
\\
&|\Tbb[\phi]|_{N-4} \lesssim (\epss + C_1\eps)^2 \, r^{-2-3\delta}\crochet^{-1+\delta},
\endaligned
\end{equation}
and  $|^{(w)}R^{\star}_{\alpha\beta}|_k$ is bounded by \eqref{eq4-09-05-2021-new}, which is stronger than the inequality \cite[\eqref{eq5-25-04-2021}]{PLF-YM-main} applied in the proof of \cite[Lemma \ref{lem1-25-04-2021}]{PLF-YM-main}. We thus focus on the first term $|h^{\mu\nu}\del_{\mu}\del_{\nu}u|_k$ and, first of all, we have  
\begin{equation}\label{eq4-04-06-2022}
|h^{\mu\nu}\del_{\mu}\del_{\nu} u_{\alpha\beta}|_k\lesssim |h^{\N00}\del_t\del_t u|_k 
+ | H \, \del\delsN u|_k + r^{-1}| H \, \del u|_{N-4}.
\end{equation}
The second term is bounded as, thanks to \eqref{eq7-15-05-2020-new} and \eqref{eq2-04-06-2022} and provided $\kappa-1/2\geq (3/2)\delta$, 
$$
|H\del\delsN u|_{N-4} \lesssim \delta^{-1}(\epss + C_1\eps)^2r^{-3}\crochet^{1-2\kappa}s^{3\delta} \lesssim
\delta^{-1}(\epss + C_1\eps)^2r^{-2-3\delta}\crochet^{-1+\delta}
\qquad
\text{in } \Mnear_{[s_0,s_1]}.
$$
On the other hand, to handle the ``far'' region we rely on  \eqref{eq1-15-08-2021-new} in combination with \eqref{eq1-30-05-2020-new} which leads us to the relevant estimate that is stated in Lemma~\ref{eq3-10-06-2022-lem}.

The last term in the right-hand side of \eqref{eq4-04-06-2022} is trivial thanks to the decreasing factor $r^{-1}$:
$$
r^{-1}| H \, \del u|_{N-4} \lesssim \delta^{-1}(\epss + C_1\eps)^2r^{-2-\kappa}\crochet^{-\kappa}s^{\delta}
\lesssim \delta^{-1}(\epss + C_1\eps)^2r^{-2-3\delta}\crochet^{-1+\delta},
$$
provided $\kappa-1/2\geq (5/4)\delta$. For the first term in the right-hand side \eqref{eq4-04-06-2022}, we first 
apply \cite[\eqref{eq1-15-08-2021}]{PLF-YM-main} and \eqref{eq1-09-05-2021-new}
$$
|h^{\N00}\del_t\del_t u|_{N-4} \lesssim \ell^{-1}\delta^{-1}(\epss + C_1\eps)^2 t^{-1+(3/2)\delta}r^{-1-\kappa}\crochet^{-1+(1-\kappa)}
\qquad
\text{ in } \Mfar_{[s_0,s_1]}.
$$
We have arrived at the first inequality. On the other hand, we rely on \eqref{eq1-30-05-2022} and apply \eqref{eq2-04-06-2022} and obtain
$$
|h^{\N00}\del_t\del_t u|_{N-4} \lesssim  (\epss + C_1\eps) r^{-2+\delta}\crochet^{-1-\kappa}\big(r \, |h^{\N00}|_{N-4}\big)
\qquad
\text{in } \Mnear_{[s_0,s_1]}.
$$
and, in combination with the above inequalities, we arrive at the second inequality.
\end{proof}


\begin{proof}[Proof of Proposition \ref{prop1-10-06-2022}]
{\bf Step 1.} 
We consider \eqref{eq5-10-06-2022} first. 
We first establish the bound in $\Mfar_{\ell,[s_0,s_1]}$, which is achieved by a direct application of Proposition~\ref{Linfini wave} in combination with 
Lemma~\ref{eq3-10-06-2022-lem}
and \eqref{eq4-10-06-2022}. We observe that $\Mfar_{\ell,[s_0,s_1]}$ is {\sl past complete}, in the sense that for any $(t,x)\in\Mfar_{\ell,[s_0,s_1]}$ the past cone $\Lambda_{(t,x)} := \{(\tau,y)\in \Mscr_{[s_0,s_1]}||y-x|\leq t-\tau\}$ is entirely contained in $\Mfar_{\ell,[s_0,s_1]}$. Then by Proposition~\ref{Linfini wave} {\bf Case 1} ($\upsilon = (3/2)\delta$, $\mu = 1-\kappa$, $\nu = \kappa$ ) and {\bf Case 2} ($\mu = \delta$, $\nu = 3\delta$), we obtain (under the condition $\kappa-1/2\geq (5/4)\delta$)
$$
|u_{\source}|_{N-4} \lesssim \ell^{-1}\delta^{-2}(\epss+C_1\eps)^2(t+r)^{-1} + \delta^{-1}(\epss+C_1\eps)^2(t+r)^{-1}, 
$$
which leads us to 
\begin{equation}
|u_{\source}|_{N-4} \lesssim \ell^{-1}\delta^{-2}(\epss + C_1\eps)^2\la r\ra^{-1}.
\end{equation}
Then recalling \eqref{eq4-10-06-2022}, we obtain \eqref{eq5-10-06-2022} in $\Mfar_{\ell,[s_0,s_1]}$.

Then we consider the region $\Mnear_{\ell,[s_0,s_1]}$ and define
\begin{equation}\label{eq2-11-06-2022}
s^* = \sup_{s'\in[s_0,s_1]}\big\{\text{\eqref{eq1-30-05-2022} is valid in } \Mnear_{\ell,[s_0,s']}\big\}.
\end{equation}
Provided $(\epss + C_1\eps)$ is sufficiently small, we have $s^*>s_0$ and for all $(t,x)\in\Mgood_{\ell,[s_0,s^*]}$
$$
\Lambda_{(t,x)} = \{(\tau,y)\in \Mscr_{[s_0,s_1]}||y-x|\leq t-\tau\}\subset \Mgood_{[s_0,s^*]}.
$$
We also have 
$$
\crochet^{-8\delta} \lesssim (\epss+C_1\eps)^{-8\delta}r^{-4\delta}
\qquad 
\text{ in } \Mgood_{[s_0,s_1]}\cap\Mnear_{\ell,[s_0,s_1]}. 
$$
Thus by
Lemma~\ref{eq3-10-06-2022-lem}, in $\Mgood_{[s_0,s^*]}$,
\begin{equation}\label{eq7-11-06-2022}
\aligned
|\Box u|_{N-4}
& \lesssim \ell^{-1}\delta^{-1}(\epss + C_1\eps)^2t^{-1+(3/2)\delta}r^{-1-\kappa}\crochet^{-1+(1-\kappa)}
\\
& \quad + (\delta^{-1}(\epss + C_1\eps)^2 + (\epss + C_1\eps)^{1-8\delta}r \, |h^{\N00}|_{N-4})r^{-2-3\delta}\crochet^{-1+\delta}
\endaligned
\end{equation} 
provided $7\delta\leq \kappa$. (Here $C$ is a constant determined by $N$.) Then by Proposition \ref{Linfini wave} {\bf Case 1} and {\bf Case 2}, we obtain
$$
|u_{\source}|_{N-4} \lesssim  
\ell^{-1}\delta^{-2}(\epss + C_1\eps)^2\la r\ra^{-1} + \delta^{-1}(\epss + C_1\eps)^{1-8\delta}\la r\ra^{-1}\sup_{\Mgood_{[s_0,s^*]}}\big( r \, |h^{\N00}|_{N-4}\big).
$$
Recalling \eqref{eq4-10-06-2022}, we obtain
$$
\la r\ra|h^{\N00}|_{N-4}
\lesssim 
(\epss + C_1\eps) + \ell^{-1}\delta^{-1} (\epss + C_1\eps)^2 + \delta^{-1}(\epss + C_1\eps)^{1-8\delta}\sup_{\Mgood_{[s_0,s^*]}}\Big( r \, |h^{\N00}|_{N-4}\Big),
$$
which leads us to (provided $(\epss+C_1\eps)$ is sufficiently small, as determined by $N$)
\begin{equation}\label{eq3-11-06-2022}
\la r\ra|h^{\N00}|_{N-4} \lesssim(\epss + C_1\eps) + \ell^{-1}\delta^{-2}(\epss + C_1\eps)^2
 \quad \text{in } \Mgood_{[s_0,s^*]}.
\end{equation}


Next we consider the region $\Mbad_{[s_0,s^*]}\cap\Mnear_{\ell,[s_0,s_1]}$ and we rely on a technique of integration toward the light cone. We observe that for $(t,x)\in\Mbad_{[s_0,s^*]}$, we denote by $(t,\bar{x})\in \Mgood_{[s_0,s^*]}\cap \Mbad_{[s_0,s^*]}$ with $x/|x| = \bar{x}/|\bar{x}|$. Then, for $\ord(Z)\leq N-4$ we have 
$$
Zh^{\N00}(t,x) = Zh^{\N00}(t,\bar{x}) - \int_{|x|}^{|\bar{x}|} \delEH_r Zh^{\N00}(t,\rho x/|x|)d\rho, 
$$
where $\delEH_r = (x^a/r)\delEH_a$. Then by \eqref{eq3-11-06-2022} and \eqref{eq1-17-07-2020-second}, 
$$
\aligned
\la r\ra|Zh^{\N00}|& \lesssim  (\epss + C_1\eps) + \ell^{-1}\delta^{-2}(\epss + C_1\eps)^2 
 + \la t\ra\big||\bar {x}| - |x|\big|\sup_{\MME_{[s_0,s^*]}}\{|\del Zh^{\N00}|_{N-4}\}
 \\
 & \lesssim  (\epss + C_1\eps) + \ell^{-1}\delta^{-2}(\epss + C_1\eps)^2  
 + (\epss + C_1\eps)^2\la t\ra^{1/2-\kappa+\delta/2}.
 \endaligned
$$
Then, combining with \eqref{eq3-11-06-2022} we have
\begin{equation}\label{eq6-11-06-2022}
\sup_{\MME_{[s_0,s^*]}}\big\{\la r\ra|h^{\N00}|_{N-4}\big\} \lesssim (\epss + C_1\eps) + \ell^{-1}\delta^{-2}(\epss + C_1\eps)^2  .
\end{equation}
Recalling \eqref{eq2-11-06-2022}, we observe that  when $(\epss + C_1\eps)$ is sufficiently small, the above bound leads us to
\begin{equation}\label{eq1-11-06-2022}
r \, |h^{\N00}|_{N-4}< \frac{1}{2}\eps_s 
\quad \text{in }\Mnear_{\ell,[s_0,s^*]}.
\end{equation}
By continuity, we conclude that $s^*=s_1$ since, otherwise, $\sup_{\Mgood_{s^*}}\Big( r \, |h^{\N00}|_{N-4}\Big) = \eps_s$ would contradict \eqref{eq1-11-06-2022}. Next, \eqref{eq5-10-06-2022} follows from \eqref{eq6-11-06-2022}, provided $\ell^{-1}\delta^{-1}(\epss + C_1\eps)\lesssim 1$. This completes the proof of \eqref{eq5-10-06-2022}.

\

{\bf Step 2.} 
We now turn to \eqref{eq5-11-06-2022}, which, in some sense, a special case of \eqref{eq5-10-06-2022} but also provides us with a key sign. We recall the following identity (established in \cite[\eqref{eq2-19-04-2021}]{PLF-YM-main})
$$
\aligned
rh^{\N00} 
& = rg^{\star}(\lbf,\lbf) + ru_{\init}(\lbf,\lbf) + r \, h^{\N00}_{\pertur} 
\\
& = r\big(g^{\star}(\lbf,\lbf) + u_{\init}(\lbf,\lbf)\big) - r\sum_{\alpha,\beta}\PsiN_{\alpha}^0\PsiN_{\beta}^0u_{\source,\alpha\beta} 
+ r\PsiN_{\alpha}^0\PsiN_{\beta}^0\Abb^{\alpha\beta}[h].
\endaligned
$$
Thanks to \eqref{eq3'-27-05-2020-initial-new} and \eqref{eq4-11-06-2022}, we find 
$
rh^{\N00}\leq -\epss + Cr|u_{\source}| + \delta^{-2}(\epss + C_1\eps)^2r^{1-2\kappa+\delta}.
$
Substituting \eqref{eq5-10-06-2022} in \eqref{eq7-11-06-2022} we obtain
\begin{equation}
|\Box u|\lesssim 
\ell^{-1}\delta^{-1}(\epss + C_1\eps)^2t^{-1+(3/2)\delta}r^{-1-\kappa}\crochet^{-1+(1-\kappa)} 
+ \delta^{-1}(\epss + C_1\eps)^{2-8\delta}r^{-2-3\delta}\crochet^{-1+\delta}
\quad \text{in } \Mgood_{[s_0,s_1]}.
\end{equation}
By Proposition \ref{Linfini wave} ({\bf Case 1} and {\bf Case 2}), we obtain
$$
\la r\ra|u_{\source}|\lesssim \ell^{-1}\delta^{-2}(\epss + C_1\eps)^{2-8\delta}
\qquad\text{in } \Mgood_{[s_0,s_1]},
$$
which leads us to (for some constant $C>0$) 
\begin{equation}\label{eq8-11-06-2022}
rh^{\N00}\leq -\epss + C\ell^{-1}\delta^{-2}(\epss + C_1\eps)^{2-8\delta}
\qquad\text{in } \Mgood_{[s_0,s_1]}.
\end{equation}
Again, we perform an integration towards the light-cone in order to cover  $\Mbad_{[s_0,s_1]}$: 
$$
\aligned
h^{\N00}(t,x) 
& = h^{\N00}(t,\bar{x}) - \int_{|x|}^{|\bar{x}|}\delEH_r h^{N00}(t,\rho x/|x|) \, d\rho
   \leq  h^{\N00}(t,\bar{x}) + C(\epss+C_1\eps)^{2}t^{-1} 
\\
& \leq -\epss |\bar{x}|^{-1} +  C\ell^{-1}\delta^{-2}(\epss + C_1\eps)^{2-8\delta}|\bar{x}|^{-1} + C(\epss+C_1\eps)^{2}t^{-1}
\\
& \leq  -\epss r^{-1} + C\big((\epss + C_1\eps)^{2} + \ell^{-1}\delta^{-2}(\epss+C_1\eps)^{2-8\delta}\big)r^{-1}
\hskip2.cm 
\text{ in } \Mbad_{[s_0,s_1]}. 
\endaligned
$$
We recall that $\eps\leq C_\star \epss$ with $C_\star$ a given constant, thus when $\epss$ sufficiently small, the above inequality together with \eqref{eq8-11-06-2022} leads us to \eqref{eq5-11-06-2022}.
\end{proof}


\subsection{Sharp decay for the gradient of good metric components} 
\label{section---53}

\paragraph{Objective and strategy.}
 
 We now turn our attention to the estimates for $\del \us$. 

\begin{proposition}\label{prop1-18-06-2022}
Under the bootstrap assumption in $[s_0, s_1]$ and as a consequence of  
\eqref{eq5-10-06-2022},   one has in $\Mnear_{\ell, [s_0,s_1]}$
\begin{equation}\label{eq11-01-31-2021-new}
|\del \us|_{N-4,k} \lesssim \big(\ell^{-\delta/2} +\delta^{-2}\big) (\epss + C_1\eps) \la r\ra^{-1}\crochet^{-1/2-\delta/2}\big(1+(\epss + C_1\eps) (\ln \la r\ra)^k\big),\qquad 0\leq k\leq N-4,
\end{equation}
\begin{equation}\label{eq8-24-03-2021-new} 
|\del\del \us|_{N-5,k}
\lesssim  (\ell^{-\delta}+\delta^{-2})(\epss + C_1\eps) \la r \ra^{-1}\crochet^{-1-\delta}\big(1+(\epss + C_1\eps) (\ln \la r\ra)^k\big),\qquad 0\leq k\leq N-5.
\end{equation}
\end{proposition}

Our proof relies on the following result, which was pointed out first in \cite{LR1} and formulated as follows in \cite{PLF-YM-main}.

\begin{proposition}[Weighted pointwise estimate in the Euclidean-merging domain] 
\label{prop1-23-07-2020-new}
Consider a metric $g^{\alpha\beta} = g_\Mink^{\alpha\beta} + H^{\alpha\beta}$ defined in $\Mscr^{\near}_{\ell, [s_0,s_1]}$ and satisfying
\begin{equation}\label{eq1'-10-01-2021-new}
|\HN^{00}|\ll 1,
\qquad
\HN^{00} \leq 0 \quad \text{ in } \Mscr^{\near}_{\ell, [s_0,s_1]}. 
\end{equation}
Given any $\rho \geq 0$, for any function $u$ defined $\Mscr^{\near}_{\ell, [s_0,s_1]}$ one has
$$
\aligned
& \crochet^\rho |(\del_t-\del_r) (ru)|(t,x)
\\
& \lesssim
\sup_{\Omega^{\ell}_{s_0,s_1}} (r-t+2)^\rho \big( r \, |\del u| + |u|\big) 
+  \int_{t_0}^t  \crochet(\tau,r)^\rho \, r
\Big(
r^{-1}|\delsN u|_{1,1} +  |H | \, | \del\delsN u| + r^{-1}|H | \, | \del u| 
+ |\Boxt_g u| \Big)\big|_{\varphi_{t,x}(\tau)} d\tau, 
\endaligned
$$
in which the supremum is taken over the set $\Omega_{s_0,s_1}^{\ell} = \Lscr_{\ell, [s_0,s_1]}\cup \Mnear_{\ell,s_0}$. 
\end{proposition}


Our strategy is to apply the above estimate to the wave equations satisfied by the components $Z\us^{\N}_{\alpha\beta}$. More precisely, we observe that
\begin{equation}
\Boxt_g u^{\N}_{\alpha\beta} = L_{1\alpha\beta} + L_{2\alpha\beta} + S_{1\alpha\beta},
\end{equation}
where $u^{\N}_{\alpha\beta} := \PsiN^{\alpha'}_{\alpha}\PsiN^{\beta'}_{\beta}u_{\alpha'\beta'}$ and
$$
\aligned
& L_{1\alpha\beta} := u_{\alpha'\beta'}\Boxt_g\big(\PhiN_{\alpha}^{\alpha'}\PhiN_{\beta}^{\beta'}\big), 
\\
& L_{2\alpha\beta} := g^{\mu\nu}\del_{\mu}\big(\PhiN_{\alpha}^{\alpha'}\PhiN_{\beta}^{\beta'}\big)\del_{\nu}u_{\alpha'\beta'}, 
\qquad
&& S_{1\alpha\beta} := \PhiN_{\alpha}^{\alpha'}\PhiN_{\beta}^{\beta'}\Boxt_gu_{\alpha'\beta'}. 
\endaligned
$$
For any $Z = \del^IL^J\Omega^K$ with $\ord(Z) \leq N-4$, we then write
\begin{equation}\label{eq10-25-07-2020}
\aligned
& \Boxt_g \big( Z u^{\N}_{\alpha\beta} \big) = Z L_{1\alpha\beta} + Z L_{2\alpha\beta} + Z S_{1\alpha\beta} + S_{2\alpha\beta}, 
\qquad 
&& S_{2\alpha\beta} = S_{2\alpha\beta}[u]:= - [Z, h^{\mu\nu}\del_\mu\del_\nu] u^{\N}_{\alpha\beta},
\endaligned
\end{equation}
\begin{equation}\label{eq1-24-03-2021}
\Boxt_g(Z\del_t u^{\N}_{\alpha\beta}) = Z \del_t\big(L_{1\alpha\beta} + L_{2\alpha\beta} + S_{1\alpha\beta}\big) + S'_{2\alpha\beta},
\qquad 
S'_{2\alpha\beta} = S'_{2\alpha\beta}[u]:= - [Z\del_t,h^{\mu\nu}\del_{\mu}\del_{\nu}]u^{\N}_{\alpha\beta}.
\end{equation}
In order to eventually control  $Z\us^{\N}_{\alpha\beta}$, we need to estimate the source terms in the right-hand side, which is our main task in the rest of this Section \ref{section---53}.


\paragraph{Estimates on source terms.} 

We are in a position to establish the following result.

\begin{proposition}
\label{proposition-label-144-new}
As a consequence of \eqref{eq5-10-06-2022}, for all $\ord(Z)\leq N-4$ and $\rank(Z) = k\leq N-4$  one has 
\begin{subequations}\label{eq6-14-08-2021-new}
\begin{equation}\label{eq9-31-01-2021-new}
\aligned
\crochet^{1/2+\delta/2}r|\Boxt_gZ\us^{\N}|
\lesssim  
(\epss + C_1\eps)r^{-1} \big(r\crochet^{1/2+\delta/2}|\del \us^{\N}|_{N-4,k-1}\big) + \delta^{-1} (\epss + C_1\eps) r^{-1-\delta}, 
\endaligned
\end{equation}
while for all $\ord(Z) \leq N-5$ and $\rank(Z) = k\leq N-5$ one has 
\begin{equation}\label{eq2-24-03-2021-new}
r \, \crochet^{1+\delta}|\Boxt_gZ\del_t \us^{\N}|
\lesssim (\epss + C_1\eps)r^{-1}\big(r\crochet^{\kappa}|\del\del \us^{\N}|_{N-5,k-1}\big)
+ \delta^{-1}(\epss + C_1\eps) r^{-1-\delta}.
\end{equation}
\end{subequations}
\end{proposition}

\begin{proof}
The proof is  similar to that in \cite[Proposition \ref{proposition-label-144}]{PLF-YM-main}. The only difference concerns the estimates for $S_{2\mu\nu}[u]$ and $S'_{2\mu\nu}[u]$, while the remaining terms are bounded by $\delta^{-1}(\epss + C_1\eps) r^{-1-\delta}$. Recalling  Proposition \ref{prop1-12-02-2020-new} together with \eqref{eq2-04-06-2022} and \eqref{eq5-10-06-2022} (which guarantees \eqref{eq1-30-05-2022}), for $(\alpha,\beta)\neq (0,0)$ we find 
\begin{equation}\label{eq1-18-06-2022}
\aligned
|S_{2\alpha\beta}|_{p,k}& \lesssim  (\epss + C_1\eps)\la r\ra^{-1}|\del\del \us^{\N}|_{p-1,k-1}
 + \delta^{-1}(\epss + C_1\eps)^2\la r\ra^{-2-\kappa}\crochet^{-\kappa}s^{3\delta},\quad && k\leq p\leq N-4
\\
|S'_{2\alpha\beta}|_{p,k}& \lesssim (\epss + C_1\eps)\la r\ra^{-1}|\del\del\del_t \us^{\N}|_{p-1,k-1}
+ \delta^{-1}(\epss + C_1\eps)^2\la r\ra^{-2-\kappa}\crochet^{-\kappa}s^{3\delta},\quad && k\leq p\leq N-5.
\endaligned 
\end{equation}
Hence, \eqref{eq9-31-01-2021-new} and \eqref{eq2-24-03-2021-new} are established.
\end{proof}


\paragraph{Proof of Proposition~\ref{prop1-18-06-2022}.}

Relying on the notation in Proposition~\ref{prop1-23-07-2020-new}, we observe that by \eqref{eq10-02-05-2020-new} and \eqref{eq2-04-06-2022} expressed on the relevant cone $r = t(1-\ell)$ or 
$r-t\simeq \ell \, r$ 
\begin{equation}
\aligned
\sup_{\Omega_{s_0,s_1}^{\ell}} \crochet^{1/2+\delta/2}\big(  r \, |\del Z\us^{\N}| +  |Z\us^{\N}|\big) 
& \lesssim  \big(\ell^{-\delta/2} + \delta^{-1}\big)(\epss + C_1\eps),
\\
\sup_{\Omega^{\ell}_{s_0,s_1}} \crochet^{1 + \delta}\big(  r \, |\del Z\del_t\us^{\N}| +  |Z\del_t\us^{\N}|\big) 
& \lesssim  \ell^{-\delta}(\epss + C_1\eps). 
\endaligned
\end{equation}
Then for $(t,x)\in \Mnear_{\ell, [s_0,s_1]}$ and $\ord(Z) = N-4, \rank(Z) = k$, after observing that \eqref{eq1'-10-01-2021-new} are guaranteed by \eqref{eq3-27-05-2020-new} and \eqref{eq1-30-05-2020-new}, we obtain 
$$
\aligned
&\crochet^{1/2+\delta/2}|(\del_t-\del_r)(rZ\us^{\N})(t,x)|
\lesssim    (\ell^{-\delta/2} + \delta^{-1})(\epss + C_1\eps)
+ \int_{t_0}^t r\la r-\tau\ra^{1/2+\delta/2} |\Boxt_g Z\us^{\N}|_{\varphi_{t,x}(\tau)}d\tau
\\
&  \hskip5.cm +  \int_{t_0}^t 
\la r-\tau\ra^{1/2+\delta/2}\Big(|\delsN \us^{\N}|_{N-3} + r|H||\del\delsN Z\us^{\N}| + |H||\del \us^{\N}|_{N-4}\Big)\Big|_{\varphi_{t,x}(\tau)}d\tau
\\
& \lesssim    
(\ell^{-\delta/2} +  \delta^{-1})(\epss + C_1\eps)
+ \int_{t_0}^t(\epss + C_1\eps) r^{-1} \crochet^{1/2+\delta/2}r|\del \us^{\N}|_{N-4,k-1}\Big|_{\varphi_{t,x}(\tau)}d\tau 
+ \delta^{-1}(\epss + C_1\eps)\int_{t_0}^t \tau^{-1-\delta}d\tau
\\
& \lesssim    \big(\ell^{-\delta/2} +\delta^{-2}\big) (\epss + C_1\eps)  
+ (\epss + C_1\eps)\int_{t_0}^t\tau^{-1}\crochet^{1/2+\delta/2}r|\del \us^{\N}|_{N-4,k-1}\Big|_{\varphi_{t,x}(\tau)}d\tau, 
\endaligned
$$
where we used\footnote{$\us^{\N}$ is a finite linear combination of $u$ with homogeneous coefficients of degree zero, thus the bounds can be applied.} \eqref{eq1-18-06-2022}, \eqref{eq10-02-05-2020-new}, and \eqref{eq7a-03-05-2020-new}.  On the other hand, observe that
$
|Z u|\lesssim  \delta^{-1} (\epss + C_1\eps) r^{-1}\crochet^{1-\kappa}s^{\delta}
\lesssim \delta^{-1}(\epss + C_1\eps) r^{\delta/2 - \kappa} \lesssim (\epss + C_1\eps),
$ 
which we apply with $u$ replaced by $\us^{\N}$.  This leads us to
\begin{equation}\label{eq1-29-11-2020-new}
\aligned
r \, \big|(\del_t-\del_r)Z \us^{\N}_{\alpha\beta} \, \big|
\lesssim (\epss + C_1\eps) + |(\del_t - \del_r)(r Z \us^{\N})|. 
\endaligned
\end{equation}
Now recalling \eqref{eq10-02-05-2020-new}, for all $\ord(Z)\leq N-4$ we have $r \, |\delsN Z u|\lesssim C_1\eps \, r^{-\kappa + \delta/2} \lesssim (\epss + C_1\eps)$. Recalling the identities $2\del_t = (x^a/r)\delsN_a + (\del_t-\del_r)$ and $\del_a = \delsN_a - (x^a/r)\del_t$, together with the ordering lemma \cite[Lemma~\ref{lem 2 high-order}]{PLF-YM-main}, we obtain the following bound in $\Mnear_{\ell,[s_0,s_1]}$: 
$$ 
\aligned
r \, \crochet^{1/2+\delta/2} |\del \us^{\N} |_{N-4,k}
& \lesssim 
\big(\ell^{-\delta/2} +\delta^{-2}\big) (\epss + C_1\eps)  
+ (\epss + C_1\eps)\int_{t_0}^t
\tau^{-1} r \, \crochet^{1/2+\delta/2}
 |\del \us^{\N}|_{N-4,k-1}\Big|_{\varphi_{t,x}(\tau)}d\tau. 
\endaligned
$$
Finally, with the notation 
$\Bbf_k(t) := \sup_{\Mnear_{\ell,[s_0,s]}} \big(  r \, \crochet^{1/2+\delta/2} |\del \us^{\N} |_{N-4,k} \big)$ 
(with $t = T^{\E}(s)$), 
the above estimate reads 
\begin{equation}\label{eq2-06-07-2022}
\Bbf_k(t)\lesssim \big(\ell^{-\delta/2} +\delta^{-2}\big) (\epss + C_1\eps) 
+ (\epss + C_1\eps)\int_{t_0}^t\tau^{-1}\Bbf_{k-1}(\tau)d\tau, 
\end{equation}
in which the last term does not exist when $k=0$. Next, by  induction on $k$ varying from $k=0$ to $k=N-4$, we conclude and arrive at \eqref{eq11-01-31-2021-new}. The estimate \eqref{eq8-24-03-2021-new} is established in a similar way, and we omit the details.


\paragraph{Useful inequalities.}

For $k\geq 1$, we can relax \eqref{eq2-06-07-2022} in the form
\begin{equation}\label{eq3-06-07-2022}
\Bbf_k(t)\lesssim \big(\ell^{-\delta/2} +\delta^{-2}\big) (\epss + C_1\eps) 
+ (\epss + C_1\eps)\int_{t_0}^t\tau^{-1}\Bbf_k(\tau)d\tau. 
\end{equation}
\vskip-.3cm
\noindent By Gronwall's inequality, we obtain the slightly weaker version 
\begin{subequations}\label{eq6-06-07-2022}
\begin{equation}\label{eq4-06-07-2022}
|\del \us|_{N-4,k} \lesssim 
\begin{cases}
\big(\ell^{-\delta/2} +\delta^{-2}\big) (\epss + C_1\eps) \la r\ra^{-1+C(\epss+C_1\vep)}\crochet^{-1/2-\delta/2},\quad &1\leq k\leq N-4,
\\
\big(\ell^{-\delta/2} +\delta^{-2}\big) (\epss + C_1\eps) \la r\ra^{-1}\crochet^{-1/2-\delta/2},\quad &k=0;
\end{cases}
\end{equation}
\begin{equation}\label{eq5-06-07-2022} 
|\del\del \us|_{N-5,k} \lesssim  
\begin{cases}
(\ell^{-\delta}+\delta^{-2})(\epss + C_1\eps)\la r\ra^{-1+C(\epss+C_1\vep)}\crochet^{-1-\delta},\quad &1\leq k\leq N-5,
\\
(\ell^{-\delta}+\delta^{-2})(\epss + C_1\eps)\la r\ra^{-1}\crochet^{-1-\delta},\quad &k=0,
\end{cases}
\end{equation}
\end{subequations}


\subsection{Pointwise estimates for metric components at low order} 

\paragraph{Objective.}

We now collect the pointwise estimates obtained in the previous sections for the different metric components. The main result is stated as follows.

\begin{proposition}  
In $\Mnear_{\ell,[s_0,s_1]}$, the metric components satisfy 
\begin{equation}\label{eq5-25-03-2021-new}
|\del h|_{N-4,k} \lesssim
\begin{cases} 
\big(\ell^{-\delta/2} +\delta^{-2}\big) (\epss + C_1\eps)\crochet^{-1/2-\delta/2} \la r\ra^{-1}, \quad &k=0,
\\
\big(\ell^{-\delta/2} +\delta^{-2}\big) (\epss + C_1\eps)\crochet^{-1/2-\delta/2} \la r\ra^{-1+C(\epss + C_1\eps)},\quad& 1\leq k\leq N-4,
\end{cases}
\end{equation}
\begin{equation}\label{eq5-25-03-2021-new'}
|\del\del h|_{N-5,k} \lesssim \begin{cases} 
\big(\ell^{-\delta} +\delta^{-2}\big) (\epss + C_1\eps)\crochet^{-1-\delta} \la r\ra^{-1}, \quad &k=0,
\\
\big(\ell^{-\delta} +\delta^{-2}\big) (\epss + C_1\eps)\crochet^{-1-\delta} \la r\ra^{-1+C(\epss + C_1\eps)},\quad& 1\leq k\leq N-5.
\end{cases}
\end{equation}
Furthermore,  the source terms $u_{\source,\alpha\beta}$ defined in \cite[\eqref{eq1-19-04-2021}]{PLF-YM-main} enjoys the near-harmonic decay in the whole domain $\MME_{[s_0,s_1]}$:  
\begin{equation}\label{eq1-29-03-2021-new}
|u_{\source}|_k\lesssim  
(\epss + C_1\eps) r^{-1+C(\epss+C_1\eps)^{1/2}},\qquad 0\leq k\leq N-5. 
\end{equation}
\end{proposition}

We also have estimates for the metric with upper indices, that is, $|\del H|:=\max_{\alpha,\beta}|\del h^{\alpha\beta}|$.

\begin{corollary}
\label{coro-5-aout-2022}  
In $\Mnear_{\ell,[s_0,s_1]}$, the metric perturbation satisfies 
\begin{equation}\label{eq7-26-03-2021-new}
\aligned
|\del u|_{N-4,k} + |\del H|_{N-4,k}& \lesssim 
\begin{cases} 
\big(\ell^{-\delta/2} +\delta^{-2}\big) (\epss + C_1\eps)\crochet^{-1/2-\delta/2} \la r\ra^{-1}, \quad &k=0,
\\
\big(\ell^{-\delta/2} +\delta^{-2}\big) (\epss + C_1\eps)\crochet^{-1/2-\delta/2} \la r\ra^{-1+C(\epss + C_1\eps)},\quad& 1\leq k\leq N-4,
\end{cases}
\\
|\del\del u|_{N-5,k} + |\del\del H|_{N-5,k} & \lesssim 
\begin{cases} 
\big(\ell^{-\delta} +\delta^{-2}\big) (\epss + C_1\eps)\crochet^{-1-\delta} \la r\ra^{-1}, \quad &k=0,
\\
\big(\ell^{-\delta} +\delta^{-2}\big) (\epss + C_1\eps)\crochet^{-1-\delta} \la r\ra^{-1+C(\epss + C_1\eps)},\quad& 1\leq k\leq N-4.
\end{cases}
\endaligned
\end{equation}
Furthermore, as a consequence of \eqref{equa-31-12-20-new}, in $\MME_{[s_0,s_1]}$, one has 
\begin{equation}\label{eq3-29-03-2021-new}
|u|_{N-5} + |h_{\alpha\beta}|_{N-5} + |h^{\alpha\beta}|_{N-5}
\lesssim  (\epss + C_1\eps) \la r\ra^{-1+C(\epss+C_1\eps)^{1/2}}.
\end{equation}
\end{corollary}


\paragraph{Sharp decay bounds on gradient and Hessian.}

The following identities will be used.

\begin{lemma}\label{lem1-08-12-2020-new}
When $|h|_{[p/2]}\ll 1$, one has
$
h^{\N}_{00} = h_{00} = -\frac{1}{4}h^{\N00} + \Bbb[h]
$
with 
\begin{equation}\label{eq2-25-03-2021-new}
\big| \del_t \Bbb[h] \big|_{p,k} \lesssim  |\del_t \hs^{\N}|_{p,k} + |\del\Abb[h]|_{p,k},
\qquad
\big| \del_t\del_t\Bbb[h] \big|_{p,k} \lesssim   |\del_t\del_t \hs^{\N}|_{p,k}
+ |\del\del\Abb[h]|_{p,k}
\end{equation}
with
$$
\aligned
&|\del\Abb[h]|_{p,k} \lesssim \sum_{p_1+p_2=p\atop k_1+k_2=k}|\del h|_{p_1,k_1}|h|_{p_2,k_2},
\\
&|\del\del \Abb[h]|_{p,k} \lesssim\sum_{p_1+p_2=p\atop k_1+k_2=k}\Big(|\del\del h|_{p_1,k_1}|h|_{p_2,k_2} + |\del h|_{p_1,k_1}|\del h|_{p_2,k_2} \Big)
+
\sum_{p_1+p_2 + p_3=p\atop k_1+k_2 + k_3=k}|\del h|_{p_1,k_1}|\del h|_{p_2,k_2}|h|_{p_3,k_3}.
\endaligned
$$
\end{lemma}

\begin{proof} We only outline the proof and refer to \cite{PLF-YM-main} where we derived the algebraic identity
\begin{subequations}
\begin{equation}\label{eq1-08-12-2020-new}
-h^{\N 00} 
= h_{00} - \frac{2x^a}{r}h_{a0} + \frac{x^ax^b}{r^2}h_{ab} - \PsiN_{\alpha}^0\PsiN_{\beta}^0\Abb^{\alpha\beta}[h], 
\end{equation}
with
\begin{equation}\label{eq1-06-07-2022}
\aligned
& h_{00} = h^{\N}_{00},\quad 
\qquad
h_{a0} = \hN_{\alpha\beta}\PsiN_a^{\alpha}\PsiN_0^{\beta} = -(x^a/r)\hN_{00}
+ \sum_{(\alpha,\beta)\neq(0,0)}\PsiN_a^{\alpha}\PsiN_0^{\beta}\hs^{\N}_{\alpha\beta},
\\
& h_{ab} = \frac{x^ax^b}{r^2}h^{\N}_{00}  + \sum_{(\alpha,\beta)\neq(0,0)}\PsiN_a^{\alpha}\PsiN_b^{\beta}\hs^{\N}_{\alpha\beta}.
\endaligned
\end{equation}
\end{subequations} 
Substituting \eqref{eq1-06-07-2022} into \eqref{eq1-08-12-2020-new} and separating the linear terms from the nonlinear ones, we obtain the result.
\end{proof}


\paragraph{Proof of \eqref{eq5-25-03-2021-new}.}

Combining Lemma \ref{lem1-08-12-2020-new} with \eqref{eq1-26-05-2021-new} and \eqref{eq1-09-05-2021-new}, we obtain  
\begin{equation}\label{eq4-25-03-2021-new}
|\del \Abb |_{N-4} + |\del\del \Abb |_{N-5} \lesssim  
\delta^{-1}(\epss + C_1\eps)^2 r^{-1-\kappa}\crochet^{-\kappa}s^{2\delta}.
\end{equation}
Combining \eqref{eq1-17-07-2020-second}, Lemma \ref{eq7-08-12-2020-new}, and the condition $\delta^{-1}(\epss+C_1\eps)\lesssim 1$, we thus find
$$
\aligned
|\del_t h^{\N}_{00}|_{N-4,k}& \lesssim  |\del_t \hs^{\N}|_{N-4,k} 
+  \delta^{-1}(\epss + C_1\eps) r^{-1-\kappa}s^{2\delta},
\\
|\del_t\del_t h^{\N}_{00}|_{N-5,k}& \lesssim  |\del_t\del_t \hs^{\N}|_{N-5,k} 
+ (\epss + C_1\eps)r^{-1-\kappa}\crochet^{-\kappa}s^{2\delta}, 
\endaligned
$$
together with
$$
\aligned
|\del_t^m h|_{p,k} & \lesssim   |\del^m \hs^{\N}|_{p,k} + |\del_t^m h^{\N}_{00}|_{p,k}
 \lesssim   |\del^m \us^{\N}|_{p,k} + |\del^m h^{\star}|_{p,k} + |\del_t^m h^{\N}_{00}|_{p,k}.
\endaligned
$$
In view of Proposition \ref{prop1-18-06-2022} and \eqref{equa-31-12-20-new}, we arrive at \eqref{eq5-25-03-2021-new} and \eqref{eq5-25-03-2021-new'}.


\paragraph{Proof of \eqref{eq7-26-03-2021-new}.}

We only outline the argument. 
For the bounds on $|\del u|$ and $|\del \del u|$, we only need to observe that $h_{\alpha\beta} = h^{\star}_{\alpha\beta} + u_{\alpha\beta}$ together with \eqref{equa-31-12-20-new}. For the bounds on $|\del H|$ and $|\del\del H|$, namely $|\del h^{\alpha\beta}|$ and $|\del\del h^{\alpha\beta}|$, we recall that $-h^{\alpha\beta} = h_{\alpha\beta} + \Abb^{\alpha\beta}[h]$ and apply \eqref{eq4-25-03-2021-new}.


\paragraph{Proof of \eqref{eq1-29-03-2021-new}.}

We rely on Proposition \ref{Linfini wave}, and the proof is quite similar to that of the source in \cite[Proposition~\ref{section-15-1} ]{PLF-YM-main}, while only difference is that here we have better pointwise bounds \eqref{eq5-25-03-2021-new}, \eqref{eq5-25-03-2021-new'} and \eqref{eq7-26-03-2021-new}. More precisely, by \eqref{eq1-04-12-2020-second}, \eqref{eq11-04-06-2020-second}, \eqref{eq1-28-11-2020-second} and \eqref{eq4-09-05-2021-new},  the terms $\Ibb^{\star}[u]$, ${ u^{\mu\nu}\del_{\mu}\del_{\nu}h^{\star}}$, $\Tbb[\phi]$ and ${ ^{(w)}R^{\star}_{\alpha\beta}}$ are bounded (provided $\kappa \geq 5\delta$ and $\mu \geq 3/4 + (7/4)\delta$) 
by
$$
\delta^{-1}(\epss + C_1\eps)^2r^{-2-3\delta}\crochet^{-1+\delta}.
$$ 

We only need to improve the estimate of $|\Box u_{\alpha\beta}|_{N-5}$ in $\Mnear_{\ell,[s_0,s_1]}$. More precisely, we have
\begin{equation}\label{eq126-04-2021-new}
|\Pbb^{\star}[u]|_k + |\Qbb^{\star}[u]|_k\lesssim 
\begin{cases}
(\ell^{-\delta} + \delta^{-4})(\epss + C_1\eps)^2\crochet^{-1-\delta} r^{-2},	\quad &k=0,
\\
(\ell^{-\delta} + \delta^{-4})(\epss + C_1\eps)^2\crochet^{-1-\delta} r^{-2+C(\epss+C_1\eps)},\quad &1\leq k\leq N-4.
\end{cases}
\end{equation}
Similarly, considering the most challenging component $h^{\N00}\del_t\del_t u$ we improve the estimate of quasi-linear terms as 
\begin{equation}
|h^{\mu\nu}\del_{\mu}\del_{\nu} u|_k\lesssim 
\begin{cases}
(\ell^{-\delta} + \delta^{-2})(\epss + C_1\eps)^2 \crochet^{-1-\delta} r^{-2}, \quad & k=0,
\\
(\ell^{-\delta} + \delta^{-2})(\epss + C_1\eps)^2 \crochet^{-1-\delta} r^{-2+C(\epss + C_1\eps)}, \quad& 1\leq k\leq N-5,
\end{cases}
\end{equation}
provided $\kappa\geq 1/2 + (5/4) \delta$. We thus conclude that, in $\Mnear_{[s_0,s_1]}$, 
\begin{equation}\label{eq2-09-05-2021-new}
\aligned
|\Box u_{\alpha\beta}|_k
& \lesssim \delta^{-1}(\epss + C_1\eps)^2s^{-2-3\delta}\crochet^{-1+\delta} 
\\
& \quad 
+ \begin{cases}
(\ell^{-\delta} + \delta^{-4})(\epss + C_1\eps)^2\crochet^{-1-\delta} r^{-2},	\quad &k=0,
\\
(\ell^{-\delta} + \delta^{-4})(\epss + C_1\eps)^2\crochet^{-1-\delta} r^{-2+C(\epss+C_1\eps)},\quad &1\leq k\leq N-4.
\end{cases}
\endaligned
\end{equation}
In $\Mfar_{\ell,[s_0,s_1]}$ and for all $k\leq N-5$, we use $\crochet^{-1} \lesssim \ell^{-1} \la r\ra^{-1}$ together with \eqref{eq10-02-05-2020-new}, \eqref{eq1-09-05-2021-new} and \eqref{eq1-15-08-2021-new}, and arrive at
\begin{equation}
\aligned
|\del u \del u|_{N-5}& \lesssim  \ell^{-4\delta} (\epss + C_1\eps)^2 r^{-2-3\delta}\crochet^{-1+\delta},
\\
|h^{\alpha\beta}\del_{\alpha}\del_{\beta} u|_{N-5}& \lesssim  \ell^{-1}{\delta^{-1}}(\epss + C_1\eps)^2t^{-1+{ (3/2)}\delta}r^{-1-\kappa} \crochet^{-1+\kappa}.
\endaligned
\end{equation} 
In conclusion, we have a control of the wave operator  
\begin{equation}
\aligned
|\Box u_{\alpha\beta}|_k& \lesssim \big(\ell^{-4\delta} + \delta^{-1}\big)(\epss + C_1\eps)^2 r^{-2-3\delta}\crochet^{-1+\delta} 
+ \ell^{-1}{\delta^{-1}}(\epss + C_1\eps)^2t^{-1+{ (3/2)}\delta}r^{-1-\kappa} \crochet^{-1+\kappa}
\\
& \quad + (\ell^{-\delta} + \delta^{-4})(\epss + C_1\eps)^2\crochet^{-1-\delta} r^{-2+C(\epss+C_1\eps)^{1/2}}.
\endaligned
\end{equation}
Here for simplicity we have replaced the decay factors $r^{-1}$ and $r^{-2+C(\epss+C_1\eps)}$ by $r^{-2+C(\epss+C_1\eps)^{1/2}}$. Applying Proposition \ref{Linfini wave} for the case $k\geq 1$, we obtain
$$
\aligned
|u_{\source}|_k
& \lesssim  \underbrace{\big(\ell^{-4\delta} + \delta^{-1}\big) \delta^{-2}(\epss + C_1\eps)^2r^{-1}}_{\text{Case 2 with } \mu = \delta, \nu = 3\delta} 
+ \underbrace{\ell^{-1}\delta^{-3}(\epss + C_1\eps)^2r^{-1}}_{\text{Case 1 with } \upsilon = (3/2)\delta\atop \mu = 1-\min(\lambda,\kappa), \nu = \min(\lambda,\kappa), \nu-\mu-\upsilon\geq \delta}
\\
&\quad
+\underbrace{(\ell^{-\delta} +\delta^{-4})\delta^{-1}(\epss + C_1\eps)^{3/2} r^{-1+C(\epss+C_1\eps)^{1/2}}}_{ 
\text{Case 4 with }\mu = \delta, \nu=C(\epss + C_1\eps)^{1/2},  C(\epss + C_1\eps)^{1/2}\leq \delta/2}.
\endaligned
$$
Provided $(\ell^{-\delta} +\delta^{-4})\delta^{-1}(\epss + C_1\eps)^{1/2} \lesssim 1$, we arrive at \eqref{eq1-29-03-2021-new}. 


\paragraph{Proof of \eqref{eq3-29-03-2021-new}.}

We rely on the decomposition 
$h_{\alpha\beta} = h^{\star}_{\alpha\beta} + u_{\alpha\beta} = h^{\star}_{\alpha\beta} + u_{\init,\alpha\beta} + u_{\source, \alpha\beta}$
and, by recalling \eqref{equa-31-12-20-new}  and \eqref{eq3-09-05-2021-new} together with \eqref{eq1-29-03-2021-new}, we can control $|h_{\alpha\beta}|$. For $|h^{\alpha\beta}|$, we only need to observe that $\max_{\alpha,\beta}|h^{\alpha\beta}|\lesssim \max_{\alpha,\beta}|h_{\alpha\beta}|$, provided that $\max_{\alpha,\beta}|h_{\alpha\beta}|$ is sufficiently small. 


\section{Closing the bootstrap estimates}

\subsection{Improved energy estimate for general metric components} 

\paragraph{Objective.}

This section is devoted to the following result. The proof begins with the energy estimate~\eqref{prop energy-ici-exterior-new-equation} applied to \eqref{eq11-15-05-2020-ab-new}. Sufficient integrable $L^2$ decay must be checked for the terms arising in the right-hand side of the energy estimate. 

\begin{proposition}[Improved energy estimates for the matter field]
\label{proposition-section17-matter-new} 
Under the bootstrap assumptions,   
the matter field satisfies 
\begin{equation}
\Fenergy_{\kappa}^{\ME,N}(s,u) 
\leq {C_1 \over 2} \, \eps \, s^\delta, 
\qquad s \in [s_0, s_1]. 
\end{equation} 
\end{proposition}


\paragraph{Nonlinearities.}

\begin{proposition}[Sharp energy estimates for nonlinearities]
\label{prop1-10-07-2022}
Under the bootstrap assumptions  
and as a consequence of the pointwise metric estimates \eqref{eq7-26-03-2021-new}, for all $s \in [s_0, s_1]$ one has 
\begin{equation}\label{eq2-21-03-2021-new}
\aligned
 \|\crochet^{\kappa}J\zeta^{-1} |\Boxt_g u |_{p,k}\|_{L^2(\MME_s)}
& \lesssim   (\ell^{-\delta/2}+\delta^{-2})(\epss + C_1\eps) \big(s^{-1}\, \Fenergy_{\kappa}^{\ME,p,k}(s,u) + s^{-1+C(\epss + C_1\eps)}\Fenergy_{\kappa}^{\ME,p,k-1}(s,u)\big) 
\\
& \quad + (\epss + C_1\eps)^2 s^{-1-\delta} + R_\err^\star(s),
\endaligned
\end{equation}
 where $R_\err^\star(s)$ denotes the Ricci upper bound in~\eqref{eq1-21-05-2021-new}. 
\end{proposition}

\begin{proof} {\it Null semi-linear terms near the light cone.}
Recalling the structure of the null terms, we find  
\begin{equation}\label{eq1-10-07-2022}
|\Qbb^{\star}[u]|_{p,k} := \max_{\alpha,\beta}|\Qbb_{\alpha\beta}^\star[u] |_{p, k}
\lesssim  \sum_{p_1+p_2 = p\atop k_1+k_2=k} |\del u|_{p_1, k_1} |\delsN u|_{p_2, k_2} 
+ | h^\star |_p\sum_{p_1+p_2=p\atop k_1+k_2=k} |\del u|_{p_1, k_1} |\del u|_{p_2, k_2} 
\end{equation}
\vskip-.3cm 
and, in combination with \eqref{eq7-26-03-2021-new}, 
\begin{equation}\label{equation-8-sept-2021-new}
\aligned
|\Qbb^\star [u] |_{p,k}
 & \lesssim 
\sum_{p_1+p_2 = p\atop k_1+k_2=k} 
|\del u|_{p_1,k_1} |\delsN u|_{p_2, k_2} + \epss 
r^{-1} |\del u|_{p,k}|\del u|_{p_1, k_1}  
\\
& \lesssim  (\ell^{-\delta/2}+\delta^{-2})(\epss + C_1\eps) \big(\la r\ra^{-1} 
|\delsN u|_{p,k} 
+ \la r\ra^{-1 + C(\epss + C_1\eps)} |\delsN u|_{p-1,k-1} \big)
\\
& \quad + (\epss + C_1\eps) \, s^{\delta} r^{-1-\kappa}|\del u|_p
=: G_1 + G_2.
\endaligned 
\end{equation}
We integrate the above inequality so that, thanks to \eqref{lem1-22-05-2020-new},  
\begin{equation}
\| \crochet^{\kappa}J\zeta^{-1}\,  G_2 \|_{L^2(\Mnear_{\ell,s})}
\lesssim (\epss + C_1\eps) s^{-1-2\kappa+\delta}\|\crochet^{\kappa}\zeta|\del u|_p\|_{L^2(\Mnear_{\ell,s})} \lesssim (\epss + C_1\eps)^2 s^{-1-\delta}.
\end{equation}
On the other hand, we have 
$$
\aligned
\| \crochet^{\kappa}J\zeta^{-1}\, G_1 \|_{L^2(\Mnear_{\ell,s})}
& \lesssim (\ell^{-\delta/2}+\delta^{-2})(\epss + C_1\eps)s^{-1} \|\crochet^{\kappa}\zeta|\delsN u|_{p,k}\|_{L^2(\Mnear_{\ell,s})}
\\
& \quad +(\ell^{-\delta/2}+\delta^{-2})(\epss + C_1\eps)s^{-1+C(\epss + C_1\eps)^{1/2}} \|\crochet^{\kappa}\zeta|\delsN u|_{p-1,k-1}\|_{L^2(\Mnear_{\ell,s})}
\\
& \lesssim  (\ell^{-\delta/2}+\delta^{-2})(\epss + C_1\eps)
\big(s^{-1}\, \Fenergy_{\kappa}^{\ME,p,k}(s,u) + s^{-1+C(\epss + C_1\eps)}\Fenergy_{\kappa}^{\ME,p-1,k-1}(s,u)\big).
\endaligned	
$$

\

\noindent{\it Quasi-null terms near the light cone.}
Recalling of \cite[Lemma~\ref{lem1-31-01-2021}]{PLF-YM-main}, we have 
$$
\aligned
& |\Pbb^{\star}[u]|_{p,k} \lesssim  |\Pbb_{00}^{\star}[u]|_{p,k} + |\slashed{\Pbb}^{\star}[u]|_{p,k}
\\
& \lesssim   \underbrace{\sum_{p_1+p_2=p\atop k_1+k_2=k}|\del \us^{\N}|_{p_1,k_1}|\del \us^{\N}|_{p_2,k_2}}_{G_3}
+ \underbrace{\sum_{p_1+p_2=p}\Big(|\delts u|_{p_1} |\del u|_{p_2}
+ | \SbbME_{p_1}[u] | |\del u|_{p_2}  \Big) 
+ \sum_{p_1+p_2+p_3=p} |h^{\star}|_{p_3} |\del u|_{p_1} |\del u|_{p_2}}_{G_4}
\endaligned
$$
The terms $G_4$ are null terms and high-order terms, bounded similarly as $G_2$, so we can focus our attention on $G_3$. By applying \eqref{eq4-06-07-2022}, we find
$$
\|s \, \crochet^{\kappa}\zeta \, G_3\|_{L^2(\Mnear_{\ell,s})}
\lesssim (\ell^{-\delta/2}+\delta^{-2})(\epss + C_1\eps) 
\big(s^{-1}\, \Fenergy_{\kappa}^{\ME,p,k}(s,u) + s^{-1+C(\epss + C_1\eps)}\Fenergy_{\kappa}^{\ME,p-1,k-1}(s,u)\big).
$$

\vskip.3cm

\noindent{\it Quadratic semi-linear terms away from the light cone.}
It remains to derive the desired bound in $\Mfar_{\ell,s}$. We observe that $\crochet^{-1} \lesssim \ell^{-1}s^{-2}$ so, thanks to \eqref{eq1-10-07-2022} and our bound on $|\Pbb^{\star}[u]|_{p,k}$
$$
\aligned
& 
|\Qbb^{\star}[u]|_{p,k} + |\Pbb^{\star}[u]|_{p,k}
 \lesssim   (1+|h^{\star}|_N)|\del u\del u|_N
\lesssim \sum_{p_1+p_2=p\atop k_1+k_2=k}|\del u|_{p_1,k_1}|\del u|_{p_2,k_2}
\\
& \lesssim  |\del u|_{p,k}|\del u|_{[N/2]} \lesssim (\epss + C_1\eps) r^{-1}\crochet^{-\kappa}s^{\delta}|\del u|_{p,k}
\lesssim  \ell^{-\delta/2}(\epss + C_1\eps) r^{-1}\crochet^{-\kappa+\delta/2}|\del u|_{p,k}.
\endaligned
$$
We thus find 
$$
\aligned 
\|\crochet^{\kappa}J\zeta^{-1}\, G_3 \|_{L^2(\Mfar_{\ell,s})}
& \lesssim  \ell^{-\delta/2}(\epss + C_1\eps)\|r^{-1}s \, \crochet^{\kappa}\zeta|\del u|_{p,k}\|_{L^2(\Mfar_{\ell,s})}
\lesssim   \ell^{-\delta/2}(\epss + C_1\eps) s^{-1} \, \Fenergy_{\kappa}^{\ME,p,k}(s,u).
\endaligned \qedhere
$$
\end{proof}


\paragraph{Commutators.}

\begin{proposition}[Sharp energy estimates for commutators]
\label{prop--eq1-27-03-2021}
As a consequence of Lemma~\ref{coro-5-aout-2022}, 
for all  $\ord(Z) = p$ and $\rank(Z) = k$ one has 
$$ 
\aligned
\|\crochet^{\kappa}J\zeta^{-1}[Z,h^{\alpha\beta}\del_{\alpha}\del_{\beta}]u\|_{L^2(\MME_s)}
& \lesssim
\delta^{-1}(\ell^{-1}+\delta^{-2})(\epss + C_1\eps)s^{-1} \, \Fenergy_{\kappa}^{\ME,p,k}(s,u) 
\\
&\quad + \delta^{-1}(\ell^{-1}+\delta^{-2})(\epss + C_1\eps)s^{-1+C(\epss+C_1\eps)} \, 
\Fenergy_{\kappa}^{\ME,p-1,k-1}(s,u) 
\\
& \quad + \|\crochet^{\kappa}J\zeta^{-1}|\Boxt_g u|_{p-1,k-1}\|_{L^2(\MME_s)}
+ \ell^{-1}\delta^{-1}(\epss + C_1\eps)^2s^{-2}.
\endaligned
$$
\end{proposition} 

\begin{proof}
\noindent{\it Estimates away from the light cone.} For simplicity in the notation, we define 
$$
\aligned
\Abf_{p,k}(s) & :=  \sum_{\ord(Z)\leq p\atop \rank(z)\leq k}
\|\crochet^{\kappa}J\zeta^{-1}[Z,h^{\alpha\beta}\del_{\alpha}\del_{\beta}]u\|_{L^2(\Mfar_{\ell,s})},
\quad
\Bbf_{p,k}(s)  :=  \|\crochet^{\kappa}J\zeta^{-1}\,|\del\del u|_{p,k}\|_{L^2(\Mfar_{\ell,s})}.
\endaligned
$$ 
We recall Proposition~\ref{propo2-22-05-2020-new}, which gives us 
\begin{equation}\label{eq1-14-03-2021-new}
\aligned
\Bbf_{p,k}(s)& \lesssim  
\Abf_{p,k}(s)
+ \ell^{-1}\|\crochet^{\kappa}J\zeta^{-1}\,|\Boxt_g u|_{p,k}\|_{L^2(\Mfar_{\ell,s})} + \ell^{-1}s^{-1} \, \Fenergy_{\kappa}^{\ME,p+1,k+1}(s,u). 
\endaligned
\end{equation}
We then rely on \eqref{eq5-12-02-2020-new}, and control each term in this inequality as follows. 

\vskip.3cm

$\bullet$ For $ |\LOmega  h |_{p_1-1}|\del\del u|_{p_2,k_2}$ we find 
$
|\LOmega h|_{p} \lesssim |\LOmega h^{\star}|_{p} + |\LOmega u|_p\lesssim \epss \la r\ra^{-1} + |\LOmega u|_p, 
$
which leads us to 
\begin{equation}\label{eq6-10-07-2022}
\|\crochet^{\kappa}J\zeta^{-1}\,|\LOmega  h |_{p_1-1}|\del\del u|_{p_2,k_2} \|_{L^2(\Mfar_{\ell,s})}
\lesssim
\begin{cases}
\delta^{-1} (\epss + C_1\eps)\Bbf_{p-1,k-1}(s),
& p_1\leq N-3,
\\
\epss \Bbf_{p-1,k-1}(s)  + \ell^{-1}\delta^{-1}(\epss + C_1\eps)^2s^{-2},& p_1\geq N-2,
\end{cases}
\end{equation}
where for the first case we have applied the third bound of \eqref{eq1-30-05-2020-new}, while for the second bound we have applied \eqref{eq1-12-05-2020-new} and \eqref{eq1-15-08-2021-new}, provided $\kappa\geq 1/2+(3/2)\delta$.

\vskip.3cm

$\bullet$ The term $|H| \, | \del\del u|_{p-1,k-1} $ is also bounded by the first case of \eqref{eq6-10-07-2022}.

\vskip.3cm

$\bullet$ The term $|\del  h |_{p_1-1,k_1}|\del\del u|_{p_2,k_2}$ is easier since a partial derivative is acting on $h$. We observe that
\begin{equation}\label{eq5-10-07-2022}
|\del h|_{p} \lesssim |\del h^{\star}|_p + |\del u|_p \lesssim \epss \la r\ra^{-2} + |\del u|_p
\end{equation} 
and, using the inequality $\crochet^{-1} \lesssim \ell^{-1}\la r\ra^{-1}$ valid in $\Mfar_{\ell,[s_0,s_1]}$, 
\begin{equation}\label{eq7-10-07-2022}
\|\crochet^{\kappa}J\zeta^{-1}\,|\del  h |_{p_1-1,k_1}|\del\del u|_{p_2,k_2}\|_{L^2(\Mfar_{\ell,s})}
\lesssim \ell^{-\delta/2}(\epss + C_1\eps)s^{-1}\, \Fenergy_{\kappa}^{\ME,p,k}(s,u). 
\end{equation}

We then deduce that 
\begin{equation}\label{eq8-10-07-2022}
\Abf_{p,k}(s) \lesssim \delta^{-1}(\epss + C_1\eps)\Bbf_{p-1,k-1}(s)
+ \ell^{-\delta/2}(\epss + C_1\eps)s^{-1}\, \Fenergy_{\kappa}^{\ME,p,k}(s,u)
+ \ell^{-1}\delta^{-1}(\epss + C_1\eps)^2s^{-2}.
\end{equation}
Combining \eqref{eq8-10-07-2022} and \eqref{eq1-14-03-2021-new} and observing that $\delta^{-1}(\epss+C_1\eps)\ll 1$, we obtain 
\begin{equation}\label{eq12-10-07-2022}
\aligned
&
\sum_{\ord(Z)\leq p\atop \rank(z)\leq k}
\|\crochet^{\kappa}J\zeta^{-1}[Z,h^{\alpha\beta}\del_{\alpha}\del_{\beta}]u\|_{L^2(\Mfar_{\ell,s})}
\\
& \lesssim 
\|\crochet^{\kappa}J\zeta^{-1}\,|\Boxt_g u|_{p-1,k-1}\|_{L^2(\Mfar_{\ell,s})}
 + \ell^{-1}\delta^{-1}(\epss + C_1\eps)s^{-1} \Fenergy_{\kappa}^{\ME,p,k}(s,u)
+ \ell^{-1}\delta^{-1}(\epss + C_1\eps)^2s^{-2},
\endaligned
\end{equation}
which completes the proof of Proposition \ref{prop--eq1-27-03-2021} in $\Mfar_{\ell,[s_0,s_1]}$.

\vskip.3cm

\noindent{\it Estimates near the light cone.} We also introduce the notation
$$
\aligned
\Abf_{p,k}(s) & :=  \sum_{\ord(Z)\leq p\atop \rank(z)\leq k}
\|\crochet^{\kappa}J\zeta^{-1}[Z,h^{\alpha\beta}\del_{\alpha}\del_{\beta}]u\|_{L^2(\Mnear_{\ell,s})},
\qquad
\Bbf_{p,k}(s) :=  \|t^{-1}\crochet^{1+\kappa}J\zeta^{-1}\,|\del\del u|_{p,k}\|_{L^2(\Mnear_{\ell,s})}.
\endaligned
$$ 
Then, in view of \eqref{eq3-28-12-2020-new} we get 
\begin{equation}\label{eq11-10-07-2022}
\Bbf_{p,k}(s) \lesssim \Abf_{p,k}(s) +  \|\crochet^{\kappa}J\zeta^{-1}|\Boxt_g u|_{p,k}\|_{L^2(\Mnear_{\ell,s})}
+ s^{-1}\Fenergy_{\kappa}^{\ME,p+,k+1}(s,u).
\end{equation}
On the other hand, we recall \eqref{eq4-12-02-2020-second} and we need to control each term in its right-hand side. 

\vskip.3cm

$\bullet$ The term $\big(|\HN^{00}| + t^{-1}|r-t| |H|\big) \, |\del\del u|_{p-1,k-1}$, thanks to \eqref{eq5-10-06-2022} and \eqref{eq1-30-05-2020-new} (the third bound), is directly controlled by 
\begin{equation}\label{eq9-10-07-2022}
\big\|\crochet^{\kappa}J\zeta^{-1} (|\HN^{00}| + t^{-1}|r-t| |H|\big) \, |\del\del u|_{p-1,k-1})\big\|\lesssim (\epss+C_1\eps)\Bbf_{p-1,k-1}(s).
\end{equation}

\vskip.3cm

$\bullet$ Thanks to \eqref{eq5-10-06-2022} and \eqref{eq1-30-05-2020-new} (i.e. the third bound therein), the term $T_{p,k}^\textbf{hier}[H,u]$ is controlled as follows. First of all, we have 
$$
|LH^{\N00}|_{N-5} + t^{-1}|r-t||H|_{N-5} \lesssim (\epss+C_1\eps)\crochet r^{-1}
$$
and 
$$
|LH^{\N00}|_{k-1} + t^{-1}|r-t||H|_{k-1} \lesssim |h^{\star}|_k + |L u|_{k-1} + t^{-1}\crochet|u|_k\lesssim \epss \la r\ra^{-1} + |\LOmega u|_{k-1} + t^{-1}\crochet |u|_k.
$$
We thus deduce that  
$$
\aligned
&\|\crochet^{\kappa}J\zeta^{-1}T_{p,k}^\textbf{hier}[H,u] \|_{L^2(\Mnear_{\ell,s})}
\\
& \lesssim  (\epss + C_1\eps)\|\crochet^{\kappa}J\zeta^{-1} \crochet r^{-1}|\del\del u|_{p-1,k-1}\|_{L^2(\Mnear_{\ell,s})}
+ \sum_{k_1=N-4}^{k}\||\crochet^{\kappa}J\zeta^{-1}|\del\del u|_{p-k_1,k-k_1}|\LOmega u|_{k_1-1}\|_{L^2(\Mnear_{\ell,s})} 
\\
& \quad + (\epss+C_1\eps)s^{2\delta}\|\crochet^{\kappa}J\zeta^{-1}\, r^{-2}\crochet^{-\kappa}|u|_N\|_{L^2(\Mnear_{\ell,s})}, 
\endaligned
$$
where \eqref{eq1-17-07-2020-second} and \eqref{eq2-04-06-2022} were used. For the latter term, we apply \eqref{eq1-12-05-2020-new} which is seen to be bounded by $\delta^{-1}(\epss + C_1\eps)s^{-2}$ (provided $\kappa \geq 1/2+(3/2)\delta$). For the second term we recall \eqref{eq1-18-08-2021-lambda1}
together with \eqref{eq7-26-03-2021-new} and obtain
$$
\aligned
&\||\crochet^{\kappa}J\zeta^{-1}|\del\del u|_{p-k_1,k-k_1}|\LOmega u|_{k_1-1}|\|_{L^2(\Mnear_{\ell,s})} 
\\
& \lesssim 
\begin{cases}
(\ell^{-\delta} + \delta^{-2})(\epss + C_1\eps) s^{-1}	\|\zeta\crochet^{-1+\kappa}|\LOmega u|_{k-1}\|_{L^2(\Mnear_{\ell,s})} \quad & \text{ (when $k_1=k$)}
\\
(\ell^{-\delta} + \delta^{-2})(\epss + C_1\eps) s^{-1+C(\epss+C_1\eps)}	\|\zeta\crochet^{-1+\kappa}|\LOmega u|_{k-2}\|_{L^2(\Mnear_{\ell,s})}  \quad & \text{ (when $k_1\leq k-1$)} 
\end{cases}
\\
& \lesssim  (\ell^{-\delta} + \delta^{-2})(\epss + C_1\eps) s^{-1}\Fenergy_{\kappa}^{\ME,p,k}(s,u)
+  (\ell^{-\delta} + \delta^{-2})(\epss + C_1\eps)s^{-1+C(\epss+C_1\eps)}\Fenergy_{\kappa}^{\ME,p-1,k-1}(s,u).
\endaligned
$$
Hence, it follows that 
\begin{equation}\label{eq10-10-07-2022}
\aligned
& 
\|\crochet^{\kappa}J\zeta^{-1}T_{p,k}^\textbf{hier}[H,u] \|_{L^2(\Mnear_{\ell,s})}
\\
& \lesssim 
(\epss+C_1\eps)\Bbf_{p-1,k-1}(s)
+ (\ell^{-\delta} + \delta^{-2})(\epss + C_1\eps)s^{-1}\Fenergy_{\kappa}^{\ME,p,k}(s,u)
\\
& \quad + (\ell^{-\delta} + \delta^{-2})(\epss + C_1\eps)s^{-1+C(\epss+C_1\eps)}\Fenergy_{\kappa}^{\ME,p-1,k-1}(s,u) + \delta^{-1}(\epss + C_1\eps)s^{-2}.
\endaligned
\end{equation}

\vskip.3cm

$\bullet$ We consider next the terms $T^{\bf easy}$ and $T^{\bf super}$, and we observe that $T^{\super}$ contains a favorable factor $t^{-1}$ while the second term of $T^{\easy}$ contains the factor $t^{-1}|r-t|$, supplying sufficient decay in $\la r\ra^{-1}$. Thus these terms are relatively easier to handle and we only discuss now the most challenging term, namely 
$|h^{\N00}|_{p_1-1,k_1}|\del\del u|_{p_2,k_2}$. When $p_1-1\leq N-3$, we apply \eqref{eq1-17-07-2020-second} and obtain
$$
\aligned
\|\crochet^{\kappa}J\zeta^{-1}|\del h^{\N00}|_{p_1-1,k_1}|\del\del u|_{p_2,k_2}\|_{L^2(\Mnear_{\ell,s})}
& \lesssim
\delta^{-1} (\epss + C_1\eps)s^{-1-2\kappa+\delta}\|\crochet^{\kappa}\zeta|\del\del u|_{p,k}\|_{L^2(\Mnear_{\ell,s})}
\\
& \lesssim \delta^{-1}(\epss + C_1\eps)^2s^{-2}.
\endaligned
$$
When $p_1-1\geq N-2$, we have $p_2\leq 1\leq N-5$ and thus
$$
\aligned
& \|\crochet^{\kappa}J\zeta^{-1}|\del h^{\N00}|_{p_1-1,k_1}|\del\del u|_{p_2,k_2}\|_{L^2(\Mnear_{\ell,s})}
\\
& \lesssim
s\|\crochet^{\kappa}\zeta|\del h^{\star}|_{p_1-1,k_1}|\del\del u|_{p_2,k_2}\|_{L^2(\Mnear_{\ell,s})} 
+ s\|\crochet^{\kappa}\zeta|\del u|_{p_1-1,k_1}|\del\del u|_{p_2,k_2}\|_{L^2(\Mnear_{\ell,s})}
\\
&\lesssim (\ell^{-\delta} + \delta^{-2})(\epss + C_1\eps)s^{-1}\Fenergy_{\kappa}^{\ME,p,k}(s,u)
+(\ell^{-\delta} + \delta^{-2})(\epss + C_1\eps)s^{-1+C(\epss+C_1\eps)}\Fenergy_{\kappa}^{\ME,p-1,k-1}(s,u).
\endaligned
$$
The weighted $L^2$ norm of the remaining terms is bounded by $\delta^{-1}(\epss + C_1\eps)^2s^{-2}$. 

Now we conclude with \eqref{eq9-10-07-2022}, \eqref{eq10-10-07-2022}, and the bounds established on $T^{\bf easy}$ and $T^{\bf super}$:
$$
\aligned
\Abf_{p,k}& \lesssim  (\epss+C_1\eps)\Bbf_{p-1,k-1}(s)
+ (\ell^{-\delta} + \delta^{-2})(\epss + C_1\eps)s^{-1}\Fenergy_{\kappa}^{\ME,p,k}(s,u)
\\
& \quad +(\ell^{-\delta} + \delta^{-2})(\epss + C_1\eps)s^{-1+C(\epss+C_1\eps)}\Fenergy_{\kappa}^{\ME,p-1,k-1}(s,u)
+ \delta^{-1}(\epss + C_1\eps)^2s^{-2}.
\endaligned
$$
Substituting \eqref{eq11-10-07-2022} in the above estimate and recalling $(\epss+C_1\eps)\ll1$, we obtain
\begin{equation}\label{eq13-10-07-2022}
\aligned
&\sum_{\ord(Z)\leq p\atop \rank(z)\leq k}
\|\crochet^{\kappa}J\zeta^{-1}[Z,h^{\alpha\beta}\del_{\alpha}\del_{\beta}]u\|_{L^2(\Mnear_{\ell,s})}
\\
& \lesssim   \|\crochet^{\kappa}J\zeta^{-1}|\Boxt_g u|_{p,k}\|_{L^2(\Mnear_{\ell,s})} 
+ (\ell^{-\delta} + \delta^{-2})(\epss + C_1\eps)s^{-1}\Fenergy_{\kappa}^{\ME,p,k}(s,u)
\\
& \quad +(\ell^{-\delta} + \delta^{-2})(\epss + C_1\eps)s^{-1+C(\epss+C_1\eps)}\Fenergy_{\kappa}^{\ME,p-1,k-1}(s,u)
+ \delta^{-1}(\epss + C_1\eps)^2s^{-2}.
\endaligned
\end{equation} 
In conclusion, by combining \eqref{eq12-10-07-2022} and \eqref{eq13-10-07-2022} we arrive at the desired result.
\end{proof}


\paragraph{Conclusion.}

We summarize our result so far as follows. 

\begin{proposition}
Under the conclusions of Propositions \ref{prop1-10-07-2022} and \ref{prop--eq1-27-03-2021}, the following estimate holds
for all $\ord(Z) \leq p$ and $\rank(Z) \leq k$:
\begin{equation}\label{eq2-27-03-2021-new}
\aligned 
& \|\crochet^{\kappa}J\zeta^{-1}\,\Boxt_g Z u_{\alpha\beta}\|_{L^2(\MME_s)}
\\
 & \lesssim   \delta^{-1}(\ell^{-1} + \delta^{-2})(\epss + C_1\eps) \big(s^{-1} \, \Fenergy_{\kappa}^{\ME,p,k}(s,u)
 + s^{-1+C(\epss+C_1\eps)}\, \Fenergy_{\kappa}^{\ME,p-1,k-1}(s,u)\big)
\\
& \quad + \ell^{-1}\delta^{-1}(\epss + C_1\eps)^2s^{-1-\delta}
+ R_\err^\star(s). 
\endaligned
\end{equation}
\end{proposition}

At this juncture it is useful to state the following elementary property.  

\begin{lemma}[Comparison of energy functionals]
\label{lem1-28-03-2021} 
Let  $(\eta,c^*,w) = (\kappa,0,u)$ or $(\eta,c^*,w) = (\mu,c,\phi)$ one has 
$$
\Eenergy_{g, \eta,c^*}^{\ME}(s,u)\geq (1/4) \, \Eenergy_{\eta,c^*}^{\ME}(s,u). 
$$
\end{lemma}

\begin{proof} All of the assumptions required in \cite[Section~\ref{section-label-11-1}]{PLF-YM-main} are valid in the present context, especially the inequality $h^{\N00}\leq 0$ in $\Mnear_{\ell,[s_0,s_1]}$ and the decay 
$
|h|\lesssim (\epss+C_1\eps)r^{-\kappa}s^{\delta}\leq (\epss+C_1\eps)r^{-1/2}, 
$
and we can follow the same arguments as in \cite{PLF-YM-main}.
\end{proof}

In order to apply \eqref{prop energy-ici-exterior-new-equation}, we still need the following two bounds for ($\ord(Z)\leq p$ and $\rank(Z)\leq k$ and)
$(\eta,w) = (\kappa,u) $ or $(\mu,\phi)$:
\begin{equation}\label{eq9-27-03-2021-new}
\int_{\MME_s} |J G_{g, \eta}[Zw]| dx\lesssim (\ell^{-\delta/2} + \delta^{-2})(\epss + C_1\eps)s^{-1} \Eenergy_{\eta}^{\ME}(s,Zw),
\end{equation}
\begin{equation}\label{eq2-11-07-2022}
\int_{\MME_s}\eta \crochet^{2\eta-1} \aleph'(r-t)  h^{\N 00} |\del_t Zw|^2 \ Jdx
\lesssim (\epss + C_1\eps)s^{-1} \Eenergy_{\eta}^{\ME}(s,Zw).
\end{equation}
The first bound \eqref{eq9-27-03-2021-new} is directly derived by applying \eqref{eq7-26-03-2021-new} (with $k=0$) and \eqref{eq2-11-07-2022} is nothing but \eqref{eq5-10-06-2022}.
Finally we are ready to apply \eqref{prop energy-ici-exterior-new-equation} with $(\eta,w,c^*) = (\kappa,u,0)$. Taking \eqref{eq2-27-03-2021-new}, \eqref{eq9-27-03-2021-new}, \eqref{eq2-11-07-2022} and \eqref{eq1-21-05-2021-new} into account, we obtain
\begin{equation}\label{eq3-27-03-2021-new}
\aligned
&\frac{d}{ds}\Eenergy_{g,\kappa}^{\ME}(s,Z u)  + \frac{d}{ds} \Eenergy^{\Lcal}_{g, c}(s,Z u;s_0)  
+  2 \kappa \int_{\MME_s}\big( g^{\N ab} \delsN_aZu \delsN_bZu +  c^2|Zu|^2\big) \crochet^{2\kappa-1} \aleph'(r-t) \ Jdx
\\
& \lesssim \big(\ell^{-\delta} + \delta^{-2}\big)(\epss + C_1\eps)s^{-1}\Eenergy_{\kappa}^{\ME}(s,Zu)ds
+ \big\|\zeta\crochet^{\kappa}\del_t Zu\|_{L^2(\MME_s)}\|\crochet^{\kappa}J\zeta^{-1} \Boxt_gZu \big\|_{L^2(\MME_s)}
\\
& \lesssim  \delta^{-1}(\ell^{-1} + \delta^{-2})(\epss + C_1\eps) \big(s^{-1} \, \Eenergy_{\kappa}^{\ME,p,k}(s,u)
+ s^{-1+C(\epss+C_1\eps)}\, \Fenergy_{\kappa}^{\ME,p-1,k-1}(s,u)\Fenergy_{\kappa}^{\ME,p,k}(s,u)\big)
\\
& \quad + \ell^{-1}\delta^{-1}(\epss + C_1\eps)^2(1 + C_{R^{\star}})s^{-1-\delta}\Fenergy_{\kappa}^{\ME,p,k}(s,u). 
\endaligned
\end{equation}
Now let us focus our attention on the left-hand side. 
In the third term we observe that $\aleph'(r-t)\geq 0$ in $\MME_s$. Furthermore, since $|h|\ll 1$ we have $\sum_{a}|\delsN_a u|^2\lesssim g^{\N ab}\delsN_a u\delsN_b u$ and this leads us to
\begin{equation}\label{eq4-27-03-2021}
\sum_{\ord(Z)\leq p\atop \rank(Z)\leq k}\int_{\MME_s}\crochet^{2\kappa -1} 
\aleph'(r-t) g^{\N ab}\delsN_a Zu\delsN_b Zu \, Jdx\geq 0.
\end{equation}
\vskip-.25cm 
For the second term in the left-hand side, we recall \eqref{eq8-27-03-2021-new} and \eqref{eq3-27-05-2020-new} and obtain
\begin{equation}\label{eq3-11-07-2022}
 \frac{d}{ds} \Eenergy^{\Lcal}_{g, c}(s,Z u;s_0)  \geq 0.
\end{equation}
Summing up, in view of \eqref{eq3-27-03-2021-new}, for all $\ord(Z)\leq p,\rank(Z)\leq k$ we find 
$$
\aligned
\frac{d}{ds}\Fenergy_{g,\kappa}^{\ME,p,k}(s,u)& \lesssim   \delta^{-1}(\ell^{-1} + \delta^{-2})(\epss + C_1\eps) \big(s^{-1} \, \Fenergy_{\kappa}^{\ME,p,k}(s,u)
+ s^{-1+C(\epss+C_1\eps)}\, \Fenergy_{\kappa}^{\ME,p-1,k-1}(s,u)\big)
\\
& \quad + \ell^{-1}\delta^{-1}(\epss + C_1\eps)^2(1 + C_{R^{\star}})s^{-1-\delta}
\\
& \lesssim   \delta^{-1}(\ell^{-1} + \delta^{-2})(\epss + C_1\eps) \Big(s^{-1} \, \Fenergy_{g,\kappa}^{\ME,p,k}(s,u)
+ \underline{s^{-1+C(\epss+C_1\eps)}\, \Fenergy_{g,\kappa}^{\ME,p-1,k-1}(s,u)}
\, 
\Big)
\\
& \quad + \ell^{-1}\delta^{-1}(\epss + C_1\eps)^2(1 + C_{R^{\star}})s^{-1-\delta}.
\endaligned
$$
Observe that when $k=0$ the {\sl underlined} term in the right-hand side {\sl does not exist.} We rewrite the above inequality in the form
$$
\aligned
\frac{d}{ds}\Fenergy_{g,\kappa}^{\ME,p,k}(s,u)
& \leq  C\delta^{-1}(\ell^{-1} + \delta^{-2})(\epss + C_1\eps) \big(s^{-1} \, \Fenergy_{g,\kappa}^{\ME,p,k}(s,u)
+ s^{-1+C(\epss+C_1\eps)}\, \Fenergy_{g,\kappa}^{\ME,p-1,k-1}(s,u)\big)
\\
& \quad + C\ell^{-1}\delta^{-1}(\epss + C_1\eps)^2(1 + C_{R^{\star}})s^{-1-\delta}
\endaligned
$$
with $C$ a constant determined by $N$.  Then by Gronwall's inequality and by an induction argument on $k$, we deduce that  
$$
\Fenergy_{g,\kappa}^{\ME,p,k}(s,u) - \Fenergy_{g,\kappa}^{\ME,p,k}(s_0,u)s^{C(\epss + C_1\eps)^{1/3}}\leq C(\epss+C_1\eps)^{4/3}s^{C(\epss+C_1\eps)^{1/3}}, 
$$
provided 
$$
C_0<C_1,\quad\ell^{-1}\delta^{-2}(1+ C_{R^{\star}}) (\epss+C_1\eps)^{1/3}\leq 1, \quad\delta^{-1}(\ell^{-1} + \delta^{-2})(\epss+C_1\eps)^{1/3}\leq 1.
$$
Now we take $C(\epss+C_1\eps)^{1/3}\leq \delta$ and $C_1>4C_0$ together with 
$\epss + C_1\eps\leq \big( C_1-4C_0\big)^3 / (4CC_1)^3$
and we arrive at
$$
\Fenergy_{\kappa}^{\ME,N}(s,u) \leq 2 \, \Fenergy_{g, \kappa}^{\ME,p,k}(s,u)\leq \frac{1}{2}(\epss + C_1\eps) s^\delta.
$$


\subsection{Improved energy estimate for the matter field} 

In comparison with the issues arising with metrics with lower decay, the control of the matter field is significantly simpler in the harmonic regime for the metric, and we do not need to distinguish between various components of the metric. We are going to establish the following result on commutators valid in $\MME_{[s_0,s_1]}$: 
\begin{equation}\label{eq5-11-07-2022}
\aligned
&\|\crochet^{\mu}J\zeta^{-1}[Z,H^{\alpha\beta}\del_{\alpha}\del_{\beta}]\phi \|_{L^2(\MME_s)}
\\
&\lesssim 
\big(\ell^{-\delta/2} +\delta^{-2}\big)(\epss+C_1\eps)s^{-1}\big(\Fenergy_{\mu,c}^{\ME,p,k}(s,\phi) + s^{C(\epss+C_1\eps)}\Fenergy_{\mu,c}^{\ME,p,k-1}(s,\phi)\big) 
 + 
\begin{cases}
0, \, &k\leq N-5,
\\
\delta^{-1}(\epss+C_1\eps)^2, \ &k\geq N-4.
\end{cases}
\endaligned
\end{equation}
To this end, we rely on the general commutator inequality~\eqref{eq5-12-02-2020-new} (with $u = \phi$ therein). Recalling \eqref{eq7-26-03-2021-new} (case $k=0$) we find 
\begin{equation}
\aligned
\|\crochet^{\mu}J\zeta^{-1}|H||\del\del\phi|_{p-1,k-1} \|_{L^2(\MME_s)}
& \lesssim  \big(\ell^{-\delta/2} +\delta^{-2}\big) (\epss + C_1\eps)s\|\la r\ra^{-1}\crochet^{\mu}\zeta|\del\del\phi|_{p-1,k-1}\|_{L^2(\MME_s)}
\\
& \lesssim  \big(\ell^{-\delta/2} +\delta^{-2}\big) (\epss + C_1\eps)s^{-1}\Fenergy_{\mu,c}^{\ME,p,k-1}(s,\phi).
\endaligned
\end{equation} 
For the second term in the right-hand side of \eqref{eq5-12-02-2020-new}, we observe that, when $p_1\leq N-5$ and thanks to \eqref{eq3-29-03-2021-new},
\begin{equation}\label{eq1-12-07-2022}
\|\crochet^{\mu}J\zeta^{-1}|\LOmega H|_{p_1-1}|\del\del \phi|_{p_2,k_2}\|_{L^2(\MME_s)} \lesssim (\epss+C_1\eps)s^{-1+C(\epss+C_1\eps)^{1/2}}\Fenergy_{\mu,c}^{\ME,p,k-1}(s,\phi),
\end{equation}
where we used $p_1\geq 1$ and $p_1+p_2=p,k_1+k_2=k$. When $p_1\geq N-4$, we have 
$$
|\LOmega H|_{p-1} \lesssim |h^{\star}|_{p_1} + |u|_{p_1} \lesssim \epss \la r\ra^{-1} +  |u|_{p_1}.
$$
\vskip-.3cm
\noindent 
Then, we have  
$$
\aligned
&\|\crochet^{\mu}J\zeta^{-1}|\LOmega H|_{p_1-1}|\del\del \phi|_{p_2,k_2}\|_{L^2(\MME_s)} 
\\
& \lesssim  \epss \|\crochet^{\mu}J\zeta^{-1}\la r\ra^{-1}|\del\del\phi|_{p,k}\|_{L^2(\MME_s)}
 + \|\crochet^{\mu}J\zeta^{-1}|u|_{p_1}|\del\del\phi|_{p_2,k_2}\|_{L^2(\MME_s)}
 \\
& \lesssim  \epss s^{-1}\Fenergy_{\mu,c}^{\ME,p,k-1}(s,\phi) 
+ (\epss+C_1\eps)s^{2\delta}\|\crochet^{\mu}J\zeta^{-1}r^{-2}\crochet^{1-\mu}|u|_{p,k}\|_{L^2(\MME_s)}
\\
& \lesssim   \epss s^{-1}\Fenergy_{\mu,c}^{\ME,p,k-1}(s,\phi) 
+ (\epss+C_1\eps)s^{1+2\delta}\|\crochet^{2-\kappa}\zeta\,r^{-2}\crochet^{-1+\kappa}|u|_{p,k}\|_{L^2(\MME_s)}
\\
& \lesssim   \epss s^{-1}\Fenergy_{\mu,c}^{\ME,p,k-1}(s,\phi)
 + (\epss+C_1\eps)s^{1-2\kappa+2\delta}\|\crochet^{-1+\kappa}|u|_{p,k}\|_{L^2(\MME_s)}, 
\endaligned
$$
where we applied Lemma~\ref{lemma-111-new}. Then in view of \eqref{eq1-12-05-2020-new} and \eqref{eq1-12-07-2022}, we obtain
\begin{equation}\label{eq4-11-07-2022}
\aligned
&\|\crochet^{\mu}J\zeta^{-1}|\LOmega H|_{p_1-1}|\del\del \phi|_{p_2,k_2}\|_{L^2(\MME_s)}
\\
& \lesssim
\begin{cases}
\delta^{-1}(\epss+C_1\eps)^2,\quad &N-4\leq  p\leq N,
\\
\big(\ell^{-\delta/2} +\delta^{-2}\big) (\epss + C_1\eps)s^{-1+C(\epss+C_1\eps)}\Fenergy_{\mu,c}^{\ME,p,k-1}(s,\phi),\ & p\leq N-5.
\end{cases}
\endaligned
\end{equation}
Finally, for the terms $|\del H|_{p_1-1,k_1}|\del\del \phi|_{p_2,k_2}$, when $p_1\leq N-3$, we apply \eqref{eq7-26-03-2021-new} and obtain
$$
\aligned
&\big\|\crochet^{\mu}J\zeta^{-1}\, |\del H|_{p_1-1,k_1}|\del\del \phi|_{p_2,k_2}\big\|_{L^2(\MME_s)}
\\
& \lesssim  \big\|\crochet^{\mu}J\zeta^{-1}\, |\del H|_{N-4,0}|\del\del \phi|_{p,k}\big\|_{L^2(\MME_s)}
+ \big\|\crochet^{\mu}J\zeta^{-1}\, |\del H|_{p_1-1,k_1}|\del \phi|_{p,k-1}\big\|_{L^2(\MME_s)}
\\
& \lesssim  	\big(\ell^{-\delta/2} +\delta^{-2}\big)(\epss+C_1\eps)s^{-1}\Fenergy_{\mu,c}^{\ME,p,k}(s,\phi) 
+ 	
\big(\ell^{-\delta/2} +\delta^{-2}\big)(\epss+C_1\eps)s^{-1+C(\epss+C_1\eps)}\Fenergy_{\mu,c}^{\ME,p,k-1}(s,\phi), 
\endaligned
$$
which leads us to \eqref{eq5-11-07-2022}. 

Recalling \eqref{eq9-27-03-2021-new}, \eqref{eq2-11-07-2022} and \eqref{eq3-11-07-2022}, which are also valid for $\phi$, we apply \eqref{prop energy-ici-exterior-new-equation} and obtain
$$
\aligned
\frac{d}{ds}\Fenergy_{g,\mu,c}^{\ME,p,k}(s,u)
& \lesssim    (\ell^{-\delta/2} + \delta^{-2})(\epss + C_1\eps) 
\Big(s^{-1} \, \Fenergy_{g,\mu,c}^{\ME,p,k}(s,u)
+ s^{-1+C(\epss+C_1\eps)}\, \Fenergy_{g,\mu,c}^{\ME,p,k-1}(s,u)\Big)
\endaligned
$$
for $p\leq N-5$ and
$$
\aligned
\frac{d}{ds}\Fenergy_{g,\mu,c}^{\ME,p,k}(s,u)
& \lesssim   (\ell^{-\delta/2} + \delta^{-2})(\epss + C_1\eps) 
\Big(s^{-1} \, \Fenergy_{g,\mu,c}^{\ME,p,k}(s,u)
+ s^{-1+C(\epss+C_1\eps)}\, \Fenergy_{g,\mu,c}^{\ME,p,k-1}(s,u)\Big)
   +\delta^{-1}(\epss+C_1\eps)^2
\endaligned
$$
for $p\geq N-4$. Here $C$ is a constant determined by $N$. Then we take
$
C_0<C_1$ and
$(\ell^{-\delta/2} + \delta^{-2})(\epss+C_1\eps)^{1/3}\leq 1$
and, by Gronwall's inequality and an induction on $k$, we have 
\begin{equation}\label{eq2-12-07-2022}
\aligned
\Fenergy_{g,\mu,c}^{\ME,N-5}(s,\phi)& \leq  \Fenergy_{g,\mu,c}^{\ME,N-5}(s_0,\phi)s^{C(\epss+C_1\eps)^{1/3}}
 + C(\epss+C_1\eps)^{4/3} s^{C(\epss+C_1\eps)^{1/3}},
\\
\Fenergy_{g,\mu,c}^{\ME,N}(s,\phi)& \leq  \Fenergy_{g,\mu,c}^{\ME,N}(s_0,\phi)s^{1+C(\epss+C_1\eps)^{1/3}}
 + C(\epss+C_1\eps)^{4/3} s^{1+C(\epss+C_1\eps)^{1/3}},
\endaligned
\end{equation}
with (another) constant $C>0$ determined by $N$. We require that
$$
C(\epss+C_1\eps)^{1/3}\leq \delta, \quad C_1>4C_0, \quad \epss+C_1\eps < \Big(\frac{C_1-4C_0}{4C C_1}\Big)^3.
$$
We have established the refined energy bound for the matter field, and the proof of Theorem~\ref{theo-main-result-qualitative} is completed. 


\small

\paragraph{Acknowledgments.}
 
Part of this work was done when PLF was a visiting research fellow at the Courant Institute for Mathematical Sciences, New York University, and a visiting professor at the School of Mathematical Sciences, Fudan University, Shanghai. The work of YM was supported by a Special Financial Grant from the China Postdoctoral Science Foundation under the grant number 2017-T100732. 
 

\bibliography{references}

\addcontentsline{toc}{section}{\large References}

\pagestyle{plain}

\end{document}